\documentclass{article}

\usepackage{setspace}
\RequirePackage[OT1]{fontenc}
\usepackage{verbatim}
\usepackage{multirow}
\usepackage{makecell}
\usepackage{tabularx}
\usepackage{tabu}
\usepackage{float}
\usepackage{stfloats}
\usepackage{bm}
\usepackage{cancel}
\usepackage[normalem]{ulem}

\usepackage{multirow}
\usepackage{siunitx}

\newcommand{\norm}[1]{\left\lVert#1\right\rVert}

\newcommand{\Bias}{\mathbb{B}\mathrm{ias}}

\newcommand{\GMSE}{\mathrm{GMSE}}
\newcommand{\CV}{\mathrm{CV}}
\newcommand{\MSE}{\mathrm{MSE}}

\newcommand{\Exp}{\mathbb{E}}
\newcommand{\Var}{\mathbb{V}\mathrm{ar}}
\newcommand{\Cov}{\mathbb{C}\mathrm{ov}}

\newcommand{\M}{\bm{Y}}
\newcommand{\D}{\bm{\lambda}}

\RequirePackage{amsthm,amsmath,amsfonts,amssymb,mathtools}

\usepackage[utf8]{inputenc} 
\usepackage[T1]{fontenc}    
\usepackage{url}            
\usepackage{booktabs}       
\usepackage{nicefrac}       
\usepackage{microtype}      
\usepackage{natbib}
\usepackage{doi}
\usepackage{xcolor}

\usepackage{cancel}
\usepackage{graphicx}
\usepackage{enumitem}
\usepackage{tabularx}

\usepackage{comment}

\usepackage{algorithm}
\usepackage{algorithmicx}
\usepackage{algpseudocode}

\newtheorem{definition}{Definition}

\newtheorem{lemma}{Lemma}%
\newtheorem{theorem}{Theorem}

\newtheorem{innercustomgeneric}{\customgenericname}
\providecommand{\customgenericname}{}
\newcommand{\newcustomtheorem}[2]{%
  \newenvironment{#1}[1]
  {%
   \renewcommand\customgenericname{#2}%
   \renewcommand\theinnercustomgeneric{##1}%
   \innercustomgeneric
  }
  {\endinnercustomgeneric}
}

\newcustomtheorem{customthm}{Theorem}
\newcustomtheorem{customlemma}{Lemma}

\theoremstyle{remark}
\newtheorem{remark}{Remark}

\usepackage{geometry}
\onecolumn 

\usepackage{hyperref} 
    \hypersetup{
    	colorlinks=true,
    	linkcolor=blue,
    	filecolor=blue,      
    	urlcolor=blue,
    	citecolor=blue
    }
   
\usepackage{fancyhdr}
\title{Scalable Generalised Accuracy Estimation for Multisource Register-based Official Statistics}

\author{Nina Deliu$^{1,2,*}$,  Piero Demetrio Falorsi$^{1,4}$, Stefano Falorsi$^{3}$, Diego Chianella$^{3}$, and Giorgio Alleva$^{1}$\\
\normalsize $^{1}$MEMOTEF Department, Sapienza University of Rome, Italy\\ \normalsize$^{2}$MRC -- Biostatistics Unit, University of Cambridge, UK\\ \normalsize$^{3}$Former Head of the Division for Methodology, ISTAT, Rome, Italy (now Retired)\\ \normalsize$^{4}$ISTAT -- Italian National Institute of Statistics, Rome, Italy\\ \normalsize$^{*}$Corresponding Author: \url{nina.deliu@uniroma1.it}}

\date{}

\begin{document}
\maketitle

\begin{abstract}
Official statistics are undergoing a significant transformation, as national statistical institutes transition from traditional 
single-source data production systems to integrated systems of statistical registers combining administrative, census, and survey data. The resulting multisource register-based estimates are prone 
to multiple interacting sources of error, yet rigorous scalable frameworks for quantifying their accuracy remain underdeveloped. This work discusses and validates a global measure of error assessment for such multisource
register-based statistics. Focusing on two central sources of uncertainty --- sampling and modelling --- we derive an analytical solution that accurately approximates the global error of mass-imputation procedures under 
a multinomial logistic model. The proposed measure is 
interpretable, flexible, and computationally scalable, enabling on-the-fly accuracy quantification for user-defined, unplanned domain-specific statistics on population totals. Its validity is established theoretically and confirmed through simulation studies. An application to education data from the Italian National Institute of Statistics is presented.

\textbf{Key Messages}:
\begin{itemize}
\item We present a global measure of accuracy for multisource register-based official statistics, the \textit{Generalised Mean Squared Error} (GMSE), which provides an interpretable and computationally feasible solution for accuracy quantification, jointly accounting for multiple sources of error.
\item The GMSE framework is scalable and designed for operational integration into official-statistics production pipelines, enabling efficient accuracy estimation of unplanned domain-specific statistics.
\end{itemize}

\end{abstract}

\section{Introduction} \label{Sec: intro}
Official statistics play a crucial role in the information system of a democratic society, providing an accurate and timely picture of its economic, demographic, social, and environmental status and observed dynamics. In addition to making transparent information available to support individual and economic decision-making, they are indispensable for evaluating public policies, while adhering to a set of fundamental principles. 
In particular, Principle 15 on ``Accessibility and Clarity'' of the European Statistics Code of Practice for the National Statistical Authorities and Eurostat states that European statistics must be presented in a clear and understandable form, released in a suitable and convenient manner, available and accessible on an impartial basis with supporting metadata and guidance~\citep{eurostat2014,eurostat2020}. 

International and national statistical agencies (NSIs) are tasked with producing independent and authoritative statistics. 
The traditional data-production paradigm has long been centered on independent statistical surveys or administrative sources. Currently, NSIs are transitioning to a new model in which different data sources are combined into a single Integrated System of Statistical Registers~\citep[ISSR;][]{citro2014multiple, istat2016_modern}, which represents the basis of statistical production~\citep{ lohr_combining_2017, radermacher2018official, de2020multi}. Each original register stores units of a relevant target population, such as businesses, individuals and households, or agricultural holdings~\citep{wallgren2014register}, and arises from statistical surveys, administrative records other than official statistics, or digital technologies~\citep{radermacher2018official,daas2015big}. Therefore, an ISSR results in more compact, rich, and complete information content than the one associated with each independent data source. However, it also introduces challenges \citep{alleva2017, alleva2017b}. To allow an ISSR to be utilised as a single informative infrastructure, the integration process is achieved through different statistical methods, adopted to reconstruct unit-level data. These include, e.g., record linkage, statistical matching, and mass imputation, all prone to uncertainty and error. In this context, the European Statistics Code of Practice was revised in 2017~\citep{eurostat2020} to include provisions on the presence of systems for assessing and validating source data, integrated data, intermediate results, and statistical outputs. Principle 12 of the Code, on Data Accuracy and Reliability, explicitly highlights the need for rigorous quantitative evaluation of statistical uncertainty. Nonetheless, although robust methods exist for assessing single-source statistics~\citep{daas2011report, zhang2012topics, eurostat2014, biemer2014system}, frameworks for evaluating outputs from multi-source statistics remain less developed.

A critical question faced by NSIs is whether to restrict the usage of an ISSR to predefined, high-quality aggregates or to allow greater flexibility for users to derive customised estimates supported by transparent measures of accuracy. The latter better aligns with the notion of official statistics as a ``public good''~\citep{eurostat2020}, but demands rigorous methods to quantify the uncertainties inherent in integrated register-based estimates, as well as procedures that absolutely ensure the confidentiality of data. Notably, while accuracy estimation (under either the design- or model-based approach) is a classical problem in official statistics, the introduction of ISSRs requires inferential procedures to account for complex production pipelines, including mass imputation used to reconstruct unit-level responses when these are not directly available from individual sources. In this setting, uncertainty must reflect multiple interacting sources of error and be made available for user-defined, on-demand statistics. 

Motivated by Eurostat requirements and the ongoing modernisation process of the Italian National Institute of Statistics (ISTAT)~\citep{istat2016_modern}, our aim is to enable statistical institutes to empower authorised users to generate custom domain-specific statistics, while consistently providing measures of statistical accuracy. In this context, the proposed accuracy measure goes beyond internal quality assessment of register-based outputs: it is designed for integration into production pipelines, supporting computationally feasible uncertainty quantification, as well as user-driven and transparent dissemination activities. To this end, our contributions are structured along three lines: (i) formalising a global measure of uncertainty tailored to this multi-source setting, (ii) developing a computationally feasible implementation via a simplified two-step linearisation strategy, and (iii) evaluating and validating the proposed measure through simulation studies and real official-statistics data.

Building on the proposal of~\cite{alleva2021measuring}, we illustrate a ``global'' measure of accuracy that accounts for different sources of uncertainty, and extend it to more general type of multi-category outcomes that are central in NSIs (see Section~\ref{sec: motivation}). We term it \textit{generalised mean squared error} (GMSE). This concept also relates to the notion of \textit{global variance} \citep{wolter_coverage_1986} and \textit{anticipated variance} \citep{isaki_survey_1982,nedyalkova_optimal_2008}, with a notable fundamental difference. \cite{isaki_survey_1982} introduced the anticipated variance as an ex-ante measure (at the construction phase of sample designs), while we use it as an ex-post measure (once all statistical procedures for reconstructing unit-data values in an ISSR have been implemented). The goal is to provide to end users information on the uncertainty in the register statistics, {\it given} a certain adopted sampling design and imputation model. While here our focus will be on the design and model uncertainties, we emphasise that the GMSE is introduced as a generalised measure, where more than these two measures of uncertainty can be accounted for (e.g., linkage error, miscoverage, etc); this is not done in the anticipated variance.

Two main approaches have been followed in the literature to measure the accuracy of an estimator. The first relies on resampling methods, mainly focusing on either parametric or non-parametric bootstrap-based strategies~\citep[see e.g.,][for a survey]{mashreghi2016survey,shao2003impact}. While broadly applicable, these methods heavily depend on the number of data to be processed (which can be millions and millions in register data), and can lead to extreme computational burden. The second approach relies on approximations based on Taylor linearisation~\citep[see e.g.,][]{graf2014variance,vallee2019linearisation}, and requires the effort to derive analytical formulas case by case, but results in computationally efficient solutions. Our work belongs to this second stream of approaches. Compared to the linearisation strategy of~\cite{alleva2021measuring}, which was based on a four-step approximation procedure, we present a substantial advance: we show that a simplified two-step linearisation is sufficient to obtain an accurate and tractable approximation, reducing computational complexity while preserving the quality of the accuracy estimates. The proposed approach is first studied theoretically and subsequently validated through simulation studies. Systematic comparisons are presented across varying population sizes and domain specifications, against resampling methods and classical model- and design-based inferential frameworks.

The remainder of this work is organised as follows. In Section~\ref{sec: motivation}, we present motivational case studies for multi-category outcomes in official statistics. In Section~\ref{Sec: Setup}, we formalise the problem and introduce the notation. In Section~\ref{sec: gmse}, we develop the proposed GMSE solution in categorical variables, and discuss its analytical computation under general estimating equations, while in Section~\ref{sec: gmse_mult}, we characterise it for a multinomial logistic model. Section~\ref{sec: application} illustrates its application for the attained level of education in Italy and Section~\ref{sec: simulations} validates the proposed solution in simulations. Finally, Section~\ref{Sec: conc} outlines advantages and limitation of the GMSE, followed by some concluding remarks.

\subsection{Motivational Case Studies} \label{sec: motivation}
In this section, we showcase concrete examples of integrated registers within NSIs and multi-source statistics. We present cases on multi-category variables such as education, emphasising their centrality in official statistics.

\subsubsection{Base Register of Individuals of the Italian ISSR}

Starting from 2016, ISTAT has been undergoing a transition process from sample or census surveys to an ISSR. By combining data originating from different registers and sources, the resulting product is a single integrated system, such as the Base Register of Individuals containing individual subjects data. While most of the basic demographic variables are easily retrieved from administrative sources, other socio-economic information is collected by means of annual sample surveys. Among these, two variables are of fundamental importance for understanding the status of a population: the attained level of education consisting in a number of $K=8$ categories (going from \texttt{1 Illiterate} to \texttt{8 PhD level}), and the employment status (with $K=2$: \texttt{1 Employed}, \texttt{2 Not employed}). For completing the unit-level information for these variables, different data sources, integration processes, and imputation methods are adopted. For the attained level of education, for instance, such reconstruction accounts for: (A) administrative data from the Ministry of Education, University and Research; (B) data from the last traditional Italian census conducted in 2011; and (C) annual sample survey data. Log-linear and multinomial modelling approaches are used for imputation, with further details provided in Section~\ref{sec: application} and in~\cite{dizio1}.


\subsubsection{Dutch virtual Population and Housing Census}
In the Netherlands, a so-called virtual Population and Housing Census is conducted to integrate and update population data from existing administrative records, supplemented by sample survey information. While for many variables of interest, administrative sources provide (near-)complete population coverage, educational attainment remains a challenge. In fact, for the 2011 Census, this variable was observed from Dutch Labour Force Surveys, which is based on sample surveys and comprises only 300,000 individuals~\citep{daalmans2017mass}. For the 2021 Census, a more extensive educational source, the Educational Attainment File, was used to include data from several registers and sample surveys, with a coverage of more than 10 million people. Nonetheless, for approximately seven million Dutch individuals (out of a total population of 17 million), educational attainment is not observed. To address this gap, mass imputation has been used to predict all missing values for the variable of interest~\citep{daalmans2017mass}, and the consequent quality issue has been promptly raised~\citep{Scholtusetal}.

\subsubsection{Canada and New Zealand virtual Censuses}
Other examples concern the virtual censuses conducted in Canada in 2021~\citep{Lundy2022predicting} and in New Zealand in 2018~\citep{Bycroft2020Use}. In both cases, administrative data played a crucial role in ensuring completeness and accuracy of information, particularly in response to low participation rates. Both countries integrated traditional questionnaire responses with government register data to fill gaps left by non-responding households, with educational attainment being one of the key variables addressed through mass imputation.

\bigskip

The showcased examples highlight the complexity of adequately integrating multiple data sources, raising an important methodological question: how to assess the overall accuracy of population statistics derived from mass-imputed register data.


\section{Problem Setup and Notation} \label{Sec: Setup}

Let $R$ denote a statistical register of $N$ statistical units corresponding to a finite population, with $N$ fixed and known at a given moment of interest. We assume that the register is not affected by problems of overcoverage and/or undercoverage and that there is a one-to-one correspondence with the target population of interest. Each register unit $i$, with $i \in \{1,\dots, N\}$, is characterised by a set of observed covariates $\bm{x}_i = (x_{i1}, \dots, x_{iJ})^T$, with $\bm{x}_i \in \mathbb{R}^J$, $J\geq 1$, and a categorical response variable with $K$ mutually exclusive categories. We express the outcome of the categorical response of unit $i$ as a $K$-dimensional vector $\bm{Y}_i = (Y_{i1},\dots,Y_{ik}, \dots,Y_{iK})^T$, with $Y_{ik} \in \{0,1\}$ and such that $\sum_{k=1}^K Y_{ik} = 1$. The different response categories $k = 1,\dots, K$ can represent the educational levels (e.g., \texttt{Illiterate}, \texttt{Primary education}, ...)  an individual may have. 

In classical statistical modelling theory, the focus is on estimating model parameters that are intrinsically unobservable. In contrast, the primary aim in official statistics is to estimate quantities defined on a finite population -- quantities that can be in principle observable by carrying out a census of this population. For example, in a classical inference approach one may be interested in estimating the probability of having a PhD, after a reasonable model is postulated; in contrast, in official statistics, one would seek at inferring the total number of individuals with a PhD. Population totals represent one of the most common parameters of interest in official statistics; in our categorical case, we are specifically interested in the total of a given category $k$ in the entire register $R$ or for a specific subpopulation or domain of $R$, say $R^{(d)}$, where $d$ indexes the domain of interest, such as a given region. Let $\bm{\gamma}^{(d)} = \left(\gamma_1^{(d)}, \dots, \gamma_i^{(d)}, \dots,\gamma_N^{(d)} \right)^T$ be the known domain membership vector, where $\gamma_i^{(d)} \in \{0, 1\}$, with value $1$ indicating that unit $i$ belongs to domain $R^{(d)}$. Then, the target parameter at the domain level is formally defined as:
\begin{align} \label{eq: theta}
     \theta_k^{(d)} = \sum_{i=1}^N \gamma_{i}^{(d)}Y_{ik}
     ,\quad k=1,\dots,K.
\end{align}
Unless a census is conducted, $\theta_k^{(d)}$ is typically unknown as the values of $Y_{ik}, i=1,\dots,N$, are only known for a subset of the $N$ register units at the given moment of interest. In practice, responses are observed for a random sample of size $n$ and we use $\bm{\lambda} = \left(\lambda_1, \dots,\lambda_i, \dots, \lambda_{N} \right)^T$, with $\lambda_i \in \{0, 1 \}$ and $\sum_{i=1}^N\lambda_i = n$, to denote the sample membership (known once the sample is selected); value $1$ indicates a membership status and $0$ a non-membership. In the context of probabilistic surveys, each register unit $i$ is characterised by a certain sample inclusion probability $\pi_i$, defined according to a probabilistic sampling scheme. In such a context, for estimating the unknown parameter $\theta_k^{(d)}$ one must first estimate, impute, or make a prediction of the unobserved $Y_{ik}$ for all $i = 1,\dots,N$ for which $\lambda_i = 0$. 

To estimate $Y_{ik}$, we consider a working model $\mathbb{E}(Y_{ik}|\bm{x}_i) = f_k(\bm{x}_i;\bm{\beta}) = p_{ik}, k = 1,\dots,K$, with $f$ a known function depending on the unknown parameter vector $\bm{\beta} = \left(\bm{\beta}_1^T,\dots,\bm{\beta}_k^T,\dots, \bm{\beta}_K^T\right)^T$, where $\bm{\beta}_k = (\beta_{k1}, \dots, \beta_{kJ})^T$ is a $J$-dimensional object measuring the relationship between the outcome $Y_{ik}$ and the covariates $\bm{x}_i \in \mathbb{R}^J$, for $k=1,\dots,K$. Note that in our specific case $p_{ik} = \mathbb{P}(Y_{ik} = 1|\bm{x}_i)$, with $\sum_{k=1}^Kp_{ik} = 1, \forall i$. Let $\bm{\hat{\beta}}$ be a model-unbiased estimator of $\bm{\beta}$; then the predictor $\hat{Y}_{ik}$ of $Y_{ik}$ is given by
\begin{align}\label{eq: Y_hat}
    \hat{Y}_{ik} = \hat{p}_{ik} = f_k(\bm{x}_i;\bm{\hat{\beta}}), \quad i=1,\dots, N, \quad k=1,\dots,K.
\end{align}
A natural estimator $\hat{\theta}_k^{(d)}$ of the target vector parameter $\theta_k^{(d)}$ is then given by 
\begin{align} \label{eq: theta_hat}
     \hat{\theta}_k^{(d)} = \sum_{i=1}^N \gamma_{i}^{(d)}\hat{Y}_{ik} = \sum_{i=1}^N \gamma_{i}^{(d)}\hat{p}_{ik} = \sum_{i=1}^N \gamma_{i}^{(d)}f_k(\bm{x}_i;\bm{\hat{\beta}}),\quad k=1,\dots,K.
\end{align}

Note that, while here we follow the deterministic approach of~\cite{kim2009unified}, alternative approaches to Eq.~\eqref{eq: Y_hat} exist for predicting the unobserved $Y_{ik}$. For example, if one restricts the prevision $\hat{Y}_{ik} \in \{0,1\}$, this can be randomly generated from a multinomial distribution with parameters $\hat{\bm{p}}_i \doteq (\hat{p}_{i1}, \dots, \hat{p}_{iK})$; see also \textit{Remark~\ref{remark_Y}}. 

The main goal of this work is to estimate the accuracy, or alternatively, the error committed when estimating the parameter of interest. This depends on various sources of uncertainty, the first being the uncertainty derived by the adopted model. Furthermore, even in the case of correct specification of the model, it is noteworthy to emphasize that the prediction accuracy may heavily depend on the sampling scheme and the number of sample units $n$. 
In what follows, we will focus on the development of the proposed accuracy estimation procedure for the case of the multinomial logistic model, which generalises the simpler logistic case developed in~\cite{alleva2021measuring}. While we focus on this model because of its common adoption by NSIs, 
alternative models, such as latent-variable models, can be easily integrated within the presented framework.

\section{The Generalised Mean Squared Error}\label{sec: gmse}

The original idea of a global measure to evaluate the accuracy of register-based estimates was first proposed in \cite{alleva2021measuring}, building on previous literature on the notions of \textit{global variance} \citep{wolter_coverage_1986} and \textit{anticipated variance} \citep{isaki_survey_1982,nedyalkova_optimal_2008}. The global nature of this measure, which we name \textit{generalised mean squared error} (GMSE), stems from the different sources of uncertainty taken into account when quantifying the errors in the estimates or predictions. In this sense, it generalises the classical mean squared error (MSE), which represents the simplest way to measure accuracy focusing on a single source of randomness in a total survey error approach \citep{biemer_total_2010}. For instance, in a design-based approach (Cochran 1977), the sample design, which determines the sample membership $\bm{\lambda}$, represents the only source of randomness, whereas the population outcomes $Y_{ik}$, for $i = 1,\dots, N$ and $k =1,\dots, K$, are treated as unknown fixed quantities. Alternatively, a model-based approach~\citep{valliant2009model,chambers2012introduction} considers $\bm{\lambda}$ as fixed, with the only source of randomness being the $Y_{ik}$ values that are generated according to a certain model. In both cases, a single source of variability is taken into account when computing the MSE, that is, either the sampling design or the variability of the model, respectively. 

In this work, our aim is to take into consideration the different random components involved in the inferential process. More specifically, according to the assumptions made in Section \ref{Sec: Setup}, we focus on two sources of error typical of register-based statistics: the sampling and model uncertainty. These drive the definition of the GMSE provided in Definition~\ref{def: gmse}, which also generalises the \textit{global mean squared error} introduced in~\cite{alleva2021measuring}.

\begin{definition} \label{def: gmse}
Let $\hat{\theta}_k^{(d)}$ be an estimator of the parameter of interest $\theta_k^{(d)}$, function of the unknown outcome variables $Y_{ik}, i=1,\dots,N$, for $k=1,\dots,K$, and for a certain domain $d$. Under the setup in Section~\ref{Sec: Setup}, let $\bm{\lambda} = (\lambda_1,\dots,\lambda_N)$ be the random variable defining the sample membership, satisfying $\Exp(\lambda_i) = \pi_i$, with $\pi_i \in (0,1]$ being the inclusion probability of unit $i$, for $i = 1,\dots,N$. Furthermore, let $\mathbb{E}_{\M}(Y_{ik}|\bm{x}_i) = f_k(\bm{x}_i;\bm{\beta}) = p_{ik}, i = 1,\dots,N, k = 1,\dots,K$ be a working model for $Y_{ik}$ given the covariate set $\bm{x}_i$. Then, the GMSE of $\hat{\theta}_k^{(d)}$ with respect to the unknown parameter $\theta_k^{(d)}$ is defined as:
\begin{align} \label{eq: gmse_def}
    \GMSE\left(\hat{\theta}_k^{(d)}, \theta_k^{(d)}\right) &= 
    \Exp_{(\D, \M)}\left(\hat{\theta}_k^{(d)}-\theta_k^{(d)}\right)^2,\quad k=1,\dots, K,
\end{align}
that is, the error of the estimator under both sampling and model randomness.
\end{definition}
Compared to the \textit{global mean squared error} introduced in~\cite{alleva2021measuring}, the GMSE in Eq.~\eqref{eq: gmse_def} considers the expectation with respect to the joint distribution of the random variables $(\D, \M)$ characterised by the sampling design and the outcome model, respectively. In~\cite{alleva2021measuring}, the authors considered the measure $\Exp_{\D}\Exp_{\M}(\cdot)$. Nevertheless, we note that the two are equivalent when the sampling design is non-informative, that is $\D \perp \M$, such that the joint distribution can be factorised in the two univariate distributions. 

\begin{remark} The GMSE in Definition~\ref{def: gmse} is a scalar that refers to the individual error committed when estimating $\theta_k^{(d)}$ for category $k$. While this is the typical interest of NSIs, one may also be interested in a cumulated version of the GMSE accounting for all $k$s jointly, that is in $\GMSE\left(\bm{\hat{\theta}}^{(d)}, \bm{\theta}^{(d)}\right)$, with $\bm{\hat{\theta}}^{(d)} = (\hat{\theta}_1^{(d)},\dots, \hat{\theta}_K^{(d)})$. In that case, its definition can be naturally extended as:
    \begin{align} \label{eq: global_gmse_def}
    \GMSE\left(\bm{\hat{\theta}}^{(d)}, \bm{\theta}^{(d)}\right) 
    &= 
    \Exp_{(\D, \M)}\left((\bm{\hat{\theta}}^{(d)}-\bm{\theta}^{(d)})^T(\bm{\hat{\theta}}^{(d)}-\bm{\theta}^{(d)})\right) 
    = \sum_{k=1}^K \GMSE\left(\hat{\theta}_k^{(d)}, \theta_k^{(d)}\right),
\end{align}
that is, the scalar given by the sum of the $K$ individual errors. 
\end{remark}
Building on the results in~\cite{alleva2021measuring}, we now outline the following theoretical result.
\begin{theorem} \label{th: gmse_th1}
Assume that: i) the estimator $\hat{\bm{\beta}}$ 
is model-design unbiased, that is, $\Exp_{(\D, \M)}(\hat{\bm{\beta}}) = \Exp_{\D}\Exp_{\M}(\hat{\bm{\beta}}\mid \bm{\lambda}) = \bm{\beta}$, (ii) the chance of drawing a sample for which $\Exp_{\M}(\hat{\bm{\beta}} | \bm{\lambda}) \neq \bm{\beta}$ is negligibly small with $\sup_{\bm{\lambda}}\norm{\Exp_{\M}(\hat{\bm{\beta}} | \bm{\lambda}) - \bm \beta} < \infty$, and iii) the sampling design is non-informative \citep[see e.g., Section 1.4 of][]{chambers2012introduction}. Then, the following upper bounds hold for the individual-category and the cumulated GMSE:
\begin{align} \label{eq: gmse_th1}
    \GMSE\left(\hat{\theta}_k^{(d)}, \theta_k^{(d)}\right) &\lessapprox \Exp_{\D}\Var_{\M}(\hat{\theta}_k^{(d)}| \bm{\lambda})\\
    \GMSE\left(\bm{\hat{\theta}}^{(d)}, \bm{\theta}^{(d)}\right) &\lessapprox \Exp_{\D}\left[\mathrm{tr}\left(\Var_{\M}(\bm{\hat{\theta}}^{(d)}| \bm{\lambda})\right)\right], \nonumber
\end{align}
where $\mathrm{tr}(\mathrm{x})$ denotes the trace of a square matrix $\mathrm{x}$ and $\lessapprox$ stands for ``less than or approximately equal to''.
\end{theorem}
\begin{proof}
    See Supplementary material B.1.
\end{proof}

The assumption of non-informative sampling relates to survey sampling schemes that do not depend on the response variable, i.e., $\bm \lambda \perp \bm Y$. In practice, most of the official statistics production is designed under this criterion (see e.g., \citealt{chambers2012introduction}, Section~1.4) for their strong guarantees on inference. 

Note that the estimator $\hat{\theta}_k^{(d)}$, for $k=1,\dots,K$, is generally a non-linear function of the two random quantities $\bm{\lambda}$ and $\bm{\hat{\beta}}$, making the direct exact computation of the bound of the GMSE as defined in Eq.~\eqref{eq: gmse_th1} not straightforward. To overcome this challenge, we adopt a two-stage linearisation procedure using Taylor approximations to linearise the estimator, first, with respect to the argument of the inner variance term $\Var_{\M}(\hat{\theta}_k^{(d)}|\bm{\lambda})$, and second, with respect to the design randomness defining the outer expectation of the upper bound in Eq.~\eqref{eq: gmse_th1}. This result represents an improvement compared to the linearisation process adopted in~\cite{alleva2021measuring}, which was based on a four-step approximation procedure.

\paragraph{First linearisation step} Assume the real-valued functions $\hat{\theta}_k^{(d)} = \sum_{i=1}^N \gamma_{i}^{(d)}f_k(\bm{x}_i;\bm{\hat{\beta}}), k=1,\dots,K$, are twice
continuously differentiable functions at the point $\bm{\beta}$; then, they have a linear approximation near this point. The estimator $\hat{\theta}_k^{(d)}$ can thus be linearised with respect to $\bm{\hat{\beta}}$ at the point $\bm{\hat{\beta}} = \bm{\beta}$ as:
\begin{align} \label{eq: Taylor1}
    \hat{\theta}_k^{(d)} = \sum_{i=1}^N \gamma_{i}^{(d)}f_k(\bm{x}_i;\bm{\hat{\beta}}) = \sum_{i=1}^N \gamma_{i}^{(d)}\left( f_k(\bm{x}_i;\bm{\beta}) + 
   \frac{\partial f_k(\bm{x}_i;\bm{\hat{\beta}})}{\partial \bm{\hat{\beta}}} \Big\lvert_{\bm{\hat{\beta}} = \bm{\beta}} \left(\bm{\hat{\beta}} - \bm{\beta} \right) + r_{1k} \right),\quad k=1,\dots,K,
\end{align}
where $r_{1k}$ is the residual of this Taylor approximation and is of order $o_p(1/\sqrt{n})$,  so that the linearisation error is asymptotically negligible at rate~$n^{-1}$.

Using the result in Eq.~\eqref{eq: Taylor1}, we get the following linear approximation for the conditional variance:
\begin{align} \label{eq: 1st_linearisation}
 \Var_{\M}(\hat{\theta}_k^{(d)}|\bm{\lambda}) 
    &\approx \Var_{\M} \left(\sum_{i=1}^N \gamma_{i}^{(d)}\left( f_k(\bm{x}_i;\bm{\beta}) + 
   \frac{\partial f_k(\bm{x}_i;\bm{\hat{\beta}})}{\partial \bm{\hat{\beta}}} \Big\lvert_{\bm{\hat{\beta}} = \bm{\hat{\beta}}} \left(\bm{\hat{\beta}} - \bm{\beta} \right) \Big \lvert\bm{\lambda}\right) \right)\nonumber\\
   & = \Var_{\M} \left(\bm{\gamma}^{(d)T}
   \mathbf{F}_k \left(\bm{\hat{\beta}} - \bm{\beta} \right) \big \lvert \bm{\lambda}\right) \nonumber \\
   & = \bm{\gamma}^{(d)T} \mathbf{F}_k
   \Var_{\M} \left(\bm{\hat{\beta}} - \bm{\beta} \big \lvert \bm{\lambda}\right) \mathbf{F}^T_k \bm{\gamma}^{(d)},\quad \quad k=1,\dots,K, 
\end{align}
where $\mathbf{F}_k$ is the matrix of the first-order partial derivatives having dimension $N \times H$, with $H = K \times J$, and form
\begin{align*}
\underset{N \times H}{\mathbf{F}_k} = \left[\frac{\partial f_k(\bm{x}_i;\bm{\hat{\beta}})}{\partial \hat{\beta}_{lj}} \Big \lvert_{\hat{\bm{\beta}} = \bm{\beta} }\right],
\end{align*}
Note that the elements of $\mathbf{F}_k$ do not depend on $\bm{\hat{\beta}}$ and thus fixed with respect to model randomness. They will be later expanded in relation to the specific multinomial logistic model considered in this work.

\paragraph{Second linearisation step} Let $\mathbf{g}_{i}(\bm{\beta}; \bm{y}, \bm{x}) = \{g_{ikj}(\bm{\beta}; \bm{y}, \bm{x}), j=1,\dots,J, k=1,\dots,K\}$ define the system of $H$ generalised estimating equations~\citep{ziegler2011generalized} for estimating the model parameter vector $\bm{\beta}$ using the $n$ sample data. For example, under a maximum likelihood approach, denoted by $\ell(\bm{\beta}; \bm{y}, \bm{x}, \bm{\lambda})$ the log-likelihood of the working model, one could search for the estimator $\hat{\bm{\beta}}$ by solving:
\begin{align} \label{eq: gee}
    \sum_{i=1}^{N} \lambda_i \mathbf{g}_{i}(\bm{\beta}; \bm{y}, \bm{x}) = \frac{\partial \ell(\bm{\beta}; \bm{y}, \bm{x}, \bm{\lambda})}{\partial \bm{\beta}} = \mathbf{0}_{H},
\end{align}
where $\mathbf{0}$ is a vector of $H$ zeroes. Assume that for large $n$ the solution $\hat{\bm{\beta}}$ of Eq.~\eqref{eq: gee} is unique and that it converges to $\bm{\beta}$; then, using the result in \cite{chambers2012introduction} that holds under standard regularity conditions (see Supplementary material B.2), we can employ a first-order approximation to linearise $\bm{\hat{\beta}}$ around the expected value $\Exp_{\D}\Exp_{\M}(\hat{\bm{\beta}} | \bm{\lambda}) = \bm{\beta}$, obtaining: 
\begin{align} \label{eq: chambers}
\left(\bm{\hat{\beta}} - \bm{\beta} \right) \approx -\mathbf{A}_{\bm{\beta}}^{-1} \sum_{i=1}^N \lambda_i\mathbf{g}_{i}(\bm{\beta}; \bm{y}, \bm{x}),
\end{align}
where $\mathbf{A}_{\bm{\beta}}^{-1}$ is the inverse of the square matrix $\mathbf{A}_{\bm{\beta}}$ of dimension $H \times H$ defined as:
\begin{align}
    \underset{H \times H}{\mathbf{A}_{\bm{\beta}}} = \left[a_{(kj)(k'j')}\big \lvert_{\hat{\bm{\beta}} = \bm{\beta}}\right],\quad\text{with}\quad a_{(kj)(k'j')} \doteq \sum_{i=1}^{N} \lambda_i\frac{ \partial \mathbf{g}_{i}(\bm{\beta}; \bm{y}, \bm{x})}{\partial \hat{\beta}_{kj}} \Bigg \lvert_{\hat{\bm{\beta}} = \bm{\beta}},\quad j = 1,\dots, J; k  = 1,\dots, K,
\end{align}
which basically represents the second-order partial derivatives matrix, i.e. the Hessian matrix. As for the first linearisation step, the approximation in Eq.~\eqref{eq: chambers} holds with a remainder $\bm{r}_{2} = (r_{21},\dots,r_{2H})$ of minor order, each of which being $o_p(1/\sqrt{n})$; see Supplementary material B.2 for more details.

Taken together, the results in Eq.~\eqref{eq: 1st_linearisation} and Eq.~\eqref{eq: chambers} give us the following linearised version of the GMSE's bound:
\begin{align} \label{eq: 2st_linearisation}
 \Exp_{\D}\Var_{\M}(\hat{\theta}_k^{(d)}| \bm{\lambda}) 
    &\approx \bm{\gamma}^{(d)T} \mathbf{F}_k
   \Exp_{\D}\Var_{\M} \left(\bm{\hat{\beta}} - \bm{\beta} \big \lvert \bm{\lambda}\right) \mathbf{F}^T_k \bm{\gamma}^{(d)} \nonumber\\
   & \approx \bm{\gamma}^{(d)T} \mathbf{F}_k
   \Exp_{\D}\Var_{\M} \left(-\mathbf{A}_{\bm{\beta}}^{-1} \sum_{i=1}^N \lambda_i\mathbf{g}_{i}(\bm{\beta}; \bm{y}, \bm{x}) \Big \lvert\bm{\lambda}\right) \mathbf{F}^T_k \bm{\gamma}^{(d)}, \quad \quad k=1,\dots,K. 
\end{align}

In principle, while the non-linear dependence on the model estimator $\hat{\bm{\beta}}$ is now removed, the exact derivation of the GMSE's bound may still be challenging due to a potential non-linear dependence on the quantities $\bm{\lambda}$ and $\bm{Y}$; this will depend on the specific choice of the estimating equations $\mathbf{g}_{i}$. To overcome this issue, \cite{alleva2021measuring} considered two additional linearisation steps, one for each random component. In this work, we will show that one single additional linearisation step is sufficient to provide an exact solution for the commonly employed multinomial model under a maximum likelihood estimation (MLE) approach. We emphasise that all the necessary conditions for the two Taylor linearisation steps hold for the multinomial logistic model under bounded
covariate space and bounded category probabilities $p_{ik}$. 

\begin{remark}
    In what follows, the focus will be on the multinomial logistic model as one possible specification for multi-category outcomes. Nonetheless, as illustrated in the two-step linearisation above, the proposed approach is not tied to one model only. Since the GMSE derivation is based on (i) a function $f_k(\bm{x}_i;\bm{\beta})$ and (ii) an estimator $\hat{\bm{\beta}}$ defined through estimating equations, whenever these two components are available and sufficiently smooth, the same two-step Taylor linearisation applies under standard regularity conditions. For this reason, the methodology can be extended to broader classes of categorical models, including ordinal multinomial models, multinomial models with random effects, etc., and specifications allowing overdispersion or latent heterogeneity. 
\end{remark}

In particular, in Section~\ref{sec: gmse_mult} the following Theorem~\ref{th: gmse_multin} will be proved, after a formal introduction of the model. 
\begin{theorem} \label{th: gmse_multin}
Consider the setup in Section~\ref{Sec: Setup} and assume that the outcome variable $\bm{Y}_i \sim \text{Multinomial}(1, \bm{p}_i)$, with $\bm{p}_i \in [0,1]^K$ and $\sum_{k=1}^K p_{ik} = 1$, for $i = 1,\dots,N$, with the $N$ units being independent. Let $\Sigma_{\bm{Y}_{i}}$ denote the population variance of the model outcome $\Var_{\M} (\bm{Y}_{i})$ of unit $i$, for $i = 1,\dots,N$; then, under the assumptions of Theorem~\ref{th: gmse_th1}, the following holds:
    \begin{align*}
    \GMSE\left(\hat{\theta}_k^{(d)}, \theta_k^{(d)}\right) &\lessapprox \Exp_{\D}\Var_{\M}(\hat{\theta}_k^{(d)}| \bm{\lambda}) \nonumber\\
    &\approx \bm{\gamma}^{(d)T} \mathbf{F}_k \left( \sum_{i=1}^N \pi_i \bar{\mathbf{U}}_{i} \Sigma_{\bm{Y}_{i}} \bar{\mathbf{U}}_{i} \right) \mathbf{F}^T_k \bm{\gamma}^{(d)},\quad k=1,\dots,K,
\end{align*} 
with $\mathbf{F}_k$ and $\bar{\mathbf{U}}_{i}$ fixed matrices depending on the covariate data $\bm{x}_i$ and the model parameters. These are further explicitated in Section~\ref{sec: gmse_mult} in the case of a multinomial model.
\end{theorem}
\begin{proof}
    See Section~\ref{sec: gmse_mult}.
\end{proof}

The tightness of the upper bound in Theorem~\ref{th: gmse_multin} (as well as Theorem~\ref{th: gmse_th1}) is discussed in detail in Supplementary material B.1, where we show that the gap between the GMSE and its approximated upper bound $\Exp_{\D}\Var_{\M}\left(\hat{\theta}_k^{(d)}|\bm{\lambda}\right)$ is of minor order, driven e.g., by $\Var_{\M}(\theta_k^{(d)})$ for internal domains. In the same Supplementary material B.1, we also provide an empirical quantification of these components, confirming that $\Exp_{\D}\Var_{\M}\left(\hat{\theta}_k^{(d)}|\bm{\lambda}\right)$ is the dominant term.

\subsection{The GMSE for the Multinomial Logistic Model}\label{sec: gmse_mult}

Driven by motivational case studies in ISTAT, we assume that the outcome variable for each unit $i$, with $i = 1,\dots, N$, follows a multinomial regression model, that is, $\bm{Y}_i \sim \text{Multinomial}(1, \bm{p}_i)$, where $\bm{p}_i = \left(p_{i1}, \dots, p_{ik}, \dots,p_{iK} \right)^T$ denotes the unknown event probabilities in $[0,1]$ of each category $k$, for $k=1,\dots,K$, such that $\sum_{k=1}^Kp_{ik} = 1$. These probabilities are modelled as covariate-dependent functions, with covariates $\bm{x}_i \in \mathbb{R}^J$, for all $i = 1,\dots,N$, and then estimated according to a multinomial logistic regression, where $K$ is taken as the baseline category:
\begin{align} \label{eq: mult_model}
    p_{ik} = f_k(\bm{x}_i;\bm{\beta}) = 
    \begin{dcases}
    \frac{\exp{\bm{x}_i}^T\bm{\beta}_k}{ 1 + \sum_{k=1}^{K-1}\exp{\bm{x}_i}^T\bm{\beta}_k}, \quad & k=1,\dots,K-1\\
    \frac{1}{ 1 + \sum_{k=1}^{K-1}\exp{\bm{x}_i}^T\bm{\beta}_k}, \quad & k = K. 
    \end{dcases}
\end{align}

An estimate of the unknown parameters $\bm{\beta}_k$, therefore, $p_{ik}$ can be obtained by MLE, and we shall focus on this estimation approach from now on. Following standard statistical results, which are fully detailed in Supplementary material B.2, under the assumption of independence among the $N$ units and accounting for the sample membership $\bm{\lambda}$, the log-likelihood of the model in Eq.~\eqref{eq: mult_model} is given by:
\begin{align*}
    \ell(\bm{\beta}; \bm{y}, \bm{x}, \bm{\lambda}) = \sum_{i = 1}^{N} \lambda_i \left[\sum_{k=1}^{K-1} y_{ik} \bm{x}_i^T\bm{\beta}_k - \log\left(1 + \sum_{k=1}^{K-1}\exp{\bm{x}_i^T\bm{\beta}_k}\right)\right],
\end{align*}
and the MLEs $\bm{\hat{\beta}}_k$ are obtained as the solution of the problem in Eq.~\eqref{eq: gee}, with $\mathbf{g}_{i}(\bm{\beta}; \bm{y}, \bm{x})$ being the $H$-dimensional vector having components:
\begin{align} \label{eq: gee_multin}
     g_{ikj}(\bm{\beta}; \bm{y}, \bm{x}) = x_{ij} \left( y_{ik} - p_{ik} \right),\quad j=1,\dots,J, k=1,\dots,K.
\end{align}
Note that the components of $\mathbf{g}_{i}(\bm{\beta}; \bm{y}, \bm{x})$ have a non-linear dependence on the fixed parameter vector $\bm{\beta}$ through $\bm{p}_{i}$, but they depend linearly on the random outcome $\bm{Y}_{i}$. The Hessian matrix $\mathbf{A}_{\bm{\beta}}$ results in a further simplification with the dependency on $\bm{Y}_{i}$ removed (the calculus is deferred to Supplementary material B.2):
\begin{align} \label{eq: A_multin}
    a_{(kj)(k'j')} = \frac{\partial^2 \ell(\hat{\bm{\beta}}; \bm{y}, \bm{x}, \bm{\lambda})}{\partial \hat{\beta}_{kj} \partial \hat{\beta}_{k'j'}} \Bigg \lvert_{\hat{\bm{\beta}} = \bm{\beta}} = \begin{cases}
     -\sum_{i = 1}^{N} \lambda_i x_{ij} x_{ij'}p_{ik}(1-p_{ik}),  &\quad k' = k \\
     \sum_{i = 1}^{N} \lambda_i x_{ij} x_{ij'}p_{ik}p_{ik'},  &\quad k' \neq k 
     \end{cases}.
\end{align}
The results in Eqs.~\eqref{eq: gee_multin}-\eqref{eq: A_multin} facilitate the analytical derivation of the variance in Eq.~\eqref{eq: 2st_linearisation}. To see this, let $\Sigma_{\bm{Y}_{i}}$ be the $K \times K$ variance-covariance matrix of the outcome variable, i.e., $\Sigma_{\bm{Y}_{i}} \doteq \Var_{\M} (\bm{Y}_{i})$, and denote by $\mathbf{\dot{X}}_{i}$ and $\mathbf{U}_{i}$ the two matrices given by
\begin{align} \label{eq: x_dot}
    \underset{H \times K}{\mathbf{\dot{X}}_{i}} =
    \begin{bmatrix} 
    x_{i1} & 0& \dots  & 0\\
    \vdots & \dots &  \dots & \vdots\\
    x_{iJ} &\dots&  \dots & \vdots\\
    0 & \dots &  \dots & 0\\
    \vdots & \dots &  \dots & x_{i1}\\
    \vdots & \dots &  \dots & \vdots\\
    0 & \dots & 0 &  x_{iJ}
    \end{bmatrix},\quad \quad \quad \quad \underset{H \times K}{\mathbf{U}_{i}} \doteq \mathbf{A}_{\bm{\beta}}^{-1} \mathbf{\dot{X}}_{i}.
\end{align}
Then, we can re-express the model variance in Eq.~\eqref{eq: 2st_linearisation} as
\begin{align} \label{eq: 2st_linearisation_multinomial}
 \Var_{\M}(\hat{\theta}_k^{(d)}| \bm{\lambda}) 
   & \approx \bm{\gamma}^{(d)T} \mathbf{F}_k
   \Var_{\M} \left(-\mathbf{A}_{\bm{\beta}}^{-1} \sum_{i=1}^N \lambda_i\mathbf{g}_{i}(\bm{\beta}; \bm{y}, \bm{x}) \Big \lvert\bm{\lambda}\right) \mathbf{F}^T_k \bm{\gamma}^{(d)} \nonumber\\
   & = \bm{\gamma}^{(d)T} \mathbf{F}_k
   \left(\sum_{i=1}^N \lambda_i^2\mathbf{U}_{i} \Var_{\M} (\bm{Y}_{i}) \mathbf{U}_{i}^T + \sum_{j\neq i}^N 
   \lambda_i^2\mathbf{U}_{i} \Cov_{\M} (\bm{Y}_{j}, \bm{Y}_{i}) \mathbf{U}_{i}^T\right) \mathbf{F}^T_k \bm{\gamma}^{(d)} \nonumber\\
   & = \bm{\gamma}^{(d)T} \mathbf{F}_k \left( \sum_{i=1}^N \lambda_i^2\mathbf{U}_{i} \Sigma_{\bm{Y}_{i}} \mathbf{U}_{i}^T \right) \mathbf{F}^T_k \bm{\gamma}^{(d)}, \quad \quad k=1,\dots,K. 
\end{align}
since, by assumption, $\Cov_{\M} (\bm{Y}_{j}, \bm{Y}_{i}) = 0, \forall i\neq j$. Note that, in light of the considered model, the exact formulation for matrix's $\mathbf{F}_k$ elements is now given by:
\begin{align} \label{eq: Fk}
   \left.\frac{\partial f_k(\bm{x}_i;\bm{\hat{\beta}})}{\partial \hat{\beta}_{lj}} \right |_{\hat{\beta}_{lj} = \beta_{lj} } =  \begin{cases}
x_{ij}p_{ik}(1-p_{ik}), \quad &l = k\\ 
 -x_{ij}p_{ik}p_{il}, \quad &l \neq k. 
    \end{cases}
\end{align}


The approximation for the GMSE's bound is finally obtained by calculating the expectation of Eq.~\eqref{eq: 2st_linearisation_multinomial} with respect to the design randomness in $\bm{\lambda}$. Since its argument depends on $\bm{\lambda}$ also through $\mathbf{U}_{i}$, more specifically, through the inverse of Hessian matrix $\mathbf{A}_{\bm{\beta}}$, one could proceed either with an additional Taylor approximation to remove the complicated dependence of $\bm{\lambda}$~\citep[as done in][]{alleva2021measuring}, or by using the following lemma. 
\begin{lemma} 
Consider the setup in Section~\ref{Sec: Setup} and assume that the outcome variable $\bm{Y}_i \sim \text{Multinomial}(1, \bm{p}_i)$, with $\bm{p}_i \in [0,1]^K$ and $\sum_{k=1}^K p_{ik} = 1$, for $i = 1,\dots,N$, with the $N$ units being independent. Let $\mathbf{A}_{\bm{\beta}}$ be the Hessian matrix associated with this model. Then, under a non-informative design, 
when the inclusion probabilities $\pi_i$ are close to the extremes 0/1, the following approximation holds:
    \begin{align*}
    \mathbf{A}_{\bm{\beta}} \approx \mathbf{A}_{\bm{\beta}} |_{\bm{\lambda} = \bm{\pi}} \doteq \bar{\mathbf{A}}_{\bm{\beta}}.
    \end{align*}
Furthermore, when $\pi_i \in \{0,1\}, i=1,\dots,N$, the result is an identity.
\end{lemma}
\begin{proof}
    The proof and further considerations on this approximation are provided in the Supplementary material B.3.
\end{proof}
\begin{remark}
    The assumption of $\pi_i$ close to the extremes 0/1 is a standard assumption for inference under the model-based approach; see e.g., pg. 57 in \cite{lohr2021sampling}. Furthermore, in practice, $\pi_i$ are either very small (e.g., for the household surveys or for small enterprise surveys) or close to 1 (e.g., large enterprise surveys).
\end{remark}

Bringing this result together with Eqs.~\eqref{eq: 2st_linearisation}-\eqref{eq: 2st_linearisation_multinomial}, and denoting by 
\begin{align} \label{eq: U_bar}
    \bar{\mathbf{U}}_{i} \doteq \bar{\mathbf{A}}_{\bm{\beta}}^{-1} \mathbf{\dot{X}}_{i},
\end{align}
we get the final result:
\begin{align} \label{eq: GMSE_final}
    \GMSE\left(\hat{\theta}_k^{(d)}, \theta_k^{(d)}\right) &\lessapprox \Exp_{\D}\Var_{\M}(\hat{\theta}_k^{(d)}| \bm{\lambda}) \nonumber\\
    &\approx \bm{\gamma}^{(d)T} \mathbf{F}_k \left( \sum_{i=1}^N \Exp_{\D}\left(\lambda_i^2\right) \bar{\mathbf{U}}_{i} \Sigma_{\bm{Y}_{i}} \bar{\mathbf{U}}_{i} \right) \mathbf{F}^T_k \bm{\gamma}^{(d)}\nonumber \\
    &\approx \bm{\gamma}^{(d)T} \mathbf{F}_k \left( \sum_{i=1}^N \pi_i \bar{\mathbf{U}}_{i} \Sigma_{\bm{Y}_{i}} \bar{\mathbf{U}}_{i} \right) \mathbf{F}^T_k \bm{\gamma}^{(d)},\quad k=1,\dots,K,
\end{align}
where $\Exp(\lambda_i^2) = \pi_i$, since $\lambda_i$ is a binary sample membership indicator, and this holds for any arbitrary sampling design~\citep{sarndal2003model}. Note that the analytical expression of the GMSE in Eq.~\eqref{eq: GMSE_final} can be adapted to different sampling designs, by characterising the inclusion probabilities $\pi_i$ accordingly.


The final linearised bound in Eq. \eqref{eq: GMSE_final} depends on the following quantities: the domain membership indicator $\gamma^{(d)}$, the sampling design specification $\bm{\pi}$, the matrix $\mathbf{F}_k$ as expressed in Eq. \eqref{eq: Fk}, the matrix $\bar{\mathbf{U}}_{i}$ given in Eq. \eqref{eq: U_bar}, and the variance-covariance matrix $\Sigma_{\bm{Y}_{i}}$. While the first two elements are pre-specified and known in advance, the latter matrices all depend on the parameter vector $\bm{p}$, which, in practice, is unknown and has to be estimated. We use its sample counterpart $\hat{\bm{p}}$, leading to the following linearised $\GMSE$ estimator:
\begin{align} \label{eq: GMSE_Lin}
\widehat{\GMSE}^{\text{Lin}}\left(\hat{\theta}_k^{(d)}, \theta_k^{(d)}\right) = \bm{\gamma}^{(d)T} \widehat{\mathbf{F}}_k \left( \sum_{i=1}^N \pi_i \widehat{\bar{\mathbf{U}}}_{i} \widehat{\Sigma}_{\bm{Y}_{i}} \widehat{\bar{\mathbf{U}}}_{i} \right) \widehat{\mathbf{F}}^T_k \bm{\gamma}^{(d)},\quad k=1,\dots,K,
\end{align}
with the generic matrix estimator $\widehat{D} = D|_{\bm{p} = \hat{\bm{p}}}$.

\begin{remark}
    In Eq.~\eqref{eq: 2st_linearisation_multinomial} and thereafter, we have assumed that $\Cov_{\M} (\bm{Y}_{j}, \bm{Y}_{i}) = 0, \forall i\neq j$. In models such as latent-class or longitudinal models, where this does not hold in general, one may use e.g., a plug-in estimate of $\Cov_{\M} (\bm{Y}_{j}, \bm{Y}_{i})$, making it possible to generalise the  GMSE to those settings. In that case, the covariance component does not vanish and the linearised GMSE takes the form
\begin{align*}
\widehat{\mathrm{GMSE}}^{\,\mathrm{Lin}}\!\left(\hat\theta_k^{(d)}, \theta_k^{(d)}\right)
  = \boldsymbol{\gamma}^{(d)T}\hat{\mathbf{F}}_k
  \left( \sum_{i=1}^N \pi_i \widehat{\bar{\mathbf{U}}}_{i} \widehat{\Sigma}_{\bm{Y}_{i}} \widehat{\bar{\mathbf{U}}}_{i} 
    + \sum_{\ell \neq i} \pi_{i,\ell}\,
      \widehat{\bar{\mathbf{U}}}_{i} \widehat{\Sigma}_{\bm{Y}_{i,\ell}} \widehat{\bar{\mathbf{U}}}_{i} 
  \right) 
  \hat{\mathbf{F}}_k^T \boldsymbol{\gamma}^{(d)},
\end{align*}
where $\widehat{\Sigma}_{\bm{Y}_{i,\ell}}$ is a plug-in estimate of
$\mathrm{Cov}_{\bm Y}(\boldsymbol{Y}_i, \boldsymbol{Y}_\ell)$ and
$\pi_{i,\ell}$ is the joint inclusion probability of units $i$ and $\ell$.
Since the exact values of $\pi_{i,\ell}$ are typically unknown, they can be approximated by
$\pi_{i,\ell} \cong \pi_i\pi_\ell \frac{N}{n}\frac{n-1}{N-1}$,
which holds exactly under simple random sampling (SRS) without replacement.
\label{remark_cov}
\end{remark}

\begin{remark} \label{remark_Y}
    The analytical expression for the GMSE derived in Eq.~\eqref{eq: GMSE_final} is based on the estimator in Eq.~\eqref{eq: theta_hat}, which considers $\hat{Y}_{ik} = \hat{p}_{ik}$, for all $i$ and all $k$. However, in many cases, an estimator directly targeting the nature of the response variable, i.e., $\hat{Y}_{ik} \in \{0,1\}$, may be of interest, instead. In a multinomial model, this is typically achieved by carrying out an additional step that involves taking a random draw from a multinomial distribution with parameters $\hat{\bm{p}}_i$, for $i = 1,\dots,N$. Since this step introduces an additional source of variability, this must be reflected in the GMSE. Under the same conditions of Theorem~\ref{th: gmse_multin}, but focusing on the estimator $\tilde{\hat{\theta}}_k^{(d)} = \sum_{i=1}^N \gamma_i \tilde{\hat{Y}}_{ik}$, with $\tilde{\hat{Y}}_i \sim \text{Multinomial}(\hat{\bm{p}}_i)$, it is straightforward to show that
    \begin{align*}
    \widehat{\GMSE}^{\text{Lin}}\left(\tilde{\hat{\theta}}_k^{(d)}, \theta_k^{(d)}\right) = \widehat{\GMSE}^{\text{Lin}}\left(\hat{\theta}_k^{(d)}, \theta_k^{(d)}\right) + \sum_{i=1}^N \gamma_i^{(d)}\hat{p}_{ik}(1-\hat{p}_{ik}),\quad k=1,\dots,K,
    \end{align*} 
    where the additional (positive) component reflects the variability of the outcomes of a multinomial model.
\end{remark}

To facilitate implementability a pseudocode for the GMSE derivation is provided in the Supplementary material C. Further, to enhance computational efficiency, an alternative formulation using the Kronecker product is provided in the Supplementary material E.

\subsection{Computational Scalability and Feasibility} \label{sec: scalability}

Following a useful suggestion from one of the Reviewers, we explicitly discuss the computational scalability and practical feasibility of the proposed GMSE approach. This aspect is particularly relevant in modern register-based official statistics, where target populations may contain millions of units and users increasingly demand timely, domain-specific outputs. In this context, computability is not a secondary issue, but a core requirement for any methodology intended for operational use. It also represents an important advantage of the proposed framework when compared with simulation-based accuracy methods, such as bootstrap or other resampling procedures, whose repeated re-estimation steps may become computationally burdensome in large-scale applications. We analyse these aspects with reference to Eq.~\eqref{eq: GMSE_Lin}, covering domain characterization, storage requirements, computational complexity, and the numerical stability of the quantities used in the computation. 

\paragraph{Domain scalability}
Our objective is to describe a method enabling statistical institutes with register systems to allow users to generate custom domain-specific statistics {\it on the fly}, along with a simultaneous measure of statistical accuracy to support informed use. A key feature is the separation between domain-invariant and domain-specific quantities in the GMSE formulation in Eq.~\eqref{eq: GMSE_Lin}. The matrices $\widehat{\mathbf{F}}_k$, 
$\widehat{\bar{\mathbf{U}}}_{i}$,
$\widehat{\Sigma}_{\bm{Y}_i}$ 
do not depend on the domain specification $d$, as they are determined by the model and the estimated probabilities $\hat{\bm{p}}_i$. The only domain-specific component is the vector $\bm{\gamma}^{(d)}$, which encodes the membership of units to the domain and is structurally part of the register. As a result, once the domain-invariant quantities are computed, GMSE estimates for different domains can be obtained {\it on the fly} by simply updating $\bm{\gamma}^{(d)}$, with negligible additional computational cost. This property is particularly advantageous in interactive environments, where users may define multiple domains dynamically, as it avoids repeated estimation steps and ensures fast and scalable computation.


\paragraph{Numerical stability}
A potential concern in the implementation of Eq.~\eqref{eq: GMSE_Lin} relates to the inversion of the Hessian matrix $\bar{\mathbf{A}}_{\bm{\beta}}$ in Eq.~\eqref{eq: U_bar}, whose dimension is $H \times H$ with $H = K \times J$. In applications with a large number of categories $K$ or covariates $J$, direct inversion may become computationally demanding and numerically unstable, especially in the presence of collinearity or sparse data structures. 
In our implementation, to enhance numerical stability we used Moore-Penrose generalised inverse, which exists and is unique for every matrix, proving robustness to near-singularity and multicollinearity issues. This is implemented in \texttt{ginv()} function from the \texttt{MASS} package in \textsf{R} \citep{venables2013modern}. Importantly, the computation of $\bar{\mathbf{A}}_{\bm{\beta}}^{-1}$ is required only once and can be reused across all domains $d$ and categories $k$, substantially reducing the overall computational burden.

\paragraph{Memory load and dimensionality} The $(H \times H)$ matrix $\widehat{\bar{\mathbf{A}}}_{\bm{\beta}}^{-1}$ is not problematic, as it requires relatively limited space even for moderately large $H$. In contrast, explicitly storing unit-level matrices $\bar{\mathbf{U}}_{i}$ and $\widehat{\Sigma}_{\bm{Y}_i}$ for $i=1,\dots,N$, may become prohibitive in large-scale applications, especially for large $K$ and $N$. To address this issue, we rely on an {\it on-the-fly} computation strategy. The rows $\hat{\mathbf{f}}_{ki}$ of $\widehat{\mathbf{F}}_k$ can be directly obtained from the analytical expression in Eq.~\eqref{eq: Fk} using $\bm{x}_i$ and $\hat{\bm{p}}_i$. Similarly, $\mathbf{\dot{X}}_i$ is a simple transformation of $\bm{x}_i$, and $\widehat{\Sigma}_{\bm{Y}_i}$ can be computed from $\hat{\bm{p}}_i$, with diagonal elements $\hat{p}_{ik}(1-\hat{p}_{ik})$ and off-diagonal elements $-\hat{p}_{ik}\hat{p}_{ik'}$ for $k \neq k'$. As a consequence, there is no need to store these matrices explicitly for all units, avoiding memory costs of order $O(NHK)$ and reducing storage requirements to $O(H^2)$, plus the storage of $\hat{\bm{p}}_i$, the parameter vector $\hat{\bm{\beta}}$, and the inclusion probabilities $\pi_i$.

\section{Application to the Attained Level of Education}\label{sec: application}

This section illustrates the application of the GMSE to education data from the Base Register of Individuals of the Italian ISSR. As described in Section~\ref{sec: motivation}, complete unit-level information for the attained level of education in Italy for a reference year $t$ (denoted by $Y_{ik}, i=1,\dots,N; k=1,\dots,K$) is retrieved from different data sources. Specifically, annual sample surveys are conducted to infer the education level of the entire population, while using auxiliary information from other two data sources: administrative data from the Ministry of Education, University and Research and the last Italian Census in 2011. These data sources define a population partition into three subgroups that differ in their informative content: 
\begin{itemize}
    \item Subgroup A ($22\%$ of the register units) -- administrative data: Individuals who entered a national study program between 2011 and $t-2$, attended a course in Italy between $t-2$ and $t-1$, and for whom the information on the education level at $t-1$ is available. Based on the reference year $t$, this information may change, therefore, it represents the dynamic component of the ISSR.
    \item Subgroup B ($73\%$ of the register units) -- 2011 census data: Individuals not in Subgroup A (not enrolled in any course between $t-2$ and $t-1$), but for whom education data from the 2011 Italian Census is available. 
    \item Subgroup C ($5\%$ of the register units): Individuals not recorded in the administrative or census data, therefore, with no direct information on the education level from the previous years. This subgroup is composed mainly of adults, with a high percentage of not Italian citizens ($>60\%$).
\end{itemize}

Unobserved unit-level education data at time $t = 2019$ (those not included in the sample) are imputed using a model-based approach, with the use of different conditional models varying in their covariate set according to the subgroup. For expositional simplicity and without loss of generality, we restrict the analysis to Subgroup B, which covers most of the register units ($73\%$). Here, the conditional probabilities $p_{ik} = \mathbb{P}(Y_{ik} = 1 | \bm{x}_i)$ depend on a set of covariates $\bm{x}$ given by $X_1 : \texttt{age class}$, $X_2 : \texttt{gender}$, $X_3 : \texttt{Italian citizenship}$, $X_4 : \texttt{attained education in}\ 2011$, and $X_5 : \texttt{province}$. A number of $K = 8$ education categories are considered: {\tt 1 Illiterate}, {\tt 2 Literate but no formal educational attainment}, {\tt 3 Primary education}, {\tt 4 Lower secondary education}, {\tt 5 Upper secondary education}, {\tt 6 Bachelor's degree or equivalent level}, {\tt 7 Master’s degree or equivalent level}, {\tt 8 PhD level}. For further details on the model choice and imputation procedure we refer to \cite{di2019imputation}.

For data confidentiality, the specific application illustrated here is based on a random subset of the entire Base Register of Individuals, accounting for $\approx 10\%$ of the register units, and to the Italian region of Emilia Romagna. Therefore, this empirical application is to be regarded as an illustrative case study rather than a comprehensive national analysis. Nonetheless, the proposed methodology is not limited to this specific setting and extends directly to other regions, subgroups, domains, and broader national-level applications within integrated statistical registers. Overall, a subset of $N = 296,565$ register units has been considered in this analysis; data are made publicly available in a link on our Github repository: \url{https://github.com/nina-DL/GMSE/}. Among these, a total of $n = 14,860$ units ($5\%$ of the total) have been surveyed within the annual sample survey program for the reference year $t=2019$. Descriptive summaries on both covariates (available for all register units) and the outcome of interest (available only for the sample units) are reported in Table~\ref{tab: covariates} and in Figure~\ref{fig: edu_counts} (right panel). Figure~\ref{fig: edu_counts} (left panel) provides the register estimated totals, after mass-imputation.

\begin{table}[b]
\centering 
\caption{Register data of Subgroup B (covariates, sample and domain 
membership) and marginal properties.} \label{tab: covariates}
  \begin{tabular}{lll}
    \toprule
      & Variable type  & Marginal frequency $\bm{p}_{X_j}$\\
    Covariates & & \\
    \midrule
$X_1: \texttt{age class}$ (\texttt{(,28]})  & Categorical: $\{1,\dots,5\}$ 
    &  $(0.012, 0.115, 0.208, 0.395, 0.270)$\\ 
$X_2 : \texttt{gender}$ (\texttt{Female}) & Binary: $\{0, 1\}$ 
    &  $(0.477, 0.523)$\\ 
$X_3 : \texttt{Italian citizenship}$ (\texttt{Yes})  & Binary: $\{0, 1\}$ 
    &  $(0.071, 0.929)$\\ 
$X_4: \texttt{education in 2011}$ & Categorical: $\{1,\dots,8\}$ 
    & $(0.005, 0.024, 0.171, 0.294, 0.363, 0.026, 0.113, 0.004)$\\
    \midrule
    Sample membership & & \\
    \midrule
$\lambda$: Sample indicator (\texttt{Yes})  & Binary: $\{0, 1\}$ 
    &  $(0.950, 0.050)$ \\ 
    \midrule
    \multicolumn{3}{l}{Domain membership (internal: used as covariate 
    in the fitted model)} \\
    \midrule
$\gamma^{(d)}, d = X_2$: \texttt{gender} & Binary: $\{0, 1\}$ 
    &  $(0.477, 0.523)$\\ 
    \midrule
    \multicolumn{3}{l}{Domain membership (external: \textit{not} used 
    as covariate in the fitted model)} \\
    \midrule
$\gamma^{(d)}, d = X_5 : \texttt{province}$ & Categorical: $\{1,\dots,9\}$ 
    & $(0.075, 0.065, 0.09, 0.1, 0.09, 0.12, 0.08, 0.16, 0.23)$\\
    \bottomrule
  \end{tabular}
\end{table}

\begin{figure}[b]
    \centering
    \includegraphics[scale = .5]{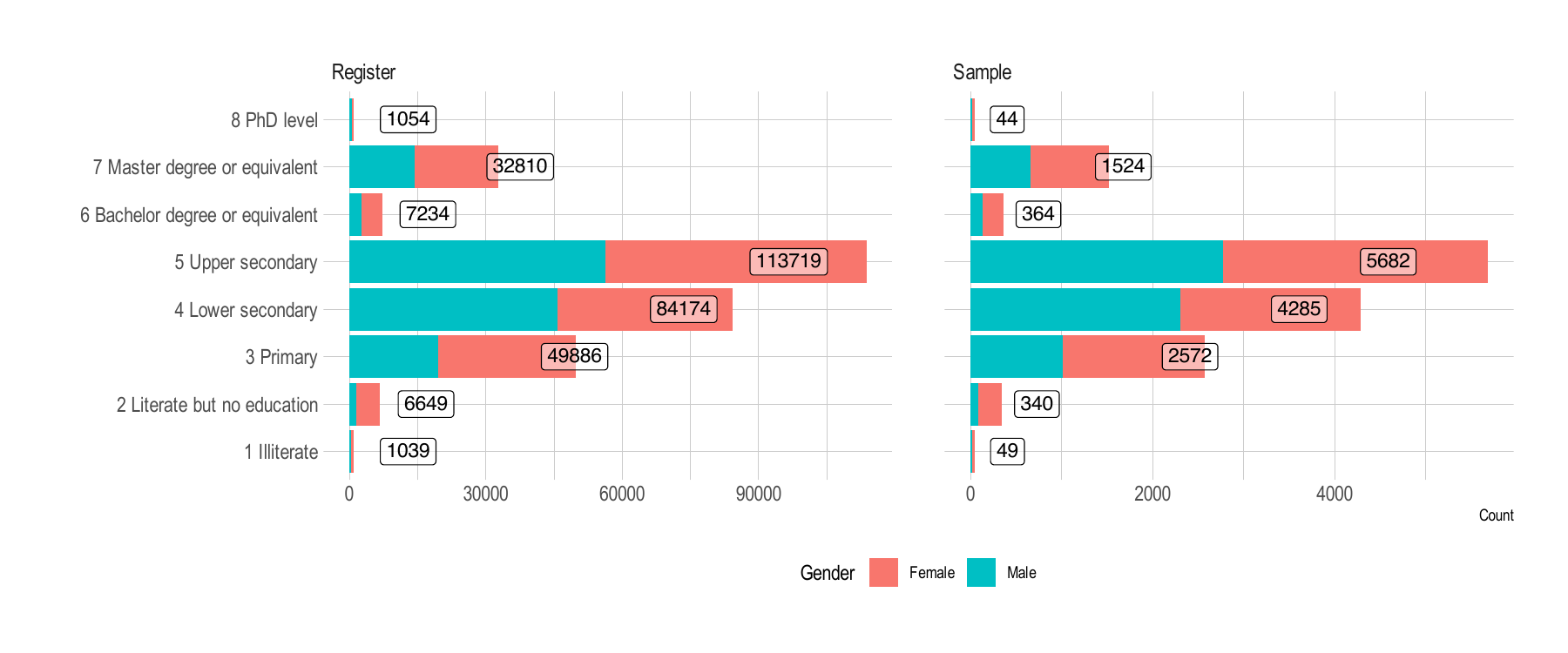}
    \caption{Number of individuals (counts) in each of the $K=8$ categories of the attained level of education: population estimates $\hat{\theta}^{(d)}_k, k=1,\dots,8$ (Subgroup B of the register; left) obtained according to Eq.~\eqref{eq: theta_hat} and sample survey data (right) with respect to the internal domain $\gamma^{(d)}, d = X_2$: Gender $ \in \{\texttt{Male, Female}\}$.}
    \label{fig: edu_counts}
\end{figure}

For imputing unit-level data for the $N-n = 281,705$ register units that are not part of the sample, a multinomial model has been fitted on the sample data using Subgroup B covariates, with the exception of $X_5: \texttt{province}$, which is treated as an external domain. Results on the estimated regression coefficients are reported in the Supplementary material D (Table 1). 
The estimates of population totals $\hat{\theta}^{(d)}_k = \sum_{i=1}^N \gamma_i^{(d)}\hat{Y}_{ik}, k=1,\dots,8$, in Figure~\ref{fig: edu_counts} (left panel) are obtained following a model-based approach with this working model and the deterministic estimator in Eq.~\eqref{eq: theta_hat}. Estimates for the population totals are consistent with the sample data (in terms of both ordering and relative size); in particular, the highest percentage of individuals are observed in categories \texttt{4 Lower secondary} ($28.4\%$) and \texttt{5 Upper secondary} ($38.4\%$), while a very limited proportion belongs to classes \texttt{1 Illiterate} ($0.35\%$) and \texttt{8 PhD level} ($0.35\%$). This has an important impact on the accuracy of the estimates, with a higher uncertainty for those classes characterised by a lower frequency. 

Accuracy estimates are reported in Table~\ref{tab: gmse_est}, comparing the 
proposed linearised estimator in Eq.~\eqref{eq: GMSE_Lin} against three 
alternative approaches: a non-parametric bootstrap with $B = 1,000$ resamples 
(pseudo-code available in Supplementary material C), a design- and a model-based approach, each based on $M = 1,000$ replicates. We recall that the model-based approach fixes the (observed) sample and replicates the outcome model, while the design-based approach fixes the (fitted) model and replicates the sampling design. Further details on their computation (and a comparison with the GMSE) are provided in Supplementary material C.1. 
Given the different size of each category $k$, as well as that of each different domain $d$, we consider a measure of relative accuracy, given by the coefficient of variation~\citep[CV; see e.g.,][] {lohr2021sampling}. Specifically, assuming that the category of interest is non-null, that is $\theta_k^{(d)} \neq 0$, this is given by:
\begin{align} \label{eq: cv}
    \CV\left(\hat{\theta}_k^{(d)}, \theta_k^{(d)}\right) = \frac{\sqrt{\GMSE\left(\hat{\theta}_k^{(d)}, \theta_k^{(d)}\right)}}{\Exp\left(\theta_k^{(d)}\right)},\quad k=1,\dots, K.
\end{align}
In practice, since both the numerator and denominator are unknown, a plug-in estimator $\widehat{\CV}_k = \widehat{\CV}\left(\hat{\theta}_k^{(d)}, \theta_k^{(d)}\right)$ is employed. The GMSE estimates are provided for completeness in the Supplementary material D (Table 2). 

\begin{table}[h]
\centering 
\caption{Estimates of totals $\hat{\theta}^{(d)}_k = \sum_{i=1}^N \gamma_i^{(d)}\hat{Y}_{ik}, k=1,\dots,8$ for the full register and for domain $d \in X_2: \text{Gender}$, with their estimated CV (\%) . The sample fraction $n_k^{(d)}/\hat{\theta}^{(d)}_k$ is between 3.9\% and 5.2\% across all cases. Bootstrap, as well as model-based and design-based estimates are based on a set of $1,000$ replicates.} \label{tab: gmse_est}
  \begin{tabular}{lrrrrrr}
    \toprule
      Category $k$ & $\hat{\theta}^{(d)}_k$ & \makecell{Sample \\ size $n_k^{(d)}$} 
       & $\widehat{\CV}^{\text{Lin}}_k$  &  $\widehat{\CV}^{\text{Boot}}_k$ &  $\widehat{\CV}^{\text{Model}}_k$&  $\widehat{\CV}^{\text{Design}}_k$\\
      \midrule
      Full register: $\gamma_i = 1,\quad i=1,\dots,N$\\
      \midrule
\texttt{1 Illiterate} &1039 &49 &11.86\% &13.38\%   &11.79\%  &12.53\% \\
\texttt{2 Literate but no education} &6649 &340 &4.70\%    &4.96\%      &4.73\%     &4.86\% \\
\texttt{3 Primary} &49886 &2572 &1.08\%    &1.11\%      &1.01\%     &1.07\% \\
\texttt{4 Lower secondary} &84174 &4285 &0.87\%    &0.92\%      &0.86\%     &0.87\% \\
\texttt{5 Upper secondary} &113719 &5682 &0.62\%    &0.67\%      &0.58\%     &0.62\% \\
\texttt{6 Bachelor degree} &7234 &364 &4.18\%    &4.17\%      &3.94\%     &3.96\%\\
\texttt{7 Master degree} &32810 &1524 &1.26\%    &1.33\%      &1.05\%     &1.11\% \\
\texttt{8 PhD level} &1054 &44 &12.02\%   &11.15\%      &8.83\%    &11.13\% \\
\midrule
      Internal domain: $\gamma^{(d)}, d = \texttt{Male}$ (47.7\%)\\
\midrule
\texttt{1 Illiterate} &300 &14 & 22.22\%  &24.66\% &18.99\%    &24.66\% \\
\texttt{2 Literate but no education} &1569 &81 &9.97\%   &10.69\% &9.94\%     &10.59\% \\
\texttt{3 Primary} &19631 &1015 & 1.72\%   &1.81\% &1.61\%     & 1.79\% \\
\texttt{4 Lower secondary} &45853 &2306 &1.12\%   &1.22\% &1.10\%      &1.08\% \\
\texttt{5 Upper secondary} &56374 &2775 &0.88\%   &0.94\% &0.85\%      &0.85\% \\
\texttt{6 Bachelor degree} &2701 &132 &7.06\%   &6.39\% &6.89\%      &6.72\% \\
\texttt{7 Master degree} &14443 &656 &1.91\%   &1.95\% &1.59\%      &1.76\% \\
\texttt{8 PhD level} &510 &20 &17.33\%   &15.18\% &11.65\%     &16.31\% \\
\midrule
      Internal domain: $\gamma^{(d)}, d = \texttt{Female}$ (52.3\%)\\
\midrule
\texttt{1 Illiterate} &739&35& 14.01\%   &15.55\% &14.60\%     &14.86\% \\
\texttt{2 Literate but no education} &5080&259 & 5.31\%   &5.46\%  &5.45\%      &5.52\% \\
\texttt{3 Primary} &30255 &1557  &1.37\%   &1.46\% &1.29\%      &1.40\% \\
\texttt{4 Lower secondary} &38321 &1979  &1.35\%   &1.44\%  &1.33\%     &1.35\% \\
\texttt{5 Upper secondary} &57345&2907  & 0.87\%   &0.94\% &0.84\%      &0.86\% \\
\texttt{6 Bachelor degree} &4533&232  &5.14\%   &5.27\% &4.77\%     &4.75\% \\
\texttt{7 Master degree} &18367&868 &1.66\%   &1.73\% &1.47\%      &1.54\% \\
\texttt{8 PhD level} &545 &24 &16.50\%   &15.62\% &13.17\%     &14.78\% \\
    \bottomrule
  \end{tabular}
\end{table}

As expected, all methods suggest a lower estimation error for those categories with higher frequency, and viceversa for smaller categories. In particular, the two extreme cases of \texttt{1 Illiterate} and \texttt{8 PhD level} suffer the most: both show an estimated CV between $8.83\%$ and $13.38\%$ when no domains are considered. When interest is in specific domains (e.g., gender), the lower-size subgroup is affected the most, especially when it comes with a sample unbalance with respect to some of the categories of interest. This occurs, for example, for the combination of domain $d = \texttt{Male}$ and category $k = \texttt{1 Illiterate}$, where estimates, being based only on 14 units in the sample (see Table~\ref{tab: gmse_est}, column 2), have an estimated CV ranging between $18.99\%$ and $24.66\%$. For the same category but domain $d = \texttt{Female}$, the number of sample units double (35), and the relative CV reduces to values between $14\%-15\%$. These values are aligned across all different accuracy estimation approaches, with a slight increase shown by the bootstrap approach, while the model-based approach is the one with the smallest estimate on average.

Depending on the granularity of the domain of interest, the error may incur an additional increase. However, the (sample/register) size of the category $k$ remains the primary driver, especially when the domain is external, that is, when this is not part of the model. An in-depth relationship between the estimated relative error, the sample category size (counts), and the domain size is given in Figure~\ref{fig: CV_province_sample} for the external domain of province, consisting in 9 items. Results are reported only for the proposed linearised estimator, given its centrality in this work. A further illustration is provided in Supplementary material D with respect to the estimated register counts, rather than the sample counts, showing similar results. For external domains, such as the province, it becomes clear that the accuracy is a function of the marginal size of the class, and not of its conditional size given the province $d$. The constant value of the estimated CV for a given class $k$ but for different provinces reflects this; see, e.g., the green points for $k = \texttt{5 Upper secondary}$ in Figure~\ref{fig: CV_province_sample}. That is certainly not the case for an internal domain, defined, e.g., by variable $X_2: \texttt{gender}$, where the conditional size has a direct impact on the error (see Supplementary material C for further details).
    
\begin{figure}[h]
    \centering
    \includegraphics[scale = .6]{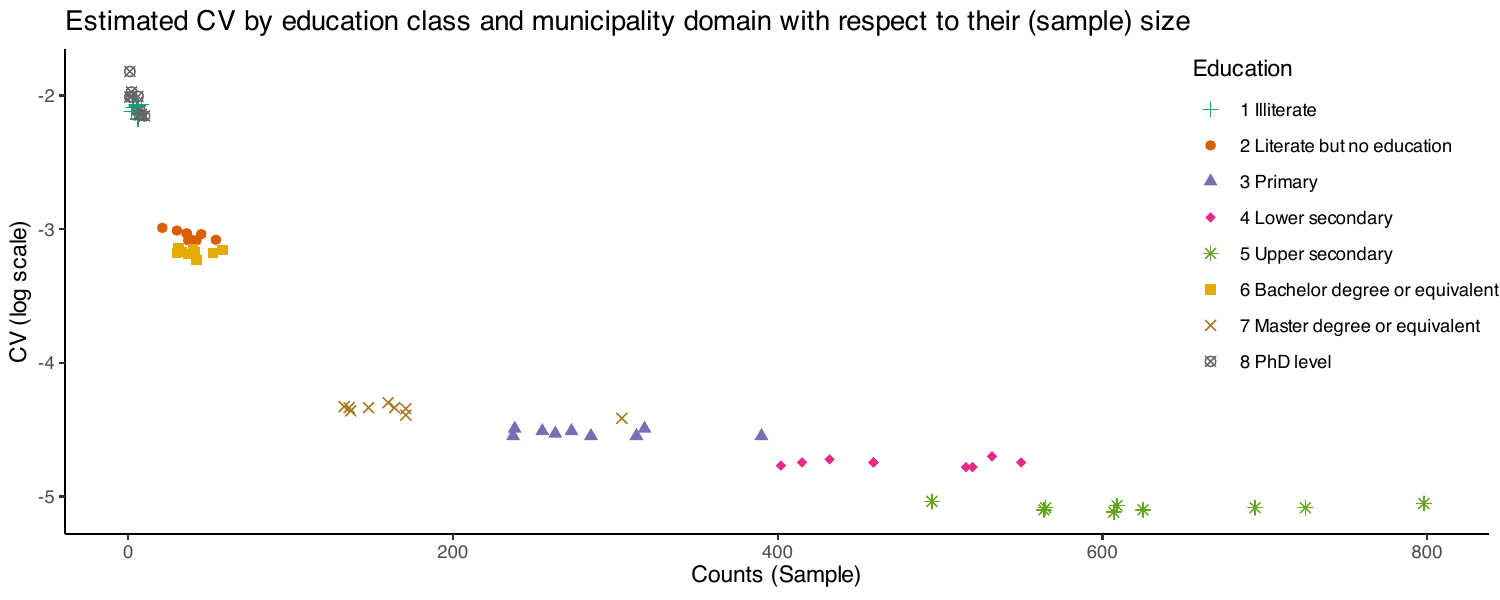}
    \caption{Coefficient of variation estimated with the linearised GMSE approach ($\widehat{\CV}^{\text{Lin}}_k$; in log scale) with respect to the sample size of each class $k = 1,\dots, 8$ and domain combination, with $d \in X_5$: \texttt{province} having $9$ modalities. }
    \label{fig: CV_province_sample}
\end{figure}

Being a real-world application, a more in-depth evaluation of the linearised 
and compared approaches will be provided in simulations (Section~\ref{sec: 
simulations}), with the introduction of the Monte Carlo (MC) benchmark estimator 
serving as ``the truth''. Nonetheless, it is already instructive to note that 
the linearised approach is the only one that {\it scales} to different domains of interest such as gender or province. As discussed in Section \ref{sec: scalability} (with pseudo-code details in Supplementary material C), the linearised approach relies on a \textit{single-time} computation of the plug-in estimators in 
Eq.~\eqref{eq: GMSE_Lin}, which are then reused for any domain or subpopulation 
of interest, making it a faster and more practical option for unplanned 
\textit{on-the-fly} estimates requested by end users. When specific domains are of interest, all other resampling-based methods will incur a computational burden that increases linearly with the number of domains. If there is a possibility to optimise the code so as to save and re-use the different replications over the domains of interest, this certainly comes with non-trivial memory burden, which may be infeasible for current systems. In terms of overall computational cost, for the full register with no domain stratification, the linearised approach requires less than 6 minutes, compared to approximately 54 minutes for the bootstrap with $B=1,000$ replicates on a standard laptop. These values result from a Darwin (macOS) Kernel Version 22.5.0; root:xnu-8796.121.3-7/RELEASE\_X86\_64, using the \texttt{system.time()} function in R version 4.3.2.
    
\section{Simulation Comparison}\label{sec: simulations}

To validate the proposed linearised GMSE estimation approach, we run a set of simulation studies with varying sample sizes $n = 0.05 \times N$, where $N \in \{100,000; 300,000; 500,000\}$. We follow a data-driven approach, guided by real data described in Section~\ref{sec: application}, which refer to the ISTAT register of the attained level of education (response variable) and include a set of covariates as described in Table~\ref{tab: covariates}. In our simulations, covariate data are re-generated independently from an infinite population specified by $\bm{x}_j \sim \text{Multinomial}(\bm{p}_{X_j}), \forall j$, with correspondent parameters given in Table~\ref{tab: covariates} (column 3). The sample membership indicator $\lambda_i$ is generated according to an SRS design without replacement, using the R function ``\texttt{sample($N, n$)}'', with $N$ the register units and $n$ the fixed sample size. 
Furthermore, we account for the internal domain variable $X_2 : \texttt{gender}$, and for an external one, given by $X_5 : \texttt{province}$ and generated as $\gamma_2 = \bm{x}_5 \sim \text{Multinomial}(\bm{p}_{X_5})$. If a variable $X_j$ is used as a model predictor it represents an internal domain; otherwise it defines an external one. The response variable is generated according to the multinomial model of interest, as described in Section~\ref{sec: gmse_mult}, with the ``true'' fixed coefficients $\bm{\beta} = (\bm{\beta}_1, \dots, \bm{\beta}_K)$ specified in Table 1 (Supplementary material D). These are obtained from fitting the same model to the real education data. To facilitate replicability of the results, both data and R codes are made publicly available at \url{https://github.com/nina-DL/GMSE}. 

Table~\ref{tab: gmse_est_sim} reports the estimated CVs for all approaches and population sizes. Here, the proposed GMSE estimation method can be compared to the MC estimator, serving as a benchmark for the true GMSE; the bootstrap, design-based, and model-based approaches considered in Section~\ref{sec: application} are also reported. Technical details on their computation are provided in Supplementary material C. The primary focus of the comparison is between $\widehat{\CV}^{\text{Lin}}_k$ and $\widehat{\CV}^{\text{MC}}_k$: the former is the proposed analytical approximation, the latter provides a numerical benchmark for the true GMSE under joint model and sampling randomness. Across all simulation scenarios, $\widehat{\CV}^{\text{Lin}}_k$ closely tracks the MC benchmark. Agreement is particularly strong for the most frequent categories (\texttt{3 Primary}, \texttt{4 Lower secondary}, \texttt{5 Upper secondary}, \texttt{7 Master degree}), where differences 
between the two are negligible and both correctly reflect the expected reduction in uncertainty as $N$ increases. The main discrepancies occur for the rarest categories (\texttt{1 Illiterate}, \texttt{8 PhD level}), where small effective sample sizes make first-order approximations more sensitive to higher-order terms. Even in these challenging cases, however, $\widehat{\CV}^{\text{Lin}}_k$ reproduces the systematic decline in uncertainty as $N$ increases. When compared to the other estimation approaches, the bootstrap shows a slight uncertainty overestimation, particularly for rare categories and smaller $N$, while model-based and design-based estimators are aligned with the benchmark. In general, as $N$ increases, the proposed linearised approach is the one that better matches the MC benchmark. 
\begin{table}[h]
\centering 
\caption{Estimates of the CV (in \%) with respect to the register totals 
estimates $\hat{\theta}^{(d)}_k = \sum_{i=1}^N \gamma_i^{(d)}\hat{Y}_{ik}, 
k=1,\dots,8$, with $\gamma_i = 1, i=1,\dots,N$, for $N \in \{100{,}000; 
300{,}000; 500{,}000\}$. Bootstrap, model-based, and design-based estimates 
are based on $1{,}000$ replicates; the MC estimator is based on $G = 100$ 
and $M = 100$ replicates.}\label{tab: gmse_est_sim}
\begin{tabular}{lrrrrrrr}
    \toprule
    Category $k$ & $\hat{\theta}^{(d)}_k$ & \makecell{Sample \\ size 
    $n_k^{(d)}$} & $\widehat{\CV}^{\text{Lin}}_k$ & 
    $\widehat{\CV}^{\text{Boot}}_k$ & 
    $\widehat{\CV}^{\text{Model}}_k$ &
    $\widehat{\CV}^{\text{Design}}_k$ & 
    $\widehat{\CV}^{\text{MC}}_k$\\
    \midrule
    \multicolumn{8}{l}{$N = 100{,}000;\quad n = 5{,}000$}\\
    \midrule
\texttt{1 Illiterate}                & 308   & 15   & 22.27\% & 26.56\% & 24.47\% & 24.23\% & 23.87\% \\
\texttt{2 Literate but no education} & 1886  & 81   & 10.19\% & 11.01\% & 10.89\% & 9.34\%  & 9.76\%  \\
\texttt{3 Primary}                   & 14537 & 724  & 2.46\%  & 2.64\%  & 2.45\%  & 2.34\%  & 2.53\%  \\
\texttt{4 Lower secondary}           & 30854 & 1517 & 1.51\%  & 1.56\%  & 1.52\%  & 1.46\%  & 1.55\%  \\
\texttt{5 Upper secondary}           & 38667 & 1929 & 1.06\%  & 1.13\%  & 1.09\%  & 1.04\%  & 1.09\%  \\
\texttt{6 Bachelor degree}           & 2183  & 107  & 7.26\%  & 7.60\%  & 7.28\%  & 7.47\%  & 7.81\%  \\
\texttt{7 Master degree}             & 11247 & 609  & 2.11\%  & 2.38\%  & 2.07\%  & 1.93\%  & 2.00\%  \\
\texttt{8 PhD level}                 & 317   & 12   & 20.59\% & 28.14\% & 23.13\% & 20.04\% & 19.73\% \\
    \midrule
    \multicolumn{8}{l}{$N = 300{,}000;\quad n = 15{,}000$}\\
    \midrule
\texttt{1 Illiterate}                & 949   & 58   & 15.67\% & 17.40\% & 16.31\% & 13.63\% & 14.12\% \\
\texttt{2 Literate but no education} & 5613  & 292  & 5.44\%  & 5.44\%  & 5.28\%  & 5.22\%  & 5.58\%  \\
\texttt{3 Primary}                   & 41411 & 2216 & 1.42\%  & 1.49\%  & 1.42\%  & 1.38\%  & 1.44\%  \\
\texttt{4 Lower secondary}           & 92610 & 4626 & 0.86\%  & 0.91\%  & 0.86\%  & 0.87\%  & 0.88\%  \\
\texttt{5 Upper secondary}           & 116326& 5892 & 0.65\%  & 0.67\%  & 0.65\%  & 0.62\%  & 0.62\%  \\
\texttt{6 Bachelor degree}           & 6464  & 296  & 4.31\%  & 4.47\%  & 4.45\%  & 4.25\%  & 4.41\%  \\
\texttt{7 Master degree}             & 33711 & 1704 & 1.15\%  & 1.22\%  & 1.17\%  & 1.10\%  & 1.16\%  \\
\texttt{8 PhD level}                 & 915   & 42   & 9.15\%  & 9.86\%  & 9.07\%  & 10.52\% & 10.90\% \\
    \midrule
    \multicolumn{8}{l}{$N = 500{,}000;\quad n = 25{,}000$}\\
    \midrule
\texttt{1 Illiterate}                & 1581  & 80   & 10.28\% & 10.72\% & 10.34\% & 10.14\% & 10.28\% \\
\texttt{2 Literate but no education} & 9375  & 469  & 4.05\%  & 3.98\%  & 3.98\%  & 3.98\%  & 4.26\%  \\
\texttt{3 Primary}                   & 72608 & 3656 & 1.12\%  & 1.15\%  & 1.09\%  & 1.09\%  & 1.14\%  \\
\texttt{4 Lower secondary}           & 154465& 7763 & 0.68\%  & 0.72\%  & 0.70\%  & 0.66\%  & 0.67\%  \\
\texttt{5 Upper secondary}           & 193322& 9680 & 0.48\%  & 0.51\%  & 0.48\%  & 0.46\%  & 0.50\%  \\
\texttt{6 Bachelor degree}           & 10936 & 552  & 3.47\%  & 3.61\%  & 3.50\%  & 3.19\%  & 3.43\%  \\
\texttt{7 Master degree}             & 56129 & 2880 & 0.89\%  & 1.00\%  & 0.90\%  & 0.87\%  & 0.89\%  \\
\texttt{8 PhD level}                 & 1583  & 79   & 8.62\%  & 9.02\%  & 8.66\%  & 8.75\%  & 8.66\%  \\
    \bottomrule
\end{tabular}
\end{table}

To further assess the inferential validity of the proposed estimator, Table~\ref{tab: coverage_length} reports empirical coverage rates and average lengths of $95\%$ confidence intervals, evaluated 
over $M = 100$ simulation replicates for $N = 100{,}000$. Coverage rates range from $0.93$ to $0.97$ across all eight education categories, 
at or above the nominal $0.95$ level for seven out of eight 
categories, with the slight undercoverage for 
\texttt{2 Literate but no education} ($0.93$) remaining close to nominal, and consistent with the asymptotic nature of the first-order linearisation (as well as the limited register size of the considered scenario). Average interval lengths scale appropriately with category size and associated uncertainty. A further assessment of the variability of the linearised estimator in this scenario is provided in the form of boxplots in 
Figure~4 of the Supplementary material D. 
\begin{table}[h]
\centering
\caption{Empirical coverage rates and average lengths of $95\%$ confidence intervals for the $\widehat{\text{GMSE}}^{\text{Lin}}_k$ estimator 
for $N = 100{,}000$. Results are evaluated over $M = 100$ simulation replicates.}\label{tab: coverage_length}
\begin{tabular}{lrr}
    \toprule
    Category $k$ & 
    Coverage Rate & 
    Average Length \\
    \midrule
\texttt{1 Illiterate}                & 0.97 & 324.36 \\
\texttt{2 Literate but no education} & 0.93 & 683.40 \\
\texttt{3 Primary}                   & 0.96 & 1396.74 \\
\texttt{4 Lower secondary}           & 0.96 & 1828.28 \\
\texttt{5 Upper secondary}           & 0.96 & 1673.04 \\
\texttt{6 Bachelor degree}           & 0.95 & 632.82 \\
\texttt{7 Master degree}             & 0.97 & 872.55 \\
\texttt{8 PhD level}                 & 0.96 & 265.56 \\
    \bottomrule
\end{tabular}
\end{table}

In conclusion, the simulation results provide strong empirical support 
for the proposed linearisation approach, demonstrating its ability to accurately reproduce the 
$\GMSE\left(\hat{\theta}_k^{(d)}, \theta_k^{(d)}\right)$ target under joint model and sampling randomness, which represents the primary inferential target of the proposed framework rather than a comparison with alternative estimators of sampling/model variance alone.

\section{Discussion and Conclusion} \label{Sec: conc}

In this work, we have proposed a global measure of accuracy, the GMSE, which results in a reliable, interpretable, computationally feasible, and flexible approach. This framework is designed to be integrated into modern NSI production pipelines. Concretely, it supports a systematic internal quality assessment by providing ex-post measures of accuracy for register-based outputs, complements existing validation procedures, and reduces reliance on computationally intensive resampling methods. At the same time, it enables user-driven production, where users can generate domain-specific statistics (typically unplanned and unknown in advance) on demand, together with corresponding measures of uncertainty. Finally, the same quantities can be used for dissemination purposes, ensuring that published statistics are accompanied by transparent accuracy indicators. The methodology is therefore intended not only for methodological research, but as a practical and scalable component of operational register-based statistical systems.

The GMSE has notable advantages. The proposed measure results in a more comprehensive and highly transparent solution as it explicitly considers the uncertainty deriving from various factors that influence the inferential process used to construct the estimates from the register. 
In terms of its computational aspects, the GMSE is user-friendly and allows for on-the-fly release of statistical estimates along with the associated accuracy measurements. Once the estimates of the unknown model parameters are obtained, the computation of the GMSE is straightforward, easy to use and communicate to users. This makes it an ideal tool for safeguarding impartiality and equal access, while adhering to a fundamental principle of official statistics that asserts that these \emph{``have to be available to all users, i.e. they have to be public.''}~\citep{bruengger2008should}.

Several extensions can broaden the applicability of the framework. First, although this paper focuses on two central sources of uncertainties that characterise all modern integrated systems of data (sampling uncertainty and model uncertainty), the GMSE is conceived as a modular quantity that can incorporate further sources of uncertainty whenever suitable analytical representations are available. Examples include: (1) record linkage errors, where false matches and missed links introduce additional variability, which could be represented by a random link indicator analogously to $\lambda_i$; (2) coverage errors, which deviate from the ideal one-to-one register--population correspondence assumed in Section~2; or (3) measurement errors, that may inflate the Hessian-based term $\mathbf{A}_{\bm{\beta}}^{-1}$. Second, the GMSE computational strategy for multi-category outcomes can be extended beyond multinomial logistic models. As detailed in Section~3, the two-step linearisation in Section~3 is formulated under general estimating equations
in Eq.~\eqref{eq: gee}, and the multinomial logistic model enters only as a specific instantiation. Any model admitting an $M$-estimator with a well-defined, non-singular Hessian and analytic or numerical expressions for $\partial f_k / \partial\bm{\beta}$ can be embedded within the same framework. This includes ordinal models, as well as random-effects multinomial models, where an additional variance component arising from the random effects can be incorporated via a first-order approximation. 
Third, the analytical GMSE in Eq.~(\ref{eq: GMSE_final}) is expressed in terms of general inclusion probabilities $\pi_i$, making it directly applicable to a wide class of non-informative
designs, including simple random sampling, stratified sampling, and
probability-proportional-to-size schemes~\citep{tille2020sampling}. While the assumption of non-informative
sampling, whereby the sampling mechanism does not depend on the response variable, is standard in official statistics production, certain modern data collection strategies, such as snowball
sampling~\citep{alleva2025building}, adaptive sampling designs~\citep{thompson2006adaptive}, or web-scraped administrative sources, may induce
informative selection that introduces a dependence between $\bm{\lambda}$ and $\bm{Y}$ not
captured by the current GMSE framework. In such settings, the bound in Theorem~\ref{th: gmse_th1} would
need to be re-examined and adjustments to the estimating equations or weighting schemes
would be required. The extension of the GMSE to informative sampling designs is left for
future research.

All such extension possibilities highlight the considerable potential of the GMSE approach presented here. At the same time, they also point to what is likely the main barrier to its widespread adoption: the need for methodological work to derive suitable linearised expressions for the predictive model used in production. This is a well-known challenge of analytical, Taylor-based approaches to accuracy estimation. In contrast with resampling methods, which are often more automatic but computationally intensive, linearisation methods require an initial investment in statistical modelling and derivation. 
This issue is particularly relevant when predictions are produced without an explicitly stated statistical model. Although such situations may arise in practice, making the predictive mechanism explicit is itself beneficial, as it improves transparency, reproducibility, and accountability in the statistical processes used to construct modern registers. In this sense, the GMSE framework may also encourage clearer documentation of the modelling assumptions underlying register-based outputs. We also note that our proposed method can be extended to predictions generated by artificial intelligence (AI). For instance, Explainable AI methods can locally approximate complex prediction functions (or ``black-box'' models) through simpler and interpretable surrogate models, thereby providing a possible route to obtain derivative-based approximations compatible with the GMSE framework~\citep{confalonieri2021historical}. Nevertheless, this area still requires substantial theoretical development and empirical validation, particularly regarding the stability of local approximations and the interpretation of resulting uncertainty measures. In the case of model-free predictions, another promising avenue is conformal prediction, whose finite-sample validity under exchangeability offers an attractive basis for uncertainty quantification~\citep{vovk2005algorithmic}. Exploring the relationship between conformal uncertainty sets and generalised mean squared error measures is an interesting direction for future research.

The GMSE could represent a useful common measure of global accuracy across the main methodological traditions within NSIs, including those who base their inferences solely on statistical models and those who primarily rely on sampling design and use models as supplementary tools. Bayesian statisticians could also incorporate the same logic by integrating an expectation operator reflecting the \textit{a-priori} variability of model parameters.

In conclusion, this work aims to support the advancement of current practices in official statistics by facilitating a flexible yet rigorous use of register data, where the production of estimates is systematically complemented by coherent and informative measures of quality.

\bibliographystyle{abbrvnat}
\bibliography{main.bib}

\begin{thebibliography}{45}
\providecommand{\natexlab}[1]{#1}
\providecommand{\url}[1]{\texttt{#1}}
\expandafter\ifx\csname urlstyle\endcsname\relax
  \providecommand{\doi}[1]{doi: #1}\else
  \providecommand{\doi}{doi: \begingroup \urlstyle{rm}\Url}\fi

\bibitem[Alleva(2017{\natexlab{a}})]{alleva2017}
G.~Alleva.
\newblock {The new role of sample surveys in official statistics}.
\newblock Technical report, {ITACOSM 2017: The 5th Italian Conference on Survey Methodology}, 2017{\natexlab{a}}.
\newblock URL \url{https://www.istat.it/it/files//2015/10/Alleva_ITACOSM_14062017.pdf}.

\bibitem[Alleva(2017{\natexlab{b}})]{alleva2017b}
G.~Alleva.
\newblock {Emerging challenges in official statistics: new sources, methods and skills}.
\newblock Technical report, {SIS2017 Statistical Conference. Statistics and Data Science: new challenges, new generations}, 2017{\natexlab{b}}.
\newblock URL \url{https://www.istat.it/wp-content/uploads/2024/01/Alleva_SIS2017_28062017.pdf}.

\bibitem[Alleva et~al.(2021)Alleva, Falorsi, Petrarca, and Righi]{alleva2021measuring}
G.~Alleva, P.~D. Falorsi, F.~Petrarca, and P.~Righi.
\newblock Measuring the accuracy of aggregates computed from a statistical register.
\newblock \emph{Journal of Official Statistics}, 37\penalty0 (2):\penalty0 481--503, 2021.

\bibitem[Alleva et~al.(2025)Alleva, Falorsi, Falorsi, Fasulo, and Righi]{alleva2025building}
G.~Alleva, P.~D. Falorsi, S.~Falorsi, A.~Fasulo, and P.~Righi.
\newblock Building estimates for totals in respondent driven sampling.
\newblock \emph{Statistical Journal of the IAOS}, 41\penalty0 (1):\penalty0 50--61, 2025.

\bibitem[Biemer et~al.(2014)Biemer, Trewin, Bergdahl, and Japec]{biemer2014system}
P.~Biemer, D.~Trewin, H.~Bergdahl, and L.~Japec.
\newblock A system for managing the quality of official statistics.
\newblock \emph{Journal of Official Statistics}, 30\penalty0 (3):\penalty0 381--415, 2014.

\bibitem[Biemer(2010)]{biemer_total_2010}
P.~P. Biemer.
\newblock Total {Survey} {Error}: {Design}, {Implementation}, and {Evaluation}.
\newblock \emph{Public Opinion Quarterly}, 74\penalty0 (5):\penalty0 817--848, Jan. 2010.
\newblock ISSN 0033-362X, 1537-5331.
\newblock \doi{10.1093/poq/nfq058}.
\newblock URL \url{https://academic.oup.com/poq/article-lookup/doi/10.1093/poq/nfq058}.

\bibitem[Bruengger(2008)]{bruengger2008should}
H.~Bruengger.
\newblock {How Should a Modern National System of Official Statistics Look?}
\newblock \emph{UNECE, Statistical Division}, 2008.

\bibitem[Bycroft and Matheson-Dunning(2020)]{Bycroft2020Use}
C.~Bycroft and N.~Matheson-Dunning.
\newblock {Use of administrative records for non-response in the New Zealand 2018 Census}.
\newblock \emph{Statistical Journal of the IAOS}, 36\penalty0 (1):\penalty0 107--115, 2020.

\bibitem[Chambers and Clark(2012)]{chambers2012introduction}
R.~L. Chambers and R.~Clark.
\newblock \emph{An introduction to model-based survey sampling with applications}.
\newblock Oxford University Press, 2012.

\bibitem[Citro(2014)]{citro2014multiple}
C.~F. Citro.
\newblock From multiple modes for surveys to multiple data sources for estimates.
\newblock \emph{Survey Methodology}, 40\penalty0 (2):\penalty0 137--162, 2014.

\bibitem[Confalonieri et~al.(2021)Confalonieri, Coba, Wagner, and Besold]{confalonieri2021historical}
R.~Confalonieri, L.~Coba, B.~Wagner, and T.~R. Besold.
\newblock A historical perspective of explainable artificial intelligence.
\newblock \emph{Wiley Interdisciplinary Reviews: Data Mining and Knowledge Discovery}, 11\penalty0 (1):\penalty0 e1391, 2021.

\bibitem[Daalmans(2017)]{daalmans2017mass}
J.~Daalmans.
\newblock \emph{Mass imputation for census estimation}.
\newblock Statistics Netherlands, 2017.

\bibitem[Daas et~al.(2011)Daas, Ossen, Tennekes, Zhang, Hendriks, Foldal~Haugen, Cerroni, Di~Bella, Laitila, Wallgren, et~al.]{daas2011report}
P.~Daas, S.~Ossen, M.~Tennekes, L.-C. Zhang, C.~Hendriks, K.~Foldal~Haugen, F.~Cerroni, G.~Di~Bella, T.~Laitila, A.~Wallgren, et~al.
\newblock Report on methods preferred for the quality indicators of administrative data sources.
\newblock Technical report, {Blue - ETS Project, Deliverable 4.2.}, 2011.
\newblock URL \url{http://www.pietdaas.nl/beta/pubs/pubs/BLUE-ETS_WP4_Del2.pdf}.

\bibitem[Daas et~al.(2015)Daas, Puts, Buelens, and van~den Hurk]{daas2015big}
P.~J. Daas, M.~J. Puts, B.~Buelens, and P.~A. van~den Hurk.
\newblock Big data as a source for official statistics.
\newblock \emph{Journal of Official Statistics}, 31\penalty0 (2):\penalty0 249--262, 2015.

\bibitem[De~Waal et~al.(2020)De~Waal, van Delden, and Scholtus]{de2020multi}
T.~De~Waal, A.~van Delden, and S.~Scholtus.
\newblock Multi-source statistics: basic situations and methods.
\newblock \emph{International Statistical Review}, 88\penalty0 (1):\penalty0 203--228, 2020.

\bibitem[Di~Zio et~al.(2019{\natexlab{a}})Di~Zio, Filippini, and Rocchetti]{di2019imputation}
M.~Di~Zio, R.~Filippini, and G.~Rocchetti.
\newblock An imputation procedure for the italian attained level of education in the register of individuals based on administrative and survey data.
\newblock \emph{Measuring well-being at local level using remote sensing and official statistics data}, 2019{\natexlab{a}}.

\bibitem[Di~Zio et~al.(2019{\natexlab{b}})Di~Zio, Filippini, and Rocchetti]{dizio1}
M.~Di~Zio, R.~Filippini, and G.~Rocchetti.
\newblock An imputation procedure for the italian attained level of education in the register of individuals based on administrative and survey data.
\newblock \emph{Rivista di Statistica Ufficiale}, Issue 2-3\penalty0 (2019):\penalty0 143--174, 2019{\natexlab{b}}.

\bibitem[{European Statistical System Committee}(2018)]{european2018bucharest}
{European Statistical System Committee}.
\newblock Bucharest memorandum on official statistics in a datafied society (trusted smart statistics).
\newblock \emph{Luxembourg City (Luxembourg): Eurostat}, 2018.

\bibitem[Eurostat(2014)]{eurostat2014}
Eurostat.
\newblock {ESS handbook for quality reports}.
\newblock Technical report, {European Statistical System}, 2014.
\newblock URL \url{https://ec.europa.eu/eurostat/documents/3859598/6651706/KS-GQ-15-003-EN-N.pdf}.

\bibitem[Eurostat(2020)]{eurostat2020}
Eurostat.
\newblock {European Statistical System handbook for quality and metadata reports}.
\newblock Technical report, {Luxembourg: Publications Office of the European Union}, 2020.
\newblock URL \url{https://ec.europa.eu/eurostat/documents/3859598/10501168/KS-GQ-19-006-EN-N.pdf}.

\bibitem[Gjaltema(2022)]{gjaltema2022high}
T.~Gjaltema.
\newblock High-level group for the modernisation of official statistics (hlg-mos) of the united nations economic commission for europe.
\newblock \emph{Statistical Journal of the IAOS}, 38\penalty0 (3):\penalty0 917--922, 2022.

\bibitem[Graf and Till{\'e}(2014)]{graf2014variance}
E.~Graf and Y.~Till{\'e}.
\newblock Variance estimation using linearization for poverty and social exclusion indicators.
\newblock \emph{Survey Methodology}, 40\penalty0 (1):\penalty0 61--80, 2014.

\bibitem[Henderson et~al.(1983)Henderson, Pukelsheim, and Searle]{henderson1983history}
H.~V. Henderson, F.~Pukelsheim, and S.~R. Searle.
\newblock {On the history of the Kronecker product}.
\newblock \emph{Linear and Multilinear Algebra}, 14\penalty0 (2):\penalty0 113--120, 1983.

\bibitem[Isaki and Fuller(1982)]{isaki_survey_1982}
C.~T. Isaki and W.~A. Fuller.
\newblock Survey {Design} under the {Regression} {Superpopulation} {Model}.
\newblock \emph{Journal of the American Statistical Association}, 77\penalty0 (377):\penalty0 89--96, Mar. 1982.
\newblock ISSN 0162-1459, 1537-274X.
\newblock \doi{10.1080/01621459.1982.10477770}.
\newblock URL \url{http://www.tandfonline.com/doi/abs/10.1080/01621459.1982.10477770}.

\bibitem[{Istat}(2016)]{istat2016_modern}
{Istat}.
\newblock {ISTAT's Modernisation Programme}.
\newblock Technical report, {ISTAT: Istituto Nazionale di Statistica}, 2016.
\newblock URL \url{https://www.istat.it/wp-content/uploads/2011/04/IstatsModernistionProgramme_EN.pdf}.

\bibitem[Kim and Rao(2009)]{kim2009unified}
J.~K. Kim and J.~Rao.
\newblock A unified approach to linearization variance estimation from survey data after imputation for item nonresponse.
\newblock \emph{Biometrika}, 96\penalty0 (4):\penalty0 917--932, 2009.

\bibitem[Lohr(2021)]{lohr2021sampling}
S.~L. Lohr.
\newblock \emph{Sampling: design and analysis}.
\newblock Chapman and Hall/CRC, 2021.

\bibitem[Lohr and Raghunathan(2017)]{lohr_combining_2017}
S.~L. Lohr and T.~E. Raghunathan.
\newblock Combining {Survey} {Data} with {Other} {Data} {Sources}.
\newblock \emph{Statistical Science}, 32\penalty0 (2):\penalty0 293--312, 2017.
\newblock ISSN 0883-4237.

\bibitem[Lundy(2022)]{Lundy2022predicting}
E.~R. Lundy.
\newblock {Predicting the quality and evaluating the use of administrative data for the 2021 Canadian Census of Population}.
\newblock \emph{Statistical Journal of the IAOS}, 38\penalty0 (4):\penalty0 1177--1183, 2022.

\bibitem[Mashreghi et~al.(2016)Mashreghi, Haziza, and L{\'e}ger]{mashreghi2016survey}
Z.~Mashreghi, D.~Haziza, and C.~L{\'e}ger.
\newblock A survey of bootstrap methods in finite population sampling.
\newblock \emph{Statistics Surveys}, 10:\penalty0 1--52, 2016.

\bibitem[Nedyalkova and Till{\'e}(2008)]{nedyalkova_optimal_2008}
D.~Nedyalkova and Y.~Till{\'e}.
\newblock Optimal sampling and estimation strategies under the linear model.
\newblock \emph{Biometrika}, 95\penalty0 (3):\penalty0 521--537, Sept. 2008.
\newblock ISSN 0006-3444, 1464-3510.
\newblock \doi{10.1093/biomet/asn027}.
\newblock URL \url{https://academic.oup.com/biomet/article-lookup/doi/10.1093/biomet/asn027}.

\bibitem[Radermacher(2018)]{radermacher2018official}
W.~J. Radermacher.
\newblock Official statistics in the era of big data opportunities and threats.
\newblock \emph{International Journal of Data Science and Analytics}, 6\penalty0 (3):\penalty0 225--231, 2018.

\bibitem[S{\"a}rndal et~al.(2003)S{\"a}rndal, Swensson, and Wretman]{sarndal2003model}
C.-E. S{\"a}rndal, B.~Swensson, and J.~Wretman.
\newblock \emph{Model assisted survey sampling}.
\newblock Springer Science \& Business Media, 2003.

\bibitem[Scholtus and Daalmans(2021)]{Scholtusetal}
S.~Scholtus and J.~Daalmans.
\newblock Variance estimation after mass imputation based on combined administrative and survey data.
\newblock \emph{Journal of Official Statistics}, 37\penalty0 (2):\penalty0 433--459, 2021.

\bibitem[Shao(2003)]{shao2003impact}
J.~Shao.
\newblock Impact of the bootstrap on sample surveys.
\newblock \emph{Statistical Science}, 18\penalty0 (2):\penalty0 191--198, 2003.

\bibitem[Thompson(2006)]{thompson2006adaptive}
S.~K. Thompson.
\newblock Adaptive web sampling.
\newblock \emph{Biometrics}, 62\penalty0 (4):\penalty0 1224--1234, 2006.

\bibitem[Till{\'e}(2020)]{tille2020sampling}
Y.~Till{\'e}.
\newblock \emph{Sampling and estimation from finite populations}.
\newblock John Wiley \& Sons, 2020.

\bibitem[Vall{\'e}e and Till{\'e}(2019)]{vallee2019linearisation}
A.-A. Vall{\'e}e and Y.~Till{\'e}.
\newblock Linearisation for variance estimation by means of sampling indicators: Application to non-response.
\newblock \emph{International Statistical Review}, 87\penalty0 (2):\penalty0 347--367, 2019.

\bibitem[Valliant(2009)]{valliant2009model}
R.~Valliant.
\newblock Model-based prediction of finite population totals.
\newblock In \emph{Handbook of Statistics}, volume~29, pages 11--31. Elsevier, 2009.

\bibitem[Venables and Ripley(2013)]{venables2013modern}
W.~N. Venables and B.~D. Ripley.
\newblock \emph{Modern applied statistics with S}.
\newblock Springer Science \& Business Media, 2013.

\bibitem[Vovk et~al.(2005)Vovk, Gammerman, and Shafer]{vovk2005algorithmic}
V.~Vovk, A.~Gammerman, and G.~Shafer.
\newblock \emph{Algorithmic learning in a random world}, volume~29.
\newblock Springer, 2005.

\bibitem[Wallgren and Wallgren(2014)]{wallgren2014register}
A.~Wallgren and B.~Wallgren.
\newblock \emph{Register-based Statistics: statistical methods for administrative data}.
\newblock John Wiley \& Sons, 2014.

\bibitem[Wolter(1986)]{wolter_coverage_1986}
K.~M. Wolter.
\newblock Some {Coverage} {Error} {Models} for {Census} {Data}.
\newblock \emph{Journal of the American Statistical Association}, 81\penalty0 (394):\penalty0 338, June 1986.
\newblock ISSN 01621459.
\newblock \doi{10.2307/2289222}.
\newblock URL \url{https://www.jstor.org/stable/2289222?origin=crossref}.

\bibitem[Zhang(2012)]{zhang2012topics}
L.-C. Zhang.
\newblock Topics of statistical theory for register-based statistics and data integration.
\newblock \emph{Statistica Neerlandica}, 66\penalty0 (1):\penalty0 41--63, 2012.

\bibitem[Ziegler(2011)]{ziegler2011generalized}
A.~Ziegler.
\newblock \emph{Generalized estimating equations}, volume 204.
\newblock Springer Science \& Business Media, 2011.

\end{thebibliography}

\section*{Acknowledgments}
The authors express their sincere gratitude to ISTAT for their continuous support, with special thanks to Romina Filippini and Simona Toti for their assistance with the data. Appreciation is also extended to Brunero Liseo and Marco Di Zio for their valuable discussions during the development of this work and for their insightful feedback. In accordance with COPE guidelines and the TITAN Guideline Checklist 2025, we acknowledge the use of Claude for MAC v1.9659.2 for linguistic editing. No content generation or interpretation was delegated to the tool. The authors remain fully responsible for the manuscript. No artificial intelligence tools have been used for image generation, data collection, coding, or data analysis.

\newpage
\section*{\centering Supplementary Material}
\appendix

\section{Innovation in the current official-statistics framework} \label{app: intro}

It is important to underline that the community of national and international agencies producing official statistics are well aware of the challenge they face in the integration of new data sources in the different phases of production. This is demonstrated by the Experimental Statistics~\footnote{\url{https://www.istat.it/en/announcement-and-analisys/experimental-statistics/}} which since some years represent a production line of many national statistical institutions (NSIs) -- also allowing the collection of user feedback -- and by international projects to experiment methodologies and computational solutions for the production of statistics based on the integration of multiple sources, including Big Data~\citep{daas2015big}. Like administrative data, Big Data are generated outside the control of NSIs, and in analogy with what has been developed and now consolidated on administrative data, various international projects have focused on the definition of quality dimensions to be considered for Big Data as input to a statistical process. 
The awareness of the NSIs of having to move in a new data ecosystem and the need for commitment to a deep innovation of production processes guided the revision of the European Statistics Code of Practice in 2017~\citep{eurostat2020}.

First, the Principle 12 on \textit{Data Accuracy and Reliability}, consistently with the processes of integration of different sources that today characterize the production of official statistics, states that Source data, integrated data, intermediate results and statistical outputs are regularly assessed and validated and that Systems for assessing and validating source data, integrated data, intermediate results and statistical outputs are in place. 

Second, in the Principle 7 on \textit{Sound Methodology} of the data production processes, the Code states that Statistical authorities actively encourage the exploration of new and innovative methods for statistics. They develop methodological work and supporting IT solutions to ensure the quality of statistics, especially when new and alternative data collection modes and sources are used as input and statistical authorities take initiatives and participate in the development of innovative methods for collecting and processing data including the integration of new and/or alternative data sources and geospatial data.

Along this direction, the Bucharest Memorandum on \textit{Official Statistics in a datafied society -- Trusted Smart Statistics}~\citep{european2018bucharest} welcomes and embraces the opportunity for official statistics in a datafied society and economy, and encourages the European Statistical System (ESS) to implement practical and mature cases of using 'big data-enhanced'~\footnote{The Istat modernisation project is consistent with the ESS Vision 2020, the Transformative Agenda for Official Statistics~\citep{eurostat2020} and the guidelines of the High Level Group on Modernisation of the production and dissemination of Official Statistics~\citep{gjaltema2022high}} statistical products and develop experimental statistics on new phenomena. In particular, agree that the variety of new data sources, computational paradigms and tools will require amendments to the statistical business architecture, processes, production models, IT infrastructures, methodological and quality frameworks, and the corresponding governance structures, and therefore invite the ESS to formally outline and assess such amendments. Therefore, the introduction of Trusted Smart Statistics is changing the approach to the statistical use of Big Data, combining traditional methods of official statistical production with a different paradigm.

\section{Proofs of the technical results} \label{app: proofs}

\subsection{Proof of Theorem 1} \label{app: th1}

\begin{customthm}{1} \label{app_th: th1}
Assume that: i) the estimator $\hat{\bm{\beta}}$ 
is model-design unbiased, that is, $\Exp_{(\D, \M)}(\hat{\bm{\beta}}) = \Exp_{\D}\Exp_{\M}(\hat{\bm{\beta}}\mid \bm{\lambda}) = \bm{\beta}$, (ii) the chance of drawing a sample for which $\Exp_{\M}(\hat{\bm{\beta}} | \bm{\lambda}) \neq \bm{\beta}$ is negligibly small with $\sup_{\bm{\lambda}}\norm{\Exp_{\M}(\hat{\bm{\beta}} | \bm{\lambda}) - \bm \beta} < \infty$, and iii) the sampling design is non-informative \citep[see e.g., Section 1.4 of][for more details on non-informative sampling]{chambers2012introduction}. Then, the following upper bounds hold for the individual-category and the cumulated GMSE:
\begin{align} \label{eq: gmse_th1}
    \GMSE\left(\hat{\theta}_k^{(d)}, \theta_k^{(d)}\right) &\lessapprox \Exp_{\D}\Var_{\M}(\hat{\theta}_k^{(d)}| \bm{\lambda})\\
    \GMSE\left(\bm{\hat{\theta}}^{(d)}, \bm{\theta}^{(d)}\right) &\lessapprox \Exp_{\D}\left[\mathrm{tr}\left(\Var_{\M}(\bm{\hat{\theta}}^{(d)}| \bm{\lambda})\right)\right], \nonumber
\end{align}
where $\mathrm{tr}(\mathrm{x})$ denotes the trace of a square matrix $\mathrm{x}$ and $\lessapprox$ stands for ``less than or approximately equal to''.
\end{customthm}
Before proving Theorem~\ref{app_th: th1}, we note that condition iii) holds for most of the surveys conducted by national statistical institutes and it postulates that the sample design is such that the sample and population models coincide. Operationally, it allows for the interchange of the order of expectations $\Exp_{\D}$ and $\Exp_{\M}$. 

To derive the approximated expression of the $\GMSE\left(\hat{\theta}_k^{(d)}, \theta_k^{(d)}\right) = \Exp_{(\D, \M)}\left(\hat{\theta}_k^{(d)} -\theta_k^{(d)}\right)^2$, we add and subtract the fixed quantity $\tilde{\theta}_k^{(d)} = \Exp_{(\D, \M)}\left(\hat{\theta}_k^{(d)}\right)$ inside the argument of the expectation, obtaining
\begin{align*}
    \GMSE\left(\hat{\theta}_k^{(d)}, \theta_k^{(d)}\right) &
    = \Exp_{(\D, \M)}\left(\hat{\theta}_k^{(d)} \pm \tilde{\theta}_k^{(d)} -\theta_k^{(d)}\right)^2\\
    & = \underbrace{\Exp_{(\D, \M)}\left(\hat{\theta}_k^{(d)} - \tilde{\theta}_k^{(d)}\right)^2}_{\text{Part\ I}} + \underbrace{\Exp_{(\D, \M)}\left(\tilde{\theta}_k^{(d)} -\theta_k^{(d)}\right)^2}_{\text{Part\ II}} + 2\underbrace{\Exp_{(\D, \M)}\left((\hat{\theta}_k^{(d)} - \tilde{\theta}_k^{(d)})(\tilde{\theta}_k^{(d)} -\theta_k^{(d)})\right)}_{\text{Part\ III}}.
\end{align*}


\paragraph{Part I.} By 
applying the law of total variance we get
\begin{align*}
    \Exp_{(\D, \M)}\left(\hat{\theta}_k^{(d)} - \tilde{\theta}_k^{(d)}\right)^2 &= \Exp_{(\D, \M)}\left(\hat{\theta}_k^{(d)} - \Exp_{(\D,\M)}\left(\hat{\theta}_k^{(d)}\right)\right)^2 = \Var_{(\D, \M)}\left(\hat{\theta}_k^{(d)}\right)\\
    & = \Exp_{\D} \Var_{\M}\left(\hat{\theta}_k^{(d)}|\bm{\lambda}\right) + \Var_{\D} \Exp_{\M}\left(\hat{\theta}_k^{(d)}|\bm{\lambda}\right)\\
    & \approx \Exp_{\D} \Var_{\M}\left(\hat{\theta}_k^{(d)}|\bm{\lambda}\right),
\end{align*}
given that, from condition ii), $\Var_{\D} \Exp_{\M}\left(\hat{\theta}_k^{(d)}|\bm{\lambda}\right) \approx 0$ for negligible values of $\varepsilon$.

To see it, recall that, using the Taylor approximation, we have that,
\begin{align*} 
    \hat{\theta}_k^{(d)} = \sum_{i=1}^N \gamma_{i}^{(d)}f_k(\bm{x}_i;\bm{\hat{\beta}}) &= \sum_{i=1}^N \gamma_{i}^{(d)}\left( f_k(\bm{x}_i;\bm{\beta}) + 
   \frac{\partial f_k(\bm{x}_i;\bm{\hat{\beta}})}{\partial \bm{\hat{\beta}}} \Big\lvert_{\bm{\hat{\beta}} = \bm{\beta}} \left(\bm{\hat{\beta}} - \bm{\beta} \right) + r_{1k} \right)\\
   & \approx \sum_{i=1}^N \gamma_{i}^{(d)}\left( f_k(\bm{x}_i;\bm{\beta}) + 
   F_{ki} \left(\bm{\hat{\beta}} - \bm{\beta} \right) \right),\quad k=1,\dots,K,
\end{align*}
where $r_{1k}$ is the residual of this Taylor approximation and is of order $o_p(1/\sqrt{n})$, while $F_{ki}$ is a fixed term.

Therefore,
\begin{align*}
\Var_{\D} \Exp_{\M}\left(\hat{\theta}_k^{(d)}|\bm{\lambda}\right) & \approx \Var_{\D} \Exp_{\M}\left(\sum_{i=1}^N \gamma_{i}^{(d)}\left( f_k(\bm{x}_i;\bm{\beta}) + 
   F_{ki} (\bm{\hat{\beta}} - \bm{\beta} \right)\right)|\bm{\lambda})\\
   & = \Var_{\D} \Exp_{\M}\left(\sum_{i=1}^N \gamma_{i}^{(d)} 
   F_{ki} (\bm{\hat{\beta}} - \bm{\beta})|\bm{\lambda}\right)\\
   &\approx \bm{\gamma}^{(d)T} \mathbf{F}_k
   \Var_{\D} \Exp_{\M}\left(\bm{\hat{\beta}} - \bm{\beta} \big \lvert \bm{\lambda}\right) \mathbf{F}^T_k \bm{\gamma}^{(d)},\quad \quad k=1,\dots,K, 
\end{align*}

Now, by assumption ii), we are saying that there exists a measurable set $A$ with $\mathbb{P}(\bm \lambda \in A)=1-\varepsilon$, such that $\Exp_{\M}(\hat{\bm{\beta}} | \bm{\lambda})  =\bm\beta, \text{for all }\lambda\in A$ (i.e., model–unbiased for all $\lambda\in A$). Denoted with $A^c$ its complement, we have that, and using the argument $M \doteq \sup_{\bm{\lambda}}\norm{\Exp_{\M}(\hat{\bm{\beta}} | \bm{\lambda}) - \bm \beta} < \infty$ of the same condition ii), we have that:
\begin{align*}
\Var_{\D} \Exp_{\M}(\hat{\bm{\beta}} -\bm \beta | \bm{\lambda}) & \leq \mathbb{E}\big[\|\Exp_{\M}(\hat{\bm{\beta}} | \bm{\lambda})-\boldsymbol\beta\|^2\big]\\
& = \mathbb{E}\big[\|\Exp_{\M}(\hat{\bm{\beta}} | \bm{\lambda})-\boldsymbol\beta\|^2\mathbf{1}_A\big]
+\mathbb{E}\big[\|\Exp_{\M}(\hat{\bm{\beta}} | \bm{\lambda})-\boldsymbol\beta\|^2\mathbf{1}_{A^c}\big]\\
& = \mathbb{E}\big[\|\Exp_{\M}(\hat{\bm{\beta}} | \bm{\lambda})-\boldsymbol\beta\|^2\mathbf{1}_{A^c}\big]\\
& \le \mathbb{P}(A^c)\,\sup_{\bm{\lambda}}\norm{\Exp_{\M}(\hat{\bm{\beta}} | \bm{\lambda}) - \bm \beta}^2
= \varepsilon M^2,
\end{align*}
given that $\Exp_{\M}(\hat{\bm{\beta}} | \bm{\lambda})-\boldsymbol\beta = 0$ on $A$. 

In conclusion, we have that
\begin{align*}
\Var_{\D} \Exp_{\M}\left(\hat{\theta}_k^{(d)}|\bm{\lambda}\right) \lessapprox  \bm{\gamma}^{(d)T} \mathbf{F}_k
   \varepsilon M^2\mathbf{F}^T_k \bm{\gamma}^{(d)},\quad \quad k=1,\dots,K, 
\end{align*}
which is negligible when $\epsilon \approx 0$.

\paragraph{Part II.} Noticing that the only term affected by randomness is $\theta_k^{(d)} = \sum_{i=1}^N\gamma_i^{(d)}Y_{ik}$, with randomness uniquely coming from the model and such that $\mathbb{E}(Y_{ik}|\bm{x}_i) = f_k(\bm{x}_i;\bm{\beta}), i=1,\dots,N$, we have
\begin{align*}
    \Exp_{(\D, \M)}\left(\tilde{\theta}_k^{(d)} -\theta_k^{(d)}\right)^2 &= \Exp_{\M}\left(\tilde{\theta}_k^{(d)} -\theta_k^{(d)}\right)^2 = \left(\tilde{\theta}_k^{(d)}\right)^2 + \Exp_{\M}\left(\theta_k^{(d)}\right)^2 - 2 \tilde{\theta}_k^{(d)}\Exp_{\M}\left(\theta_k^{(d)}\right)\\
    & = \Exp_{\M}\left(\theta_k^{(d)}\right)^2 - \left(\tilde{\theta}_k^{(d)}\right)^2 = \Exp_{\M}\left(\theta_k^{(d)}\right)^2 - \left(\Exp_{\M}\left(\theta_k^{(d)}\right)\right)^2\\
    & = \Var_{\M}\left(\theta_k^{(d)}\right).
\end{align*}

\paragraph{Part III.} For the last component, using condition i) and similar considerations as in Part I and Part II, and distinguishing between two cases depending on whether $\Exp_{\D}\left(\hat{\theta}_k^{(d)}|\bm{Y}\right) = \theta_k^{(d)}$ $(A)$ or not $(B)$, we have that 
\begin{align*}
    \Exp_{(\D, \M)}\left((\hat{\theta}_k^{(d)} - \tilde{\theta}_k^{(d)})(\tilde{\theta}_k^{(d)} -\theta_k^{(d)})\right) &= \tilde{\theta}_k^{(d)}\Exp_{(\D, \M)}\left(\hat{\theta}_k^{(d)}\right) - \Exp_{(\D, \M)}\left(\hat{\theta}_k^{(d)}\theta_k^{(d)}\right) - \left(\tilde{\theta}_k^{(d)}\right)^2 +  \tilde{\theta}_k^{(d)}\Exp_{\M}\left(\theta_k^{(d)}\right)\\
    & = \left(\tilde{\theta}_k^{(d)}\right)^2 - \Exp_{\M}\left(\theta_k^{(d)}\Exp_{\D}\left(\hat{\theta}_k^{(d)}|\bm{Y}\right)\right) - \cancel{\left(\tilde{\theta}_k^{(d)}\right)^2} +  \cancel{\left(\tilde{\theta}_k^{(d)}\right)^2}\\
    & = 
    \begin{cases}\left(\Exp_{\M}\left(\theta_k^{(d)}\right)\right)^2 - \Exp_{\M}\left(\theta_k^{(d)}\right)^2 = -\Var_{\M}\left(\theta_k^{(d)}\right)\quad &(A)\\ \tilde{\theta}_k^{(d)}\Exp_{\M}\left(\theta_k^{(d)}\right) - \Exp_{\M}\left(\theta_k^{(d)}\Exp_{\D}\left(\hat{\theta}_k^{(d)}|\bm{Y}\right)\right) =-\Cov_{\M}\left(\Exp_{\D}\left(\hat{\theta}_k^{(d)}|\bm{Y}\right), \theta_k^{(d)}\right),\quad &(B)
    \end{cases}
\end{align*}
where the last equality follows from $\tilde{\theta}_k^{(d)} = \Exp_{(\D,\M)}(\hat{\theta}_k^{(d)}) = \Exp_{\M}\left(\Exp_{\D}\left(\hat{\theta}_k^{(d)}|\bm{Y}\right)\right)$.
Notice that the case $(A)$ is associated with the setting in which $\gamma_i^{(d)} \in \bm{x}_i, i=1,\dots,N$, that is, when the model $f_k(\bm{x}_i;\bm{\beta})$ has a domain intercept. This is likely to happen for large domains, indicating that there are sample observations at the domain level which are used for estimating the model.

Jointly considering the three parts and the two cases $(A)$ and $(B)$, we have that
\begin{align*}
    \GMSE\left(\hat{\theta}_k^{(d)}, \theta_k^{(d)}\right) &\approx \begin{cases}\Exp_{\D} \Var_{\M}\left(\hat{\theta}_k^{(d)}|\bm{\lambda}\right) -\Var_{\M}\left(\theta_k^{(d)}\right)\quad &(A): \gamma_i^{(d)} \in \bm{x}_i \\
    \Exp_{\D} \Var_{\M}\left(\hat{\theta}_k^{(d)}|\bm{\lambda}\right) +\Var_{\M}\left(\theta_k^{(d)}\right) -2\Cov_{\M}\left(\Exp_{\D}\left(\hat{\theta}_k^{(d)}|\bm{Y}\right), \theta_k^{(d)}\right)\quad &(B): \gamma_i^{(d)} \notin \bm{x}_i
    \end{cases}\\
    & \lessapprox \Exp_{\D} \Var_{\M}\left(\hat{\theta}_k^{(d)}|\bm{\lambda}\right).
\end{align*}
The final inequality arises from observing that in $(A)$ the term $\Var_{\M}\left(\theta_k^{(d)}\right) > 0$ enters with a negative sign and that in $(B)$ it is of the same order of magnitude of $\Cov_{\M}\left(\Exp_{\D}\left(\hat{\theta}_k^{(d)}|\bm{Y}\right), \theta_k^{(d)}\right)$ as it represents the model covariance between the two population totals for the same category $k$, $\theta_k^{(d)}$ and $\Exp_{\D}\left(\hat{\theta}_k^{(d)}|\bm{Y}\right)$, having thus a positive sign. Furthermore, it is also important to note that the first component $\Exp_{\D}\Var_{\M}\left(\hat{\theta}_k^{(d)}|\bm{\lambda}\right)$ is the dominant component, of primary order compared to the remaining terms. The following Table~\ref{tab: th1_decomp}, reflecting the simulation setup of Section 5, supports the statement.

\begin{table}[h!]
\centering
\caption{Empirical quantification of the GMSE decomposition components 
from Theorem~1 and the tightness of the upper bound, with reference to 
the simulation setup of Section 5 for $N = 500{,}000$. 
$\widehat{\text{GMSE}}^{\text{Lin}}$ approximates the upper bound 
$\Exp_{\D}\Var_{\M}\left(\hat{\theta}_k^{(d)}|\bm{\lambda}\right)$; 
$\Var_{\M}\left(\theta_k^{(d)}\right)$ is the minor-order term entering 
with a negative sign in case $(A)$; 
$\Var_{\D}\Exp_{\M}\left(\hat{\theta}_k^{(d)}|\bm{\lambda}\right)$ is 
the term vanishing under condition ii) of Theorem~1.
} \label{tab: th1_decomp}
\resizebox{\textwidth}{!}{%
\begin{tabular}{lrrrrr}
    \toprule
    Category $k$ & 
    $\widehat{\text{GMSE}}^{\text{Lin}}_k$ & 
    $\Exp_{\D}\Var_{\M}\left(\hat{\theta}_k^{(d)}|\bm{\lambda}\right)$ & 
    $\widehat{\text{GMSE}}^{\text{MC}}_k$ &
    $\Var_{\M}\left(\theta_k^{(d)}\right)$ & 
    $\Var_{\D}\Exp_{\M}\left(\hat{\theta}_k^{(d)}|\bm{\lambda}\right)$ \\
    \midrule
\texttt{1 Illiterate}                & 28174   & 27558   & 26484   & 1394   & 320   \\
\texttt{2 Literate but no education} & 163212  & 142668  & 135432  & 7858   & 623   \\
\texttt{3 Primary}                   & 639638  & 600955  & 574524  & 31905  & 5474  \\
\texttt{4 Lower secondary}           & 1079819 & 1051547 & 1006703 & 54404  & 9560  \\
\texttt{5 Upper secondary}           & 889771  & 939370  & 902905  & 45552  & 9087  \\
\texttt{6 Bachelor degree}           & 129845  & 125931  & 120250  & 6808   & 1128  \\
\texttt{7 Master degree}             & 248783  & 247501  & 236347  & 12434  & 1280  \\
\texttt{8 PhD level}                 & 16690   & 19153   & 18534   & 854    & 235   \\
    \bottomrule
\end{tabular}}
\end{table}

\subsection{Technical note on the validity of the two-step linearisation} \label{app: Taylor_conditions_remainders}

This section provides sufficient conditions for the validity of the two Taylor expansions used in Section~3 and establishes the stochastic order of the remainder terms.


Let $(U_N)_{N\geq1}$ be a sequence of finite populations with sample size $n=n_N$ such that $N\to\infty$ and $n\to\infty$. 

\begin{enumerate}[label=(A\arabic*),leftmargin=2em]
\item \textbf{Consistency of the estimator.} The solution $\hat{\bm{\beta}}$ of Eq. (9) is unique with probability tending to one and
\[
\hat{\bm{\beta}} \xrightarrow{p} \bm{\beta},
\qquad
\sqrt n(\hat{\bm{\beta}}-\bm{\beta})=O_p(1).
\]

\item \textbf{Smooth prediction functions.} For each $k=1,\dots,K$ and $i=1,\dots,N$, the map
$\bm{\beta}\mapsto f_k(\bm{x}_i;\bm{\beta})$
is twice continuously differentiable in an open neighbourhood of $\bm{\beta}$.

\item \textbf{Bounded derivatives.} There exists $C<\infty$ such that, uniformly in $i$,
\[
\left\|
\frac{\partial f_k(\bm{x}_i;\bm b)}{\partial \bm b}
\right\|\le C,
\qquad
\left\|
\frac{\partial^2 f_k(\bm{x}_i;\bm b)}{\partial \bm b\,\partial \bm b^T}
\right\|\le C
\]
for all $\bm b$ in a neighbourhood of $\bm{\beta}$.

\item \textbf{Smooth estimating equations.} Each $\mathbf g_i(\bm{\beta};\bm y,\bm x)$ is continuously differentiable in $\bm{\beta}$, and its derivative is uniformly bounded in probability.

\item \textbf{Non-singular Jacobian.} The matrix
\[
\mathbf A_{\bm{\beta}}
=
\sum_{i=1}^N
\lambda_i
\frac{\partial \mathbf g_i(\bm{\beta};\bm y,\bm x)}{\partial \bm{\beta}^T}
\]
is nonsingular, and $\|\mathbf A_{\bm{\beta}}^{-1}\|=O(1)$.
\end{enumerate}

These assumptions are standard and are satisfied by the multinomial logistic model under bounded covariates and probabilities bounded away from $0$ and $1$.

\paragraph{Order of the remainders}

Recalling that $\hat{\theta}_k^{(d)}
= \sum_{i=1}^N
\gamma_i^{(d)}f_k(\bm x_i;\hat{\bm{\beta}})$, by the multivariate Taylor theorem, 
\[
f_k(\bm x_i;\hat{\bm{\beta}})
=
f_k(\bm x_i;\bm{\beta})
+
\frac{\partial f_k(\bm x_i;\bm{\beta})}{\partial\bm{\beta}^T}
(\hat{\bm{\beta}}-\bm{\beta})
+
\rho_{ik},
\]
where
\[
\rho_{ik}
=
\frac12
(\hat{\bm{\beta}}-\bm{\beta})^T
\mathbf H_{ik}(\tilde{\bm{\beta}}_i)
(\hat{\bm{\beta}}-\bm{\beta}),
\]
with $\mathbf H_{ik}$ being the Hessian matrix of $f_k$, and $\tilde{\bm{\beta}}_i$ lying between $\hat{\bm{\beta}}$ and $\bm{\beta}$. Hence,
\[
\hat{\theta}_k^{(d)}
=
\sum_{i=1}^N \gamma_i^{(d)} f_k(\bm x_i;\bm{\beta})
+
\bm{\gamma}^{(d)T}\mathbf F_k(\hat{\bm{\beta}}-\bm{\beta})
+
r_{1k},
\]
with
\[
r_{1k}
=
\sum_{i=1}^N \gamma_i^{(d)}\rho_{ik}.
\]
Under condition (A3), we have that $|\rho_{ik}|
\le C\|\hat{\bm{\beta}}-\bm{\beta}\|^2$. Therefore,
\[
|r_{1k}|
\le
C\Big(\sum_{i=1}^N |\gamma_i^{(d)}|\Big)
\|\hat{\bm{\beta}}-\bm{\beta}\|^2.
\]
Since by (A1) $\sqrt n(\hat{\bm{\beta}}-\bm{\beta})=O_p(1)$, it is verified that $\|\hat{\bm{\beta}}-\bm{\beta}\|^2 = O_p(n^{-1})$; hence
$r_{1k}=O_p(n^{-1})=o_p(n^{-1/2})$.

Regarding the second linearisation step, recalling that $\hat{\bm{\beta}}$ is obtained by solving:
\begin{align*}
    \sum_{i=1}^{N} \lambda_i \mathbf{g}_{i}(\bm{\beta}; \bm{y}, \bm{x}) = \frac{\partial \ell(\bm{\beta}; \bm{y}, \bm{x}, \bm{\lambda})}{\partial \bm{\beta}} = \mathbf{0}_{H},
\end{align*}
where $\mathbf{0}$ is a vector of $H$ zeroes, a first-order Taylor expansion around $\bm{\beta}$ gives
\[
\mathbf 0
=
\mathbf \sum_{i=1}^N \lambda_i \mathbf g_i(\bm{\beta};\bm y,\bm x)
+
\mathbf A_{\bm{\beta}}(\hat{\bm{\beta}}-\bm{\beta})
+
\mathbf q_n,
\]
where the remainder $\mathbf q_n$ satisfies $\|\mathbf q_n\| = O_p\!\left(\|\hat{\bm{\beta}}-\bm{\beta}\|^2\right)$.

Solving for $\hat{\bm{\beta}}-\bm{\beta}$, we get
\[
\hat{\bm{\beta}}-\bm{\beta}
=
-\mathbf A_{\bm{\beta}}^{-1} \sum_{i=1}^{N} \lambda_i \mathbf{g}_{i}(\bm{\beta}; \bm{y}, \bm{x})
-\mathbf A_{\bm{\beta}}^{-1}\mathbf q_n,
\]
where we denoted the remainder by $\bm r_{2}
= -\mathbf A_{\bm{\beta}}^{-1}\mathbf q_n$.

By condition (A5), we get the order of $\bm r_{2}$, given by
\[
\|\bm r_{2}\|
\le
\|\mathbf A_{\bm{\beta}}^{-1}\|\,\|\mathbf q_n\|
=
O_p\!\left(\|\hat{\bm{\beta}}-\bm{\beta}\|^2\right)
=
O_p(n^{-1})
=
o_p(n^{-1/2}).
\]

Since both $r_{1k}$ and $\bm r_{2}$ are $o_p(n^{-1/2})$, omitted second-order terms are asymptotically negligible relative to the leading term.


\subsection{Derivation of the multinomial model results} \label{app: multinomial_model}

We assume that the outcome variable for each unit $i$, with $i = 1,\dots, N$, follows a multinomial regression model, that is, $\bm{Y}_i \sim \text{Multinomial}(1, \bm{p}_i)$, where $\bm{p}_i = \left(p_{i1}, \dots, p_{ik}, \dots,p_{iK} \right)^T$ denotes the unknown event probabilities in $[0,1]$ of each category $k$, such that $\sum_{k=1}^Kp_{ik} = 1$. Under this model, we have that:
\begin{align*}
    P(Y_{i1} = y_{i1}, \dots, Y_{iK} = y_{iK} | \bm{p}_i) = p_{i1}^{y_{i1}}\cdots p_{iK}^{y_{iK}} = \prod_{k=1}^{K} p_{ik}^{y_{ik}}, \quad i=1,\dots N.
\end{align*}

Probabilities $\bm{p}_i$ are defined based on the set of individual covariates $\bm{x}_i \in \mathbb{R}^J$, for all $i = 1,\dots,N$, according to the multinomial logistic regression model where $K$ is taken as the baseline category:
\begin{align} \label{eq: mult_model_app}
    p_{ik} = f_k(\bm{x}_i;\bm{\beta}) = 
    \begin{dcases}
    \frac{\exp{\bm{x}_i}^T\bm{\beta}_k}{ 1 + \sum_{k=1}^{K-1}\exp{\bm{x}_i}^T\bm{\beta}_k}, \quad & k=1,\dots,K-1\\
    \frac{1}{ 1 + \sum_{k=1}^{K-1}\exp{\bm{x}_i}^T\bm{\beta}_k}, \quad & k = K. 
    \end{dcases}
\end{align}

Under the assumption of independence among the $N$ units, and considering the sampling membership indicators $\lambda_1,\dots, \lambda_N$, the joint density function is given by:
\begin{align}
    f(\bm{y}; \bm{x}, \bm{\beta}, \bm{\lambda}) 
    &= \prod_{i = 1}^{N} \left(\prod_{k=1}^{K} p_{ik}^{y_{ik}} \right)^{\lambda_i} \nonumber \\ 
    & = \prod_{i = 1}^{N} \prod_{k=1}^{K} p_{ik}^{y_{ik} \lambda_i}, \tag{$(ab)^x = a^x b^x$} \\ 
    & = \prod_{i = 1}^{N} \prod_{k=1}^{K-1} p_{ik}^{y_{ik}\lambda_i}\cdot p_{iK}^{y_{iK}\lambda_i} \tag{separate the $K$-th term}  \\ 
    & = \prod_{i = 1}^{N} \prod_{k=1}^{K-1} p_{ik}^{y_{ik}\lambda_i}\cdot p_{iK}^{(1-\sum_{k=1}^{K-1}y_{ik})\lambda_i}\tag{$\sum_{i=1}^ky_{ik}=1$}  \\ 
    & = \prod_{i = 1}^{N} \prod_{k=1}^{K-1} p_{ik}^{y_{ik}\lambda_i}\cdot \frac{p_{iK}^{\lambda_i}}{p_{iK}^{\sum_{k=1}^{K-1}y_{ik}\lambda_i}}\nonumber \\ 
    & = \prod_{i = 1}^{N} \prod_{k=1}^{K-1} p_{ik}^{y_{ik}\lambda_i}\cdot \frac{p_{iK}^{\lambda_i}}{\prod_{k=1}^{K-1}p_{iK}^{y_{ik}\lambda_i}}\tag{$a^{x+y} = a^x a^y$}  \\ 
    & = \prod_{i = 1}^{N} \prod_{k=1}^{K-1} \left(\frac{p_{ik}}{p_{iK}}\right)^{y_{ik}\lambda_i}\cdot p_{iK}^{\lambda_i}. \nonumber
    \end{align}
Note that in a classical inference setting procedures are conditioned on the observed sample only, thus, the sample membership terms are omitted. 

Now, taken into account Eq.~\eqref{eq: mult_model_app}, the likelihood and the log-likelihood functions are given by:
\begin{align*}
    L(\bm{\beta}; \bm{y}, \bm{x}, \bm{\lambda}) 
    &= \prod_{i = 1}^{N} \prod_{k=1}^{K-1} \left(\exp{\sum_{j=1}^J x_{ij}\beta_{kj}} \right)^{y_{ik} \lambda_i}  \cdot\left(\frac{1}{ 1 + \sum_{k=1}^{K-1}\exp{\sum_{j=1}^J x_{ij}\beta_{kj}}}\right)^{\lambda_i},\label{eq: likely}\\
    \ell(\bm{\beta}; \bm{y}, \bm{x}, \bm{\lambda}) 
    &= \sum_{i = 1}^{N} \left[\sum_{k=1}^{K-1} \left( \lambda_i y_{ik} \sum_{j=1}^J x_{ij}\beta_{kj} \right) - \lambda_i \log\left(1 + \sum_{k=1}^{K-1}\exp{\sum_{j=1}^J x_{ij}\beta_{kj}}\right)\right] \nonumber\\
    & = \sum_{i = 1}^{N} \lambda_i \left[\sum_{k=1}^{K-1} \left( y_{ik} \sum_{j=1}^J x_{ij}\beta_{kj} \right) - \log\left(1 + \sum_{k=1}^{K-1}\exp{\sum_{j=1}^J x_{ij}\beta_{kj}}\right)\right].
\end{align*}

The MLEs  $\hat{\bm{\beta}}_k = (\hat{\beta}_{k1}, \dots, \hat{\beta}_{kJ})$, $k=1,\dots,K$, of the unknown parameters can be obtained as the solution of the problem:
\begin{align}
    \frac{\partial \ell(\bm{\beta}; \bm{y}, \bm{x}, \bm{\lambda})}{\partial \beta_{kj}} 
    &= 0, \quad k=1,\dots,K, j=1,\dots, J,
\end{align}
with overall $K \times J $ number of equations, which is equivalent to setting
\begin{align*}
    \sum_{i=1}^{N} \lambda_i \mathbf{g}_{i}(\bm{\beta}; \bm{y}, \bm{x}) = \frac{\partial \ell(\bm{\beta}; \bm{y}, \bm{x}, \bm{\lambda})}{\partial \bm{\beta}} = \mathbf{0},
\end{align*}
where $\mathbf{0}$ is a vector of $H$ zeroes and $\mathbf{g}_i = \{g_{ikj}\}_{k=1,\dots, K-1, j=1,\dots,J}$ is a function of both $\bm{\beta}$ and $\mathbf{y}_i$.

Although technically a matrix, we can consider $\bm{\beta}$ to be a column vector by appending each of the additional columns below the first. This would be useful in forming the matrix of second partial derivatives $\mathbf{A}_{\bm{\beta}}$. Now, for a generic $k$ and $j$, we have that:
\begin{align} \label{eq: deriv1}
    \frac{\partial \ell(\bm{\beta}; \bm{y}, \bm{x}, \bm{\lambda})}{\partial \beta_{kj}} 
    & = \sum_{i = 1}^{N} \lambda_i \left[y_{ik} x_{ij} - \frac{\frac{\partial}{\partial \beta_{kj}} \left(\sum_{k=1}^{K-1}\exp{\sum_{j=1}^J x_{ij}\beta_{kj}}\right)}{1 + \sum_{k=1}^{K-1}\exp{\sum_{j=1}^J x_{ij}\beta_{kj}}} \right]\nonumber \\
    & = \sum_{i = 1}^{N} \lambda_i \left[ y_{ik} x_{ij} -  \frac{x_{ij}\cdot\exp{\sum_{j=1}^J x_{ij}\beta_{kj}}}{1 + \sum_{k=1}^{K-1}\exp{\sum_{j=1}^J x_{ij}\beta_{kj}}} \right] \nonumber \\
    & = \sum_{i = 1}^{N} \lambda_i  x_{ij} \left( y_{ik} - \frac{\exp{\sum_{j=1}^J x_{ij}\beta_{kj}}}{1 + \sum_{k=1}^{K-1}\exp{\sum_{j=1}^J x_{ij}\beta_{kj}}} \right ) \\
    & = \sum_{i = 1}^{N} \lambda_i  \underbrace{x_{ij} \left( y_{ik} - p_{ik} \right)}_{g_{ikj}} \label{eq: gikj}\\
    & = \sum_{i=1}^{N} \lambda_i  g_{ikj} (y_{ik} | \bm{\beta}) 
\end{align}

The Hessian, or second-order derivatives matrix, say $\mathbf{A}_{\bm{\beta}}$, of dimension $H \times H$ is now be obtained as:
\begin{align} \label{eq: hessian_app}
    \mathbf{A}_{\bm{\beta}} = [a_{(kj)(k'j')}],\quad a_{(kj)(k'j')} \doteq \frac{\partial^2 \ell(\bm{\beta}; \bm{y}, \bm{x}, \bm{\lambda})}{\partial \beta_{kj} \partial \beta_{k'j'}}, 
\end{align}
where 
\begin{align*}
     \frac{\partial^2 \ell(\bm{\beta}; \bm{y}, \bm{x}, \bm{\lambda})}{\partial \beta_{kj} \partial \beta_{k'j'}} &= \frac{\partial }{\partial \beta_{k'j'}} \sum_{i = 1}^{N}  \left[\lambda_i y_{ik} x_{ij} - \lambda_i \frac{x_{ij}\cdot\exp{\sum_{j=1}^J x_{ij}\beta_{kj}}}{1 + \sum_{k=1}^{K-1}\exp{\sum_{j=1}^J x_{ij}\beta_{kj}}} \right] \\
     & = \frac{\partial }{\partial \beta_{k'j'}} \sum_{i = 1}^{N} -  \lambda_i \frac{x_{ij}\cdot\exp{\sum_{j=1}^J x_{ij}\beta_{kj}}}{1 + \sum_{k=1}^{K-1}\exp{\sum_{j=1}^J x_{ij}\beta_{kj}}} \\
     & = \sum_{i = 1}^{N} -  \lambda_i x_{ij} \frac{\partial }{\partial \beta_{k'j'}}  \left(\frac{\exp{\sum_{j=1}^J x_{ij}\beta_{kj}}}{1 + \sum_{k=1}^{K-1}\exp{\sum_{j=1}^J x_{ij}\beta_{kj}}} \right)\\
     & = \begin{cases}
     \sum_{i = 1}^{N} - \lambda_i x_{ij} \frac{\left(1 + \sum_{k=1}^{K-1}e^{\sum_{j=1}^J x_{ij}\beta_{kj}}\right)\cdot e^{\sum_{j=1}^J x_{ij}\beta_{kj}} \cdot x_{ij'} - e^{\sum_{j=1}^J x_{ij}\beta_{kj}} \cdot e^{\sum_{j=1}^J x_{ij}\beta_{kj}} \cdot x_{ij'}}{\left(1 + \sum_{k=1}^{K-1}e^{\sum_{j=1}^J x_{ij}\beta_{kj}}\right)^2},  &\quad k' = k \\
     \sum_{i = 1}^{N} -  \lambda_i x_{ij} \frac{ 0 - e^{\sum_{j=1}^J x_{ij}\beta_{kj}} \cdot e^{\sum_{j=1}^J x_{ij}\beta_{k'j}} \cdot x_{ij'}}{\left(1 + \sum_{k=1}^{K-1}e^{\sum_{j=1}^J x_{ij}\beta_{kj}}\right)^2},  &\quad k' \neq k 
     \end{cases}\\
     & = \begin{cases}
     \sum_{i = 1}^{N} -  \lambda_i x_{ij} x_{ij'}\frac{e^{\sum_{j=1}^J x_{ij}\beta_{kj}} \left(1 + \sum_{k=1}^{K-1}e^{\sum_{j=1}^J x_{ij}\beta_{kj}} - e^{\sum_{j=1}^J x_{ij}\beta_{kj}} \right)}{\left(1 + \sum_{k=1}^{K-1}e^{\sum_{j=1}^J x_{ij}\beta_{kj}}\right)^2},  &\quad k' = k \\
     \sum_{i = 1}^{N}  \lambda_i x_{ij} x_{ij'}\frac{e^{\sum_{j=1}^J x_{ij}\beta_{kj}} \cdot e^{\sum_{j=1}^J x_{ij}\beta_{k'j}} }{\left(1 + \sum_{k=1}^{K-1}e^{\sum_{j=1}^J x_{ij}\beta_{kj}}\right)^2},  &\quad k' \neq k 
     \end{cases}\\
     & = \begin{cases}
     -\sum_{i = 1}^{N}  \lambda_i x_{ij} x_{ij'}p_{ik}(1-p_{ik}),  &\quad k' = k \\
     \sum_{i = 1}^{N}  \lambda_i x_{ij} x_{ij'}p_{ik}p_{ik'},  &\quad k' \neq k 
     \end{cases}.
\end{align*}

Therefore,
\begin{align*} 
\underset{H \times H}{\mathbf{A}_{\bm{\beta}}} 
& = 
\begin{bmatrix} 
\mathbf{A}_{\bm{\beta}_{(1,1)}} & 
\mathbf{A}_{\bm{\beta}_{(1,2)}}
& \dots & \mathbf{A}_{\bm{\beta}_{(1,K)}}\\
\mathbf{A}_{\bm{\beta}_{(2,1)}}& \mathbf{A}_{\bm{\beta}_{(2,2)}}& \dots &  \mathbf{A}_{\bm{\beta}_{(2,K)}}\\
\vdots & \vdots & \vdots & \vdots \\
\mathbf{A}_{\bm{\beta}_{(K,1)}}&
\mathbf{A}_{\bm{\beta}_{(K,2)}}& \dots &
\mathbf{A}_{\bm{\beta}_{(K,K)}}
\end{bmatrix} 
\\
& = \begin{bmatrix} 
\overbrace{\begin{bmatrix} 
-\sum_{i = 1}^{N}  \lambda_i \mathbf{x}_{i} \mathbf{x}_{i}'p_{i1}(1-p_{i1})
\end{bmatrix}}^{J \times J} & 
\dots
& \dots & \overbrace{\begin{bmatrix} 
\sum_{i = 1}^{N}  \lambda_i \mathbf{x}_{i} \mathbf{x}'_{i}p_{i1}p_{i2}
\end{bmatrix}}^{J \times J}\\
\begin{bmatrix} 
\sum_{i = 1}^{N}  \lambda_i \mathbf{x}_{i} \mathbf{x}'_{i}p_{i2}p_{i1}
\end{bmatrix}& \vdots & \vdots &  \begin{bmatrix} 
\sum_{i = 1}^{N}  \lambda_i \mathbf{x}_{i} \mathbf{x}'_{i}p_{i2}p_{i(K-1)}
\end{bmatrix}\\
\vdots & \vdots & \vdots & \vdots \\
\underbrace{\begin{bmatrix} 
\sum_{i = 1}^{N}  \lambda_i \mathbf{x}_{i} \mathbf{x}'_{i}p_{i(K-1)}p_{i1}
\end{bmatrix}}_{J \times J}&
 \dots & \dots &
\underbrace{\begin{bmatrix} 
-\sum_{i = 1}^{N}  \lambda_i \mathbf{x}_{i} \mathbf{x}'_{i}p_{i(K-1)}(1-p_{i(K-1)})
\end{bmatrix}}_{J \times J}
\end{bmatrix} \nonumber\\
& = -\sum_{i = 1}^{N}  \lambda_i 
\begin{bmatrix} 
\overbrace{\begin{bmatrix} 
\mathbf{x}_{i} \mathbf{x}_{i}' \sigma_{i(1)}^2
\end{bmatrix}}^{J \times J} & 
\overbrace{\begin{bmatrix} 
\mathbf{x}_{i} \mathbf{x}'_{i} \sigma_{i(1,2)}
\end{bmatrix}}^{J \times J}
& \dots & \overbrace{\begin{bmatrix} 
\mathbf{x}_{i} \mathbf{x}'_{i} \sigma_{i(1,K-1)}
\end{bmatrix}}^{J \times J}\\
\begin{bmatrix} 
\mathbf{x}_{i} \mathbf{x}'_{i}\sigma_{i(2,1)}
\end{bmatrix}& \begin{bmatrix} 
\mathbf{x}_{i} \mathbf{x}'_{i} \sigma_{i(2)}^2
\end{bmatrix} & \dots &  \begin{bmatrix} 
\mathbf{x}_{i} \mathbf{x}'_{i}\sigma_{i(2,K-1)}
\end{bmatrix}\\
\vdots & \vdots & \vdots & \vdots \\
\underbrace{\begin{bmatrix} 
\mathbf{x}_{i} \mathbf{x}'_{i}\sigma_{i(K-1,1)}
\end{bmatrix}}_{J \times J}&
\underbrace{\begin{bmatrix} 
\mathbf{x}_{i} \mathbf{x}'_{i}\sigma_{i(K-1,2)}
\end{bmatrix}}_{J \times J}& \dots &
\underbrace{\begin{bmatrix} 
\mathbf{x}_{i} \mathbf{x}'_{i} \sigma_{i(K-1)}^2
\end{bmatrix}}_{J \times J}
\end{bmatrix} 
\end{align*}
with the generic $\sigma_{i(k)}^2 = p_{ik}(1-p_{ik})$ being the variance of $Y_{ik}$ and $\sigma_{i(k,k')} = -p_{ik}p_{ik'}$ the covariance between $Y_{ik}$ and $Y_{ik'}$. The full variance-covariance matrix $\bm{\Sigma}_{\bm{Y}_i}$ related to the the multinomial model is given by 
\begin{align*} 
 \underset{K \times K}{\bm{\Sigma}_{\bm{Y}_i}} = 
    \begin{bmatrix} 
    p_{i1}(1-p_{i1}) &-p_{i1}p_{i2}&\dots& -p_{i1}p_{i(K-1)}\\
    -p_{i2}p_{i1}&p_{i2}(1-p_{i2}) &\dots& -p_{i2}p_{i(K-1)}\\
    \dots & \dots & \dots & \dots\\
    -p_{i(K-1)}p_{i1} &-p_{i(K-1)}p_{i2}&\dots& p_{i(K-1)}(1-p_{i(K-1)})
    \end{bmatrix}.   
\end{align*}

Given its role in the linearised GMSE, below we also provide the formulation of matrix $\underset{N \times H}{\mathbf{F}_k} = \left[\left.\frac{\partial f_k(\boldsymbol{x}_i;\boldsymbol{\hat{\beta}})}{\partial \hat{\beta}_{lj}} \right |_{\hat{\beta}_{lj} = \beta_{lj} }\right]$, in the case of the multinomial model:
\begin{align*}
    \begin{cases}
\left.\frac{\partial f_k(\boldsymbol{x}_i;\boldsymbol{\hat{\beta}})}{\partial \hat{\beta}_{lj}} \right |_{\hat{\beta}_{lj} = \beta_{lj} } = \frac{x_{ij}e^{\sum_{j=1}^J x_{ij}\beta_{kj}} \left(1 + \sum_{k=1}^{K-1}e^{\sum_{j=1}^J x_{ij}\beta_{kj}} - e^{\sum_{j=1}^J x_{ij}\beta_{kj}}\right)}{\left(1 + \sum_{k=1}^{K-1}e^{\sum_{j=1}^J x_{ij}\beta_{kj}}\right)^2} = x_{ij}p_{ik}(1-p_{ik}), \quad &l = k\\ 
\left.\frac{\partial f_k(\boldsymbol{x}_i;\boldsymbol{\hat{\beta}})}{\partial \hat{\beta}_{lj}} \right |_{\hat{\beta}_{lj} = \beta_{lj} } = \frac{-x_{ij}e^{\sum_{j=1}^J x_{ij}\beta_{kj}} \left( e^{\sum_{j=1}^J x_{ij}\beta_{lj}}\right)}{\left(1 + \sum_{k=1}^{K-1}e^{\sum_{j=1}^J x_{ij}\beta_{kj}}\right)^2} = -x_{ij}p_{ik}p_{il}, \quad &l \neq k.\\ 
    \end{cases}
\end{align*}

\subsection{Approximation of matrix $\mathbf{A_\beta}$} \label{app: approx_Ab}

The Hessian matrix $\mathbf{A}_{\bm{\beta}}$ has been derived in Eq.~\eqref{eq: hessian_app} and is given by:
\begin{align*}
    \mathbf{A}_{\bm{\beta}} = [a_{(kj)(k'j')}],\quad a_{(kj)(k'j')} \doteq \frac{\partial^2 \ell(\bm{\beta}; \bm{y}, \bm{x}, \bm{\lambda})}{\partial \beta_{kj} \partial \beta_{k'j'}} = \begin{cases}
     -\sum_{i = 1}^{N} \lambda_i x_{ij} x_{ij'}p_{ik}(1-p_{ik}),  &\quad k' = k \\
     \sum_{i = 1}^{N} \lambda_i x_{ij} x_{ij'}p_{ik}p_{ik'},  &\quad k' \neq k 
     \end{cases}. 
\end{align*}

We want to show that, under the setup in Sections 2 and 3.1, the following approximation is reasonable, removing in this way the dependence on the random variable $\bm{\lambda}$:
\begin{align*}
    \mathbf{A}_{\bm{\beta}} \approx \mathbf{A}_{\bm{\beta}} |_{\bm{\lambda} = \bm{\pi}}.
\end{align*}

We first notice that the only random variable in $\mathbf{A}_{\bm{\beta}}$ is given by $\lambda_i, i=1,\dots,N$, while the terms $x_{ij}, i=1,\dots,N, j=1,\dots,J$ and $p_{ik}, i=1,\dots,N, k=1,\dots,K$ are fixed. For simplicity, we focus on a single combination $j, j', k, k'$, with the generalisation naturally following. We set $c_i \doteq x_{ij}x_{ij'}p_{ik}p_{ik'}$ (the case with $k' = k$ would not make a difference), and let:
\begin{align*}
    a_i \doteq a_{(kj)(k'j')} = \sum_{i = 1}^{N} \lambda_i c_i.
\end{align*}

Noticing that $c_i \in \mathbb{R}$ are a set of real constants, we want to show that:
\begin{align*}
    \sum_{i = 1}^{N} \lambda_i c_i \approx \sum_{i = 1}^{N} \pi_i c_i. 
\end{align*}

Assuming a non-informative design, with $\Exp(\lambda_i) = \pi_i$ and $\Var(\lambda_i) = \pi_i(1-\pi_i)$, first, we note that the two quantities have the same expectation:
\begin{align*}
    \Exp(a_i) = \Exp\left(\sum_{i = 1}^{N} \lambda_i c_i\right) =  \sum_{i = 1}^{N} \pi_i c_i  = \tilde{a_i} = \Exp(\tilde{a_i}). 
\end{align*}
For the variance, assuming for simplicity a simple random design with replacement (implying independence), we have that:
\begin{align*}
    \Var(a_i) = \Var\left(\sum_{i = 1}^{N} \lambda_i c_i\right) =  \sum_{i = 1}^{N} \pi_i(1-\pi_i) c_i^2, 
\end{align*}
and the lower the variance, the better the approximation. In particular, the identity holds with probability one when this variance goes to $0$, meaning that $a_i$ is a degenerate distribution at $\tilde{a}_i$. This happens in the following cases:

\begin{enumerate}
    \item[(i)] when $c_i = 0$, for all $i$; 
    \item[(ii)] when $\pi_i \in \{0, 1\}$, for all $i$, that is, when the entire population is part of the sample (this rarely happens, except for census surveys, and in that case, one would not need to evaluate the register accuracy under sampling randomness).
\end{enumerate}

Since cases (i)-(ii) are rare or of no interest in practice, we are interested in a good approximation that holds when quantities $c_i$'s and/or $\pi_i(1-\pi_i)$ are very small, potentially close to $0$. Regarding the $c_i$'s, in our specific setting, the design matrix is defined based on a dummy format (with only $0/1$ values) and the $p_{ik}$'s are values in $(0,1)$, making it a very small constant. For the terms $\pi_i(1-\pi_i)$, these are high when $\pi_i \approx 0.5$ and very small when they are close to the extremes $0/1$. The latter is often the case in surveys, with $\pi_i$ being typically very small; in our case, for instance, we have $\pi_i = 0.05$ and this is already a decently high value. 

\section{Details on accuracy estimation approaches and pseudo-codes} \label{app: pseudoalgo}

\begin{algorithm}[H]
\caption{Linearised GMSE estimator for register-based totals}
\begin{algorithmic}[1]

\Require Register of $N$ units with complete covariates $\{\bm{x}_i\}_{i=1}^N$ of dimension $H$ 
  and partially observed outcomes $Y_{ik}$, $k=1,\dots,K$,
  for $i \in \mathcal{S}_n$ (sample of size $n$);
  inclusion probabilities $\{\pi_i\}_{i=1}^N$; a given domain of interest $d$;  domain membership vectors $\bm{\gamma}^{(d)}$; pre-specified categorical working model $f_k(\cdot;\bm{\beta})$.

\Ensure Adopt the Taylor linearisation approach to get the estimate
$\widehat{\GMSE}^{\text{Lin}}\!\left(\hat{\theta}_k^{(d)},\theta_k^{(d)}\right)$
  for $k=1,\dots,K$.

\State \textbf{Model fitting and linearisation (domain-invariant; one-shot estimate)}. Fit the working model $f_k(\bm{x}_i;\bm{\beta})$ on the sample $\mathcal{S}_n$
  by solving
  \[
    \sum_{i=1}^N \lambda_i\,\mathbf{g}_i(\bm{\beta};\bm{y},\bm{x}) = \mathbf{0}_H,
  \]
  with $\mathbf{g}_{i}(\bm{\beta}; \bm{y}, \bm{x}) = \{g_{ikj}(\bm{\beta}; \bm{y}, \bm{x}), j=1,\dots,J, k=1,\dots,K\}$ the system of $H$ generalised estimating equations. 
  Obtain $\hat{\bm{\beta}}$ and predicted probabilities
  $\hat{p}_{ik} = f_k(\bm{x}_i;\hat{\bm{\beta}})$ for all $i=1,\dots,N$, and 
  $k=1,\dots,K$.

\vspace{4pt}
\Statex \textbf{Domain-invariant quantities (one-shot estimates)}

\State Compute the matrix of the first-order partial derivatives, for $k=1,\dots,K$:
\begin{align*}
\underset{N \times H}{\mathbf{F}_k} = \left[\frac{\partial f_k(\bm{x}_i;\bm{\hat{\beta}})}{\partial \hat{\beta}_{lj}} \Big \lvert_{\hat{\bm{\beta}} = \bm{\beta} }\right],
\end{align*}
and get the plug-in estimates $\widehat{\mathbf{F}}_k,\ k = 1,\dots,K$.

\State Compute the plug-in Hessian matrix $\hat{\mathbf{A}}_{\bm{\beta}}$ of dimension $H \times H$, where
\begin{align*}
    \underset{H \times H}{\mathbf{A}_{\bm{\beta}}} = \left[a_{(kj)(k'j')}\big \lvert_{\hat{\bm{\beta}} = \bm{\beta}}\right],\quad\text{with}\quad a_{(kj)(k'j')} \doteq \sum_{i=1}^{N} \lambda_i\frac{ \partial \mathbf{g}_{i}(\bm{\beta}; \bm{y}, \bm{x})}{\partial \hat{\beta}_{kj}} \Bigg \lvert_{\hat{\bm{\beta}} = \bm{\beta}},\quad j = 1,\dots, J;\quad  k  = 1,\dots, K.
\end{align*}
Approximate $\mathbf{A}_{\bm{\beta}} \approx \mathbf{A}_{\bm{\beta}} |_{\bm{\lambda} = \bm{\pi}} \doteq \bar{\mathbf{A}}_{\bm{\beta}}$ according to Lemma 3 (main manuscript) and get the inverse $\bar{\hat{\mathbf{A}}}_{\bm{\beta}}^{-1}$.

\State For each unit $i=1,\dots,N$, compute:
  \begin{itemize}
    \item $\dot{\mathbf{X}}_i$ ($H\times K$): auxiliary matrix derived from
          $\bm{x}_i$ (see Eq. (16) in the main manuscript).
    \item $\bar{\hat{\mathbf{U}}}_i
          = \bar{\hat{\mathbf{A}}}_{\bm{\beta}}^{-1}\dot{\mathbf{X}}_i$
          \  ($H\times K$) following to Eq. (19) of the main manuscript.
    \item $\widehat{\Sigma}_{\bm{Y}_i}$ ($K\times K$):
          plug-in estimate of $\Var_{\M}(\bm{Y}_i)$, obtained from $\hat{\bm{p}}_i$.
  \end{itemize}

\State \textbf{Compute the GMSE estimates}. For $k=1,\dots,K$ and a given domain $d$, compute:
  \begin{align}
    \widehat{\GMSE}^{\text{Lin}}\left(\hat{\theta}_k^{(d)}, \theta_k^{(d)}\right) = \bm{\gamma}^{(d)T} \widehat{\mathbf{F}}_k \left( \sum_{i=1}^N \pi_i \widehat{\bar{\mathbf{U}}}_{i} \widehat{\Sigma}_{\bm{Y}_{i}} \widehat{\bar{\mathbf{U}}}_{i} \right) \widehat{\mathbf{F}}^T_k \bm{\gamma}^{(d)}. \tag{GMSE estimate}
  \end{align}

\State \textbf{Remark: Domain scalability}.
  for any additional domain $d'$, update $\bm{\gamma}^{(d')}$  and
  repeat Step~5 only.
\end{algorithmic}
\end{algorithm}

\begin{algorithm}[H]
\caption{Non-parametric Bootstrap~\citep[see][pg. 136] {chambers2012introduction} for register-based totals
}
\begin{algorithmic}[1]
\Require Register of $N$ units, with complete set of covariates $\{\bm{x}_i\}_{i=1}^N$ and partially observed outcomes $Y_{ik}, k=1,\dots,K,\ \text{and}\ i \in \mathcal{S}_n$, where $\mathcal{S}_n$ is a sample of size $n$; number of bootstrap replications $B$.
\Ensure Execute bootstrap to get the distribution of $\hat{\theta}_k,\ k=1,\dots,K$.
\For{$b = 1,\dots, B$}
    \State Draw a bootstrap sub-sample $\mathcal{S}^{(b)}_n$ of size $n$ with replacement from the originally observed sample $\mathcal{S}_n$.
    \State Draw a bootstrap sub-sample $\bar{\mathcal{S}}^{(b)}_n$ of size $N-n$ with replacement from the non-sampled register units.
    \State Use $\mathcal{S}^{(b)}_n$ to estimate the working model $f_k(\bm{x}_i; \bm{\beta})$, obtaining $\hat{\bm{\beta}}^{(b)}$ and $\hat{\bm{p}}_i^{(b)}, i = 1,\dots, N$.
    \State Estimate / predict the outcome values for all register units:
        \[
        \hat{Y}_{ik}^{(b)} = \hat{p}^{(b)}_{ik} = f_k(\bm{x}_i; \hat{\bm{\beta}}^{(b)}), 
        \quad i=1,\dots,N, \; k=1,\dots,K.
        \]
    \State Get the bootstrap estimate of the target totals ${\theta}_k,\ k=1,\dots,K$:
        \[
        \hat{\theta}_k^{(b)} = \sum_{i=1}^N \hat{Y}_{ik}^{(b)}, 
        \quad i=1,\dots,N, \; k=1,\dots,K.
        \]
        If a domain $d$ is considered, then the expression has to account for the domain membership $\bm{\gamma}^{(d)}$.
\EndFor
\State Given a set of bootstrap estimates $\{\hat{\theta}_k^{(b)}, b = 1,\dots, B \}$, these can now be used to obtain an estimate of the quantities of interest; in our case:  
\begin{align*}
\widehat{\GMSE}\!\left(\hat{\theta}^{\text{boot}}_k\right) &= \frac{1}{B} \sum_{b=1}^B \left( \hat{\theta}_k^{(b)} - \hat{\theta}_k \right)^2 = \widehat{\Var}\left(\hat{\theta}^{\text{boot}}_k\right) + \widehat{\Bias}^2\!\left(\hat{\theta}^{\text{boot}}_k\right), \tag{GMSE estimate}
    \end{align*}
    where $\hat{\theta}_k = \sum_{i=1}^N \hat{p}_{ik}$ is computed using the original sample ($\mathcal{S}_n$) estimates $\hat{p}_{ik}$, and 
    \begin{align*}
\widehat{\Var}\left(\hat{\theta}^{\text{boot}}_k\right) &= \frac{1}{B-1} \sum_{b=1}^B \left( \hat{\theta}_k^{(b)} - \bar{\theta}_k^{\text{boot}} \right)^2 \tag{Variance estimate}\\
\widehat{\Bias}^2\!\left(\hat{\theta}_k^{\text{boot}}\right) & = \left(\bar{\theta}_k^{\text{boot}} - \hat{\theta}_k\right)^2, \tag{Bias estimate}
    \end{align*}
with 
\begin{align*}
    \bar{\theta}_k^{\text{boot}} &= \frac{1}{B} \sum_{b=1}^B \hat{\theta}_k^{(b)}. \tag{Mean estimate}
    \end{align*}
\end{algorithmic}
\end{algorithm}

\subsection{Model-based and design-based approaches}

The model-based and design-based estimators of the GMSE are obtained by 
isolating each source of uncertainty separately, that is, by considering 
the expectation with respect to either the model or the sampling design 
alone. These are introduced here 
as complementary quantities, since they are employed in the application of 
Section 4 as reference estimators. However, neither estimator alone captures the full variability captured by the GMSE.

\paragraph{Design-based approach}

The design-based estimator isolates the variability due to the sampling design, treating the outcome variable $Y$ as fixed and replicating only the sampling mechanism. Specifically, for $m = 1, \dots, M$ independent replicates of the sample membership indicator $\lambda^{(m)}$, a multinomial model is refitted on each sample, yielding estimates 
$\hat{\theta}_k^{(d,m)}$, for all $k$'s. The design-based GMSE estimator is then defined as:
\begin{align}
    \widehat{\text{GMSE}}^{\text{Design}}\left(\hat{\theta}_k^{(d)}, 
    \theta_k^{(d)}\right) = \frac{1}{M}\sum_{m=1}^{M} 
    \left(\hat{\theta}_k^{(d,m)} - \hat{\theta}_k^{(d)}\right)^2, 
    \quad k = 1, \dots, K,
\end{align}
where $\hat{\theta}_k^{(d)} = \sum_{i=1}^N \gamma_i^{(d)} \hat{p}_{ik}$ 
is the baseline estimate obtained from the full (observed) sample, and 
$\hat{\theta}_k^{(d,m)} = \sum_{i=1}^N \gamma_i^{(d)} \hat{p}_{ik}^{(m)}$ 
is the estimate obtained by refitting the model on the $m$-th replicate 
sample. This quantity captures design variability only, as the model 
parameters are re-estimated under a new sampling realisation while the 
population outcome structure remains fixed.

\paragraph{Model-based approach}

The model-based estimator isolates the variability due to the outcome model, treating the sample membership $\lambda$ as fixed at its observed value and replicating only the model randomness. Specifically, for 
$m = 1, \dots, M$ independent replicates of the outcome variable, a new 
response $\tilde{Y}_i^{(m)}$ is generated as $\tilde{Y}_i^{(m)} \sim \text{Multinomial}(1, \hat{p}_i)$, with $\hat{p}_i$ the baseline estimate obtained from the (originally observed) data. A multinomial model is then refitted on the observed sample using 
$\tilde{Y}_i^{(m)}$ as the response, yielding updated probability estimates 
$\hat{p}_{ik}^{(m)}$ and corresponding total estimates 
$\hat{\theta}_k^{(d,m)} = \sum_{i=1}^N \gamma_i^{(d)} \hat{p}_{ik}^{(m)}$. 
The model-based GMSE estimator is then defined as:
\begin{align}
    \widehat{\text{GMSE}}^{\text{Model}}\left(\hat{\theta}_k^{(d)}, 
    \theta_k^{(d)}\right) = \frac{1}{M}\sum_{m=1}^{M} 
    \left(\hat{\theta}_k^{(d,m)} - \hat{\theta}_k^{(d)}\right)^2, 
    \quad k = 1, \dots, K,
\end{align}
where $\hat{\theta}_k^{(d)}$ is the same baseline estimate as in the design-based approach. This quantity captures model variability only, as the sampling design is held fixed and only the stochastic component of the outcome model is replicated.

\paragraph{Relationship to the GMSE}

Under a non-informative sampling design, as illustrated in the main technical proof of Supplementary material B, the GMSE in Definition~1 can be 
decomposed as:
\begin{align}
    \text{GMSE}\left(\hat{\theta}_k^{(d)}, \theta_k^{(d)}\right) 
    = \Exp_{\D} \Var_{\M}\left(\hat{\theta}_k^{(d)}|\bm{\lambda}\right) + \Var_{\D} \Exp_{\M}\left(\hat{\theta}_k^{(d)}|\bm{\lambda}\right) + \text{Minor order terms},
\end{align}
reflecting the two sources of uncertainty that the model-based and 
design-based estimators isolate separately. In particular:
\begin{itemize}
    \item $\widehat{\text{GMSE}}^{\text{Model}}$ approximates the inner 
    model variance $\text{Var}_Y(\hat{\theta}_k^{(d)} \mid \lambda)$ for a fixed sample, resulting in a lower (most often) or higher value depending on the specific sample;
    \item $\widehat{\text{GMSE}}^{\text{Design}}$ approximates the outer 
    design expectation $\mathbb{E}_\lambda$, by varying the sampling realisation while fixing the model, and, most importantly, reflects the variability of the model estimator across different samples. Since the potential of (even small) small bias is non-zero (especially in small samples), this variability may contribute substantially to the overall uncertainty, beyond model uncertainty.
\end{itemize}

\subsection{Monte Carlo approach (benchmark)}

The Monte Carlo (MC) estimator serves as a benchmark for the true GMSE, 
and is used in the simulation study of Section 5 (main manuscript) to validate the performance of the proposed linearised estimator. Unlike the model-based and design-based estimators, the MC  estimator jointly replicates both sources of uncertainty --- the sampling 
design and the outcome model --- by generating $G$ independent replicates 
of the sample membership indicator $\lambda^{(g)}$ and, for each, $M$ 
independent replicates of the response variable $Y^{(m,g)}$. Specifically, 
for each pair $(m, g)$, a multinomial model is fitted on the $g$-th sample 
replicate using the $m$-th response replicate, yielding probability estimates 
$\hat{p}_{ik}^{(m,g)}$ and corresponding total estimates 
$\hat{\theta}_k^{(d,m,g)} = \sum_{i=1}^N \gamma_i^{(d)} \hat{p}_{ik}^{(m,g)}$.

The corresponding GMSE and CV estimators are given, respectively, by:
\begin{align*}
\widehat{\GMSE}^{\text{MC}}\left(\hat{\theta}_k^{(d)}, \theta_k^{(d)}\right) &= \frac{1}{G}\sum_{g=1}^G \widehat{\MSE}_g^{\text{MC}}\left(\hat{\theta}_k^{(d)}, \theta_k^{(d)}\right)
    = \frac{1}{G} \sum_{g=1}^G \left(\frac{1}{M} \sum_{m=1}^M \left(\sum_{i=1}^N  \gamma_{i}^{(d)}\hat{p}_{ik}^{(m, g)} - \sum_{i=1}^N  \gamma_{i}^{(d)}p_{ik}\right)^2 \right),\\
    \widehat{\CV}^{\text{MC}}\left(\hat{\theta}_k^{(d)}, \theta_k^{(d)}\right) &= \frac{\sqrt{\widehat{\GMSE}^{\text{MC}}\left(\hat{\theta}_k^{(d)}, \theta_k^{(d)}\right)}}{\Exp\left(\theta_k^{(d)}\right)},\quad k=1,\dots, K.
\end{align*}
where $\hat{p}_{ik}^{(m, g)}$ is the probability estimate obtained in 
the MC run corresponding to the $m$-th replicate of the response model 
and the $g$-th replicate of the sample membership; $p_{ik}$ denotes 
the true category probabilities used in the data-generating process.

\section{Additional results} \label{app: simu}

\begin{table}[h]
\centering 
\caption{
Model coefficients $\bm{\beta}_k, k = 1,\dots, K=8$ for Subgroup B based on the sample survey data. Coefficients for $k=8$ are all set to zero since this is used as baseline category. 
} \label{tab: true_coeff}
  \begin{tabular}{lrrrrrrrr}
    \toprule
    & \multicolumn{8}{c}{Response Category $k$} \\ \cmidrule{2-9}
      & $k=1$  & $k=2$ & $k=3$  & $k=4$ & $k=5$  & $k=6$ & $k=7$  & $k=8$ \\
       Model covariate (dummy) &  & &   & &   & &   \\
      \midrule
$X_0$: Intercept &  8.897 & -0.543 & 6.021 & 8.730 & 8.018 & 0.028 & 0.574 & 0.000\\
$X_1$ :: \texttt{[29, 39]}  &  -1.344& 6.426&  -0.130&  -1.586& -2.059& -1.687& -2.862& 0.000\\
$X_1$ :: \texttt{[40, 49]}  & -0.894& 7.927&  0.577&  -0.852& -1.427& -1.638& -2.270&  0.000\\
$X_1$ :: \texttt{[50, 69]}  & 0.966&  8.214&  2.157&  -0.267& -1.068& -1.392& -1.825& 0.000\\
$X_1$ :: \texttt{[70, )}  & 2.640&   11.259& 4.524&  0.755&  -0.085& -1.151& -0.867& 0.000\\
$X_2$ :: \texttt{Female}  & 0.821&  0.947&  0.286&  -0.046& 0.079&  0.311&  0.233& 0.000\\
$X_3$ :: \texttt{Italian citizenship}  & -2.048& -1.250&  -0.411& 0.056&  0.032&  -0.291& 0.311&  0.000\\
$X_4$ :: $k=2$  & 0.251& 3.542& 3.475& 2.110&  2.391& 8.305& 1.772& 0.000\\
$X_4$ :: $k=3$  & -0.999& 1.997&  4.431&  2.827&  2.548&  7.951&  8.471&  0.000\\
$X_4$ :: $k=4$  & -6.377& -4.379& -2.595& -0.143& -0.492& 3.409&  3.863&  0.000\\
$X_4$ :: $k=5$  & 0.016&  3.170&   4.512&  5.732&  9.772&  13.889& 13.409& 0.000\\
$X_4$ :: $k=6$  & -16.979& -14.753& -7.478&  -5.416&  -2.491&  7.184&   5.581&   0.000\\
$X_4$ :: $k=7$  & -10.654& -9.324&  -7.827&  -6.885&  -4.805&  2.696&   5.800& 0.000\\
$X_4$ :: $k=8$  & -14.816& -14.945& -15.966& -10.870&  -9.100&    -8.896&  -0.020& 0.000\\
    \bottomrule
  \end{tabular}
\end{table}

\begin{table}[h]
\centering
\caption{Estimates of the GMSE for totals $\hat{\theta}^{(d)}_k = 
\sum_{i=1}^N \gamma_i^{(d)}\hat{Y}_{ik}, k=1,\dots,8$ for the full 
register and for domain $d \in X_2: \text{Gender}$. The sample fraction 
$n_k^{(d)}/\hat{\theta}^{(d)}_k$ is between 3.9\% and 5.2\% across all 
cases. Bootstrap, as well as model-based and design-based estimates are 
based on $1{,}000$ replicates.} \label{tab: gmse_est_gmse}
\begin{tabular}{lrrrrrr}
    \toprule
    Category $k$ & $\hat{\theta}^{(d)}_k$ & \makecell{Sample \\ size 
    $n_k^{(d)}$} & $\widehat{\text{GMSE}}^{\text{Lin}}_k$ & 
    $\widehat{\text{GMSE}}^{\text{Boot}}_k$ & 
    $\widehat{\text{GMSE}}^{\text{Model}}_k$ & 
    $\widehat{\text{GMSE}}^{\text{Design}}_k$\\
    \midrule
    \multicolumn{7}{l}{Full register: $\gamma_i = 1,\quad i=1,\dots,N$}\\
    \midrule
\texttt{1 Illiterate}                &1039   &49   &15195  &19318  &17060  &18038  \\
\texttt{2 Literate but no education} &6649   &340  &97462  &109015 &104099 &92364  \\
\texttt{3 Primary}                   &49886  &2572 &288343 &305130 &284640 &262546 \\
\texttt{4 Lower secondary}           &84174  &4285 &530144 &600977 &530370 &500403 \\
\texttt{5 Upper secondary}           &113719 &5682 &497936 &573528 &489321 &443761 \\
\texttt{6 Bachelor degree}           &7234   &364  &91337  &90701  &81660  &71952  \\
\texttt{7 Master degree}             &32810  &1524 &171777 &189906 &131736 &125123 \\
\texttt{8 PhD level}                 &1054   &44   &16074  &13679  &13518  &11258  \\
    \midrule
    \multicolumn{7}{l}{Internal domain: $\gamma^{(d)},\ d = 
    \texttt{Male}$ (47.7\%)}\\
    \midrule
\texttt{1 Illiterate}                &300   &14   &4435   &5420   &5401   &5442   \\
\texttt{2 Literate but no education} &1569  &81   &24444  &28082  &27619  &27116  \\
\texttt{3 Primary}                   &19631 &1015 &114420 &125573 &122862 &105258 \\
\texttt{4 Lower secondary}           &45853 &2306 &261562 &315529 &246235 &243389 \\
\texttt{5 Upper secondary}           &56374 &2775 &243509 &282340 &227522 &227757 \\
\texttt{6 Bachelor degree}           &2701  &132  &36359  &29585  &32880  &25936  \\
\texttt{7 Master degree}             &14443 &656  &76429  &79055  &64843  &56698  \\
\texttt{8 PhD level}                 &510   &20   &7813   &5975   &6705   &5061   \\
    \midrule
    \multicolumn{7}{l}{Internal domain: $\gamma^{(d)},\ d = 
    \texttt{Female}$ (52.3\%)}\\
    \midrule
\texttt{1 Illiterate}                &739   &35   &10721  &13267  &12044  &11993  \\
\texttt{2 Literate but no education} &5080  &259  &72826  &77146  &78734  &67396  \\
\texttt{3 Primary}                   &30255 &1557 &172843 &195492 &179959 &156370 \\
\texttt{4 Lower secondary}           &38321 &1979 &265811 &304650 &267393 &251891 \\
\texttt{5 Upper secondary}           &57345 &2907 &251083 &292115 &244189 &232986 \\
\texttt{6 Bachelor degree}           &4533  &232  &54303  &57122  &46166  &45562  \\
\texttt{7 Master degree}             &18367 &868  &93330  &100399 &79951  &74574  \\
\texttt{8 PhD level}                 &545   &24   &8069   &7120   &6436   &6067   \\
    \bottomrule
\end{tabular}
\end{table}

\begin{figure}[H]
    \centering
    \includegraphics[scale = .65]{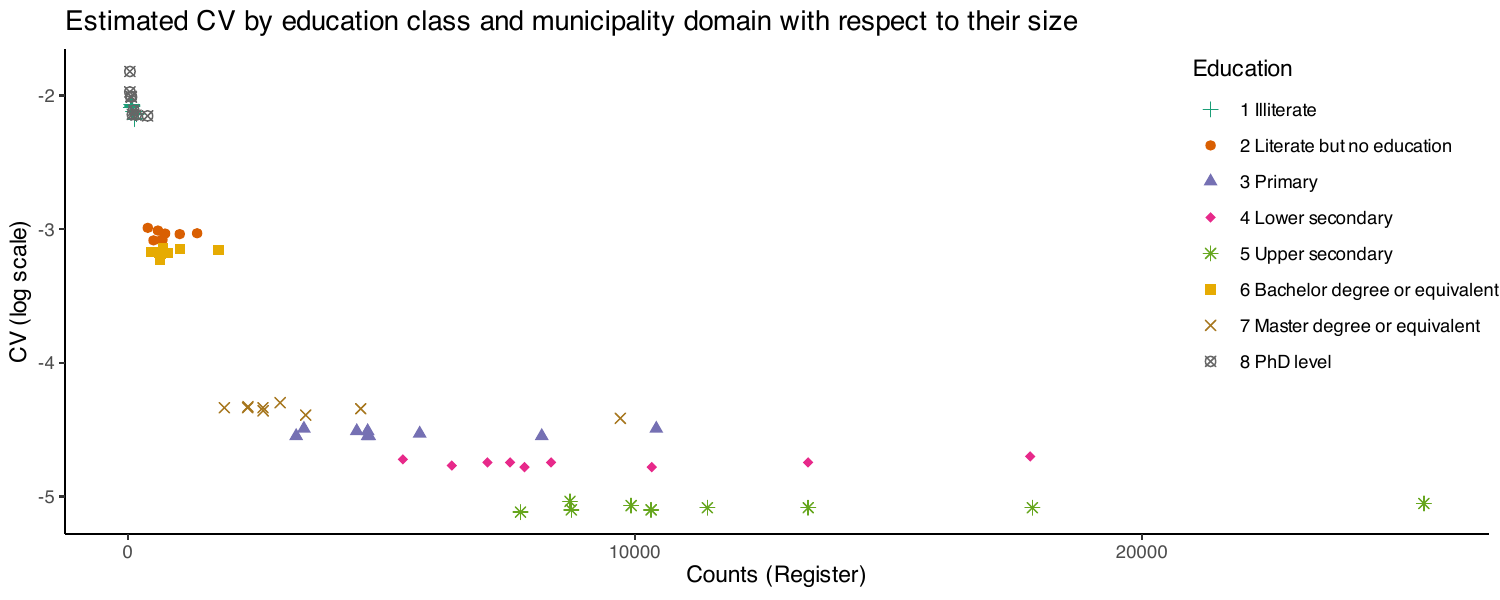}
    \caption{Coefficient of variation ($\widehat{\CV}^{\text{Lin}}_k$; in log scale) with respect to the register total estimates $\hat{\theta}^{(d)}_k$ for each class $k = 1,\dots, 8$ and domain combination, with $d \in X_5$: \texttt{province} having $9$ modalities. }
    \label{fig: CV_province}
\end{figure}

\begin{figure}[H]
    \centering
    \includegraphics[scale = .58]{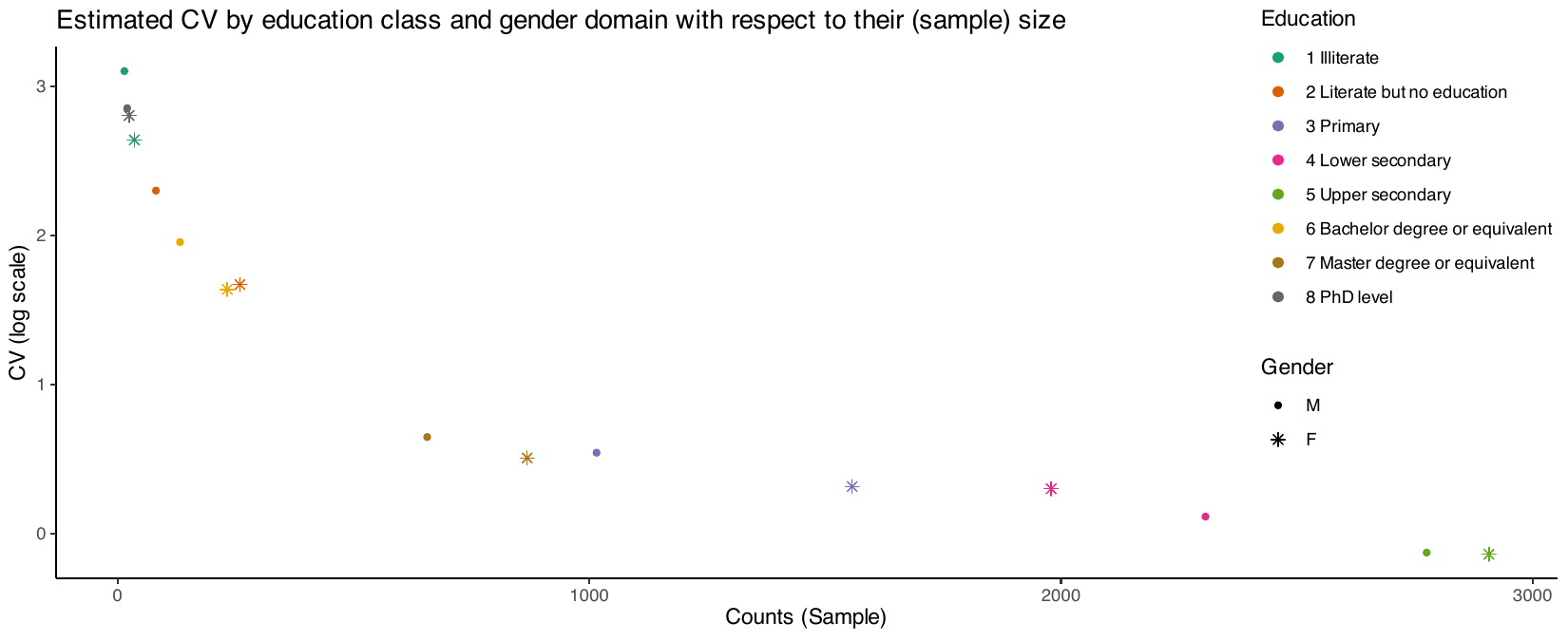}
    \caption{Coefficient of variation ($\widehat{\CV}^{\text{Lin}}_k$; in log scale) with respect to the sample size of each class $k = 1,\dots, 8$ and domain combination, with $d \in X_2$: \texttt{gender} having $2$ modalities.}
    \label{fig: CV_gender_sample}
\end{figure}

\begin{figure}[H]
    \centering
    \includegraphics[scale = .58]{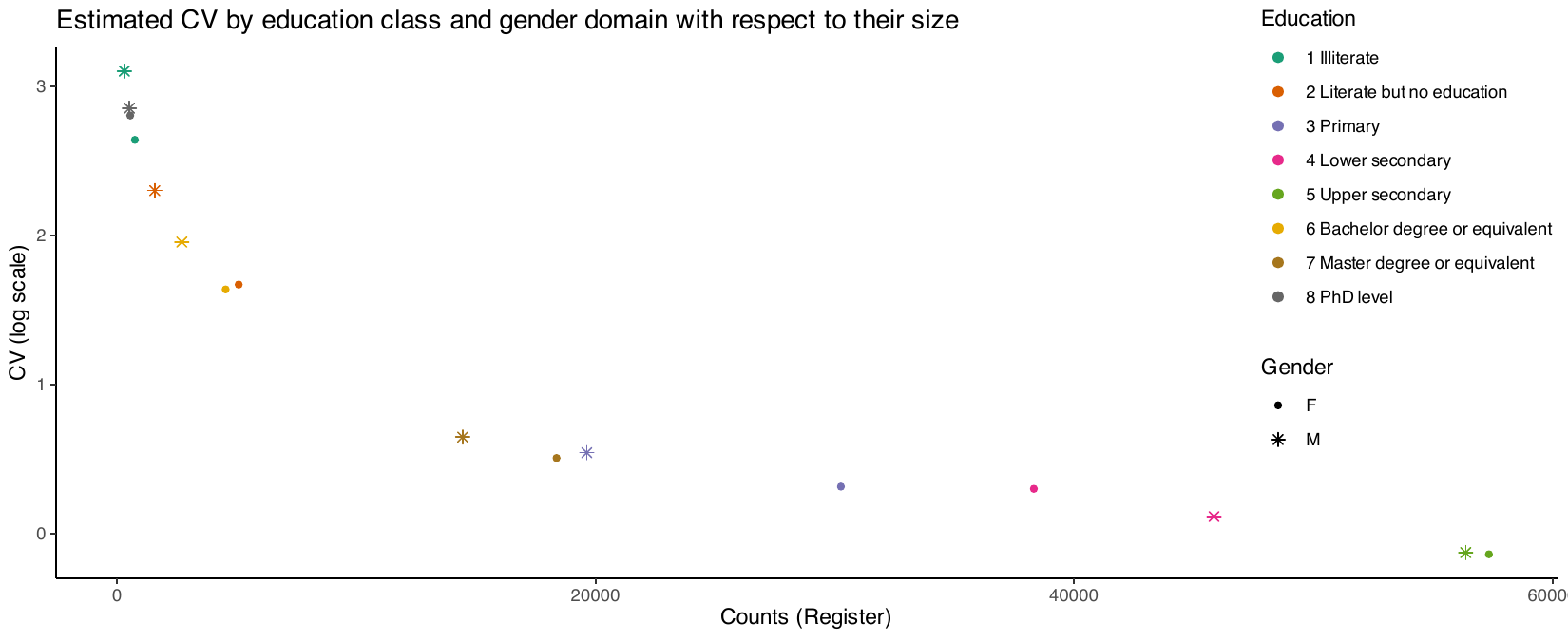}
    \caption{Coefficient of variation ($\widehat{\CV}^{\text{Lin}}_k$; in log scale) with respect to the register total estimates of each class $k = 1,\dots, 8$ and domain combination, with $d \in X_2$: \texttt{gender} having $2$ modalities.}
    \label{fig: CV_gender}
\end{figure}

\begin{figure}[h]
    \centering
    \includegraphics[width=.95\linewidth]{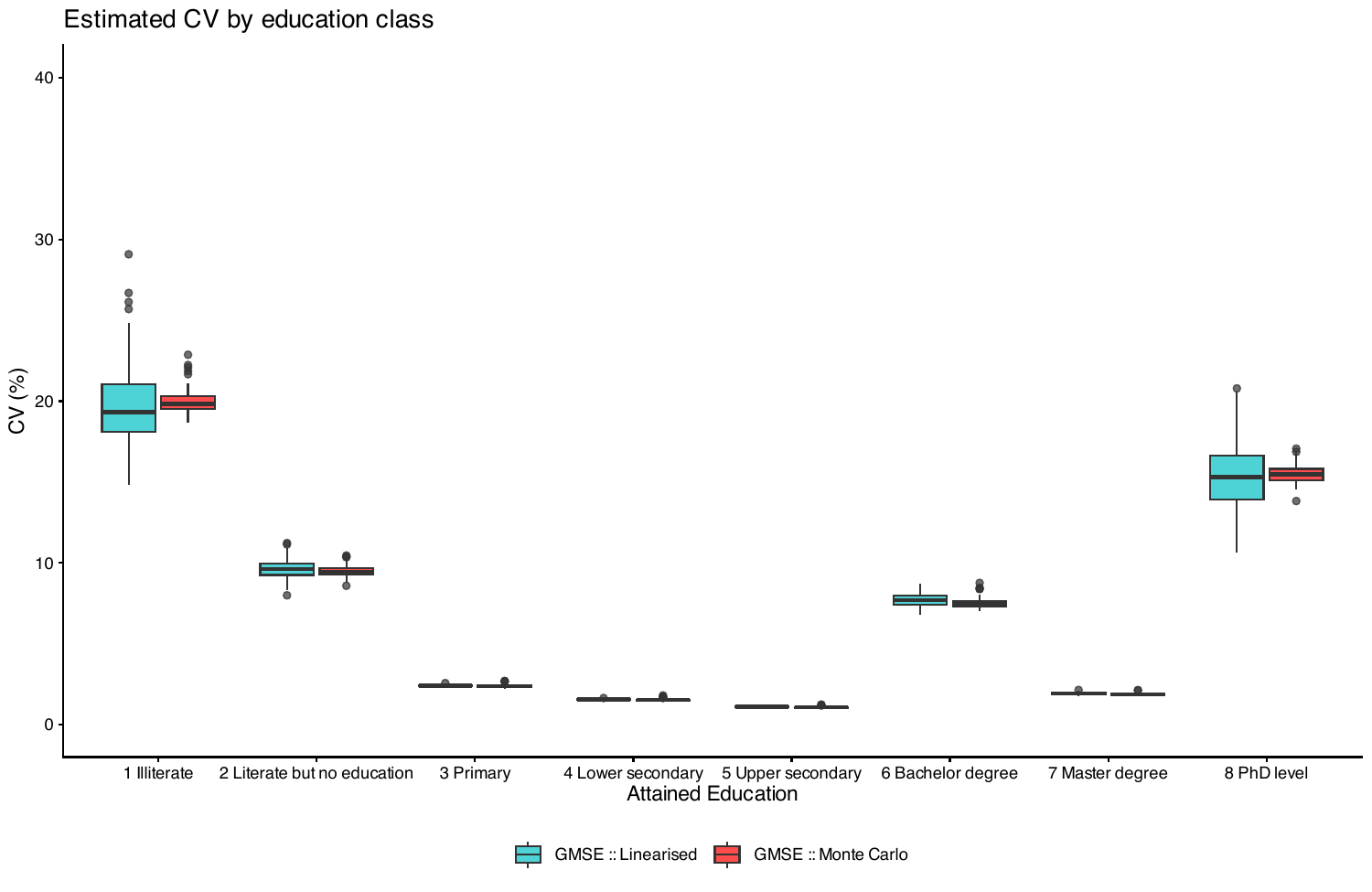}
    \caption{Estimates of the CV obtained with the linearised approach and compared to the MC benchmark. The considered scenario is $N  = 100,000$, with $n = 0.05 \times N = 5,000$. All results are based on $S=100$ replications of the two evaluation procedures.}
    \label{fig: cv_100k}
\end{figure}

\section{Kronecker Formulation} \label{app: kron}
\subsection{Why the need to introduce an alternative formulation \label{app: why_KF}}
In this section, we derive an alternative formulation of the GMSE estimator for the multinomial model. This is a more compact and completely matrix formulation of the GMSE estimator that, compared to the one derived in the article, hereinafter referred to as the standard formulation (SF), allows us to avoid the introduction of summations. Furthermore, the alternative formulation leads to the definition of a single formula that allows us to obtain the GMSE estimate simultaneously for all K estimation categories and for all D domains. This feature can further facilitate the application of the method by enhancing one of its main characteristics for which it was proposed, that is, that of producing on-the-fly estimates of the quality level of the estimates from the register. Furthermore, the characteristics described can be successful in the application of the studied method on large masses of data, such as statistical registers of National Statistical Insitutes that refer to population units. For this reason, in writing the R code that implements the GMSE, two different R programs were developed. The first implements the formulas of the SF, while the second implements those of the  alternative formulation. Furthermore, since the new formulation can produce very large matrices, for this formulation, a version based on a block sum was then derived. This allows to reduce the number of records of the matrices referred to each single block. Naturally, the two new versions coincide in the case of a single block. In the following pages we give only the alternative formulation in the case of a single block because its adaptation to the more general case of two or more blocks is straightforward.   

\subsection{Extended notation \label{app: notation_KF}}
In the following, we introduce a general notation based on block vectors and matrices, useful for deriving the GMSE estimator according to the new formulation. In particular, $\mathbf{*}=col_{b=1}^{B}\{\mathbf{*}_{b}\}=[\mathbf{*}_b]$ is a column vector (when $\mathbf{*}\equiv a$, for $a$ a generic scalar) or a column block vector (when $\mathbf{*}\equiv \mathbf{a}$) or a column block matrix (when $\mathbf{*}\equiv \mathbf{A}$), being $b$ ($b=1,\dots, B$) the index denoting the $b$-th element or block; for example, when $\mathbf{*}_b=\mathbf{A}_b$ the general notation becomes $\mathbf{A}=col_{b=1}^{B}\{\mathbf{A}_{b}\}=[\mathbf{A}_b]$. In most cases, we will use the last notation after the equals sign, which is more synthetic and compact. Similarly we indicate with $\mathbf{*}=row_{b=1}^{B}\{\mathbf{*}_{b}\}$ a row vector or row block vector or block matrix ($b=1,\dots B$). Furthermore $diag_{b=1}^{B}\{\mathbf{*}_{b}\}$ is a diagonal matrix made up of $B$ elements or blocks $\mathbf{*}_b$ ($\equiv a_b, \mathbf{a}_b, \mathbf{A}_b$); $mat_{b,b'=1}^{B, B'}\{\mathbf{*}_{bb'}\}$ is matrix or block matrix made up of $B\times B'$ elements or blocks $\mathbf{*}_{bb'}$ ($\equiv a_{bb'}, \mathbf{a}_{bb'}, \mathbf{A}_{bb'})$; $\mathbf{\Delta}_{\mathbf{a}}=diag_{b=1}^{B}\{a_{b}\}$ is a diagonal matrix made up of the elements of a vector, $\mathbf{a}=[a_{b}]$ ($b=1,\dots, B$). 
It is also useful to introduce some indicator vectors and matrices:   $\mathbf{1}_{c}$, $\mathbf{0}_{c}$ and $\mathbf{0}_{c; 1c'}$ are vectors of order $c$: the first is composed of all values equal to "1", the second of all values equal to "0" while the third is equal to the previous one except for the element $c'$ which is set equal to "1"; $\mathbf{I}_{a}$, is the identity matrix of order $a$; $\mathbf{I}_{a;b}=\mathbf{1}_{b}\otimes \mathbf{I}_{c} $ is a $[(b \times c) \times c]$-dimensional matrix composed by $b$ column blocks each one equal to $\mathbf{I}_{a}$. Finally for two general matrices $\mathbf{A}$ of order $(r \times c)$ and matrix $\mathbf{B}$ of order $(s \times d)$, let's consider the following multiplication operators: $\mathbf{A}\#\mathbf{B}$, which applies under condition $r=s$ and
denotes the matrix of order matrix [$r \times  (c\times s)]$ obtained by the \textit{element-by-element product}, $\mathbf{a}_{.j}\#\mathbf{B}$, of each column vector $\mathbf{a}_{.j}$, for $j=1,\dots, c$, of matrix $\mathbf{A}$ by matrix $\mathbf{B}$. Let's note that in case in which we have a vector $\mathbf{b}$ of order $r$, instead of matrix $\mathbf{B}$, the generic column vector of $\mathbf{A}\#\mathbf{b}$ is $\mathbf{a}_{.j}\#\mathbf{b}$;  $\mathbf{A}\otimes\mathbf{B}$, is the matrix of order $[(r\times s) \times (c \times d)]$ obtained by the \textit{direct product} of $\mathbf{A}$ and $\mathbf{B}$ denoting the product $a_{ij}\mathbf{B}$ $(i=1, \dots, r; j=1, \dots, c)$ of each element of $\mathbf{A}$ for matrix $\mathbf{B}$. The matrix $\mathbf{A}\otimes\mathbf{B}$ goes by the name of Kronecker product because of Kronecker's association with the determinant of $\mathbf{A}\otimes\mathbf{B}$, although in this regard~\citep{henderson1983history} suggest that ``Zehfluss product'' would be more appropriate. 

In this regard, it is noted that since the alternative formulation is based on the application of the Kronecker product to the vectors and matrices already introduced for the SF, the alternative formulation of the GMSE given below will be referred to in the following as the Kronecker Formulation (KF) of the GMSE. 

\subsection{Parameter and estimator  \label{app: KF_parameter}}

Under KF, the target parameter is a $(D\times K)$-dimensional vector formally defined as 
\begin{align} \label{eq: theta_kr}
     \mathbf{\theta} =[\mathbf{\theta}^{(d)}]= \mathbf{\Gamma}^{T}\mathbf{y},
\end{align}
whose $d$-th block is the $K$-dimensional vector 
\begin{align} \label{eq: theta_kr_d}
     \mathbf{\theta}^{(d)}=[\theta^{(d)}_{k}] =\mathbf{\Gamma}^{(d),T}\mathbf{y},
\end{align}
in which $K$-th element is
\begin{align} \label{eq: theta_kr_d_k}
\theta^{(d)}_{k}=\mathbf{\gamma}^{(d),T}_{k}\mathbf{y},
\end{align}
being Eq. (1) the correspondent expression under SF. In the above formulas: $\mathbf{\Gamma} = \left(\mathbf{\Gamma}^{(1)}, \dots, \mathbf{\Gamma}^{(d)}, \dots,\mathbf{\Gamma}^{(D)} \right)$ is a matrix of order $[T \times (D \times K)]$, for $T=N \times K$,  in which $\mathbf{\Gamma}^{(d)}=[\mathbf{\gamma}^{(d)}\otimes \mathbf{1}_K]\cdot \mathbf{I}_{K;N}$ is the $d$-th block matrix of order $[T \times K]$ being $\mathbf{\gamma}^{(d)}_{k}$ ($k=1, \dots, K)$ the $k$-th column vector of $\mathbf{\Gamma}^{(d)}$; $\mathbf{y}=[\mathbf{Y}_i]$ is a $T$-dimensional vector whose $i$-th, for $i=1,\dots,N$, block is the vector $\mathbf{Y}_i=[Y_{ik}]$ of the $K$ response categories $Y_{ik} \in \{0,1\}$ such that $\sum_{k=1}^K Y_{ik} =1$ with $\boldsymbol{Y}_i \sim Mult(\boldsymbol{1}_K, \mathbf{p}_{i})$.

The natural estimator $\hat{\mathbf{\theta}} =[\mathbf{\theta}^{(d)}]$ for $\hat{\mathbf{\theta}}^{(d)}=[\theta^{(d)}_{k}]$ is 
\begin{align} \label{eq: theta_kr_est}
     \hat{\mathbf{\theta}} =[\hat{\mathbf{\theta}}^{(d)}]=\mathbf{\Gamma}^{T}\hat{\mathbf{y}}.
\end{align}
in which $\hat{\mathbf{y}}=[\hat{\mathbf{Y}}_{i}]$ is the predictor of $\mathbf{y}=[\mathbf{Y}_{i}]$  for $\hat{\mathbf{Y}}_{i}=[\hat{Y}_{ik}]$.

We derive first the KF of $\hat{Y}_{ik}$, for $i=1,\dots, N; k=1,\dots, K$. To this aim let's consider the general expression of working model already introduced for SF and through Kronecker product we pass for the $i$-th unit from the $J$-dimensional vector $\mathbf{x}_i$ to $[(J \times K) \times K]$-dimensional matrix 
$\dot{\mathbf{X}}_i=\mathbf{x}_i\otimes\mathbf{I}_K$ 
. Note that $\dot{\mathbf{X}}_i$ has been already introduced under SF before formula Eq. (16) of the main paper. Under KF we get a reordering  by $k$ and $j$, for $k=1,\dots, K$ and $j=1,\dots, J$, of the rows of the original one. Let's introduce now  the working model $\mathbf{E}(Y_{ik}|\dot{\mathbf{x}}_i) = f(\dot{\mathbf{x}}_i;\mathbf{\beta}) = p_{ik}$, $i=1,\dots, N; k=1, \dots, K$, with $f$ a known function depending on the $H=(J\times K)$-dimensional unknown parameter vector $\mathbf{\beta} = \left(\mathbf{\beta}_1^T,\dots,\mathbf{\beta}_j^T,\dots, \mathbf{\beta}_J^T\right)^T$, for $\mathbf{\beta}_j = (\beta_{1j}, \dots, \beta_{Kj})^T$, for $k=1,\dots, K$ and $j=1,\dots, J$, measuring the relationship between the outcome $Y_{ik}$ and the set of covariates $\dot{\mathbf{x}}_i \in \mathbf{R}^H$. Note that in our specific case of Multinomial model, $p_{ik} = \mathbb{P}(Y_{ik} = 1|\dot{\mathbf{x}}_i)$, with $\sum_{k=1}^Kp_{ik} = 1, \forall i$. 
Then, the predictor $\hat{Y}_{ik}$ of $Y_{ik}$ is 
\begin{align}\label{eq: Y_hat_ik}
    \hat{Y}_{ik} = \hat{p}_{ik} = f(\dot{\mathbf{x}}_i;\mathbf{\hat{\beta}}), \quad i=1,\dots, N, \quad k=1,\dots,K,
\end{align}
and the predictor $\hat{\mathbf{y}}=[\hat{\mathbf{Y}}_{i}]$ of $\mathbf{y}=[\mathbf{Y}_{i}]$, for $\hat{\mathbf{Y}}_{i}=[\hat{Y}_{ik}]$, is
\begin{align}\label{eq: Y_hat_Kron}
    \hat{\mathbf{y}}= \hat{\mathbf{p}} = f(\dot{\mathbf{X}};\mathbf{\hat{\beta}}),
\end{align}
in which $\mathbf{\hat{\beta}}$ and $\hat{\mathbf{p}}=(\hat{\mathbf{p}}_1^{T}, \dots, \hat{\mathbf{p}}_{i}^{T}, \dots, \hat{\mathbf{p}}_{N}^{T})^{T}$, for $\hat{\mathbf{p}}_{i}=[\hat{p}_{ik}]$ and $\mathbf{p}_{i}=[p_{ik}]$, are the correspondent consistent estimators of $\mathbf{\beta}$ and  $\mathbf{p}=(\mathbf{p}_1^{T}, \dots, \mathbf{p}_{i}^{T}, \dots,\mathbf{p}_{N}^{T})^{T}$ respectively based on the data of a generic sample $s$ of fixed size $n$. In particular $\mathbf{\lambda} = (\lambda_{1}, \dots, \lambda_{i}, \dots, \lambda_{N})^{T}$ for $\sum_{i=1}^{N}\lambda_{i}=n$
denotes the random variable vector defining the sample membership satisfying $E(\mathbf{\lambda}) = \mathbf{\pi}= [\pi_{i}]$ for $\pi_{i}\in (0, 1]$ the inclusion
probability of unit $i$, for $i = 1,\dots, N$, and $\dot{\mathbf{X}}=[\dot{\mathbf{X}}_{i}^{T}]$ being the overall covariate matrix of order $(T \times H)$. In order to derive KF let's introduce too the following extra notation $\dot{\mathbf{\lambda}}=\mathbf{\lambda}\otimes \mathbf{1}_{K}$ and $\dot{\mathbf{\pi}}=\mathbf{\pi}\otimes \mathbf{1}_{K}$. We note that the $k$-th ($k=1, \dots, K)$) row of $\dot{\mathbf{X}}^{T}_i=[\dot{\mathbf{x}}_{i}^{T}]$ is a $H$-dimensional vector of the form  
$\dot{\mathbf{x}}^{T}_{ik}=(x_{i1}\cdot \mathbf{0}^{T}_{K; 1k},\cdots, x_{ij}\cdot \mathbf{0}^{T}_{K; 1k},\cdots, x_{iJ}\cdot \mathbf{0}^{T}_{K; 1k})$.

\subsection{Multinomial logistic estimator} 
To give the explicit expression of $\hat{\mathbf{p}}$, of the unknown parameter vector $\mathbf{p}$, under multinomial logistic regression model Eq. (13), let's introduce the following extra notation: $T$-dimensional vector $\hat{\mathbf{e}}=\exp (\dot{\mathbf{X}}\hat{\mathbf{\beta}})$; the correspondent $(N \times K)$-dimensional matrix $\hat{\mathbf{E}}=[\hat{\mathbf{e}}_{i}^{T}]$  whose $i$-th row, for $i=1,\dots, N$, is given by the $K$ dimensional vector $\hat{\mathbf{e}}^{T}_{i}$; $\hat{\mathbf{e}}_{+}=\hat{\mathbf{E}}\cdot{\mathbf{1}_{K}}$; $\hat{\mathbf{d}}_{+}=\mathbf{1}_{T} +\hat{\mathbf{e}}_{+} \otimes \mathbf{1}_{K}$. The KF estimator of $\mathbf{p}$ is given by
\begin{equation}
\label{p_hat}
\hat{\mathbf{p}}=\frac{\hat{\mathbf{e}}}{\hat{\mathbf{d}}_{+}}=[\frac{\hat{\mathbf{e}}_{i}}{\hat{\mathbf{d}}_{+i}}]=[\hat{\mathbf{p}}_{i}].
\end{equation} 
The generic element of vector $\hat{\mathbf{e}}=[e_{ik}]$ referred to the baseline category $k=K$ is given by $\hat{e}_{iK}=\exp (\dot{\mathbf{x}}_{iK}^{T}\hat{\mathbf{\beta}})=1$ being $\hat{\mathbf{\beta}}=\mathbf{0}_{H}$. 

The estimator in Eq.~\eqref{p_hat} is function of $\hat{\mathbf{\beta}}$ and the MLE estimator of $\mathbf{\beta}$ is solution of the set of $H$ estimating equations  
$\mathbf{l}_{\mathbf{\beta}}=\ell(\mathbf{\beta})/\partial \mathbf{\beta}=\mathbf{0}_{H}$ in which $\ell(\mathbf{\beta})$ is the log-likelihood of the working model. The KF of the estimating equations is  
\begin{align} \label{KF_g}
    \mathbf{g}=\mathbf{G}^{T}\dot{\mathbf{\lambda}} = \mathbf{0}_{H}.
\end{align}
in which $\mathbf{G}$ is a matrix of order $(T\times H)$  depending on $\dot{\mathbf{X}}$ given by  
\begin{align} \label{KF_G}
 \mathbf{G}=\dot{\mathbf{X}}\#[\mathbf{y}-\mathbf{p}].
\end{align}
Using the result in \cite{chambers2012introduction}, we can employ the first-order approximation in Eq. (10) of the main manuscript to linearize $\mathbf{\hat{\beta}}$ around its expected value $\Exp_{\M}(\hat{\mathbf{\beta}} | \mathbf{\lambda}) = \mathbf{\beta}$
\begin{align} \label{eq: KF_chambers}
\left(\hat{\mathbf{\beta}} - \mathbf{\beta} \right) \approx -\mathbf{A}_{\mathbf{\beta}}^{-1} \mathbf{G}^{T}\dot{\mathbf{\lambda}}. 
\end{align}
The explicit expression of $H$-dimensional square matrix $\mathbf{A}_{\mathbf{\beta}}$ is 
\begin{align} \label{KF_A}
    \mathbf{A}_{\mathbf{\beta}}=\frac{\partial \mathbf{g}}{\partial \mathbf{\beta}}=[\dot{\mathbf{X}}^{T}\#
\mathbf{\Delta}]\dot{\mathbf{\pi}},
\end{align}
in which $\mathbf{\Delta}$ is a matrix of order $(T \times H)$ given by
\begin{align} \label{KF_Delta}    
\mathbf{\Delta}=\frac{\partial \mathbf{p}}{\partial \mathbf{\beta}}=-[\dot{\mathbf{X}}+\dot{\mathbf{X}}_{+}\#\mathbf{P}_{J+}]\#\mathbf{p}
\end{align}
where $\dot{\mathbf{X}}_{+}=\dot{\mathbf{X}}\otimes \mathbf{1}_{K}^{T}\otimes \mathbf{1}_{K}$ and $\mathbf{P}_{J+}=\mathbf{P}\otimes \mathbf{1}_{K} \otimes {1}_{J}$ in which $\mathbf{P}=[\mathbf{p}_{i}^{T}]$.

Finally taking account of SF (Eq. (18) in the main manuscript) of the  linearised $GMSE$ estimator, we obtain the correspondent KF for $\hat{\theta}$. This is given by the diagonal elements of 
\begin{align} \label{eq: KF_GMSE_Lin}
GMSE^{\text{Lin}}\left(\hat{\theta}\right) = \mathbf{\Gamma}^{T} \left(\dot{\mathbf{\pi}}\# \dot{\bar{\mathbf{U}}} \dot{\mathbf{\Sigma}}_{\mathbf{Y}} \dot{\bar{\mathbf{U}}}\right)\mathbf{\Gamma}
\end{align}
in which:
\begin{align} \label{KF_U}    
\dot{\bar{\mathbf{U}}}=-[\dot{\mathbf{X}}\#\dot{\mathbf{\pi}}]\mathbf{A}_{\mathbf{\beta}}^{-1},
\end{align}
\begin{align} \label{KF_Sigma_y}    
\dot{\mathbf{\Sigma}}_{\mathbf{Y}}=\mathbf{1}_{T}\otimes \mathbf{\Sigma}_{\mathbf{Y}}
\end{align}
where 
\begin{align} \label{Sigma_y}    
\mathbf{\Sigma}_{\mathbf{Y}}=\mathbf{I}_{K;N}-\mathbf{p}\# [\mathbf{P} \otimes \mathbf{1}_{J}].
\end{align}

\end{document}